\documentclass[12pt,leqno,titlepage]{amsart}
\usepackage[foot]{amsaddr}
\usepackage{amsmath}
\usepackage{graphicx,psfrag,epsf}
\usepackage{enumerate}
\usepackage{natbib}
\usepackage{bibunits}
\defaultbibliographystyle{agsm}
\defaultbibliography{references}
\usepackage{url} 
\usepackage{mathtools, empheq}
\usepackage{amssymb, setspace}
\usepackage{bm, bbm}
\usepackage{graphicx, tabularx, array}
\usepackage[svgnames,table]{xcolor}
\usepackage{enumitem}
\usepackage{multirow}
\usepackage{hyperref}
\usepackage[bottom=.75in]{geometry}
\usepackage{makecell}
\usepackage{booktabs}
\usepackage{amsmath}
\usepackage{array}
\usepackage{algorithm}
\usepackage{algorithmic}
\usepackage{longtable}
\usepackage{colortbl}
\usepackage{ulem}
\usepackage{pdfpages} %
\usepackage{tikz}
\usepackage{pifont}
\DeclareRobustCommand{\ultraboldcheck}{%
  \ooalign{%
    \hfil\ding{51}\hfil\cr
    \hfil\raisebox{0.3pt}{\ding{51}}\hfil\cr
    \hfil\raisebox{-0.3pt}{\ding{51}}\hfil\cr
  }%
}

\usepackage[table]{xcolor}
\definecolor{myblue}{HTML}{8CBAD8}

\newcolumntype{C}[1]{>{\centering\arraybackslash}m{#1}}

\newcommand{\G}{\mathcal{G}}

\newcommand{\DeltaP}{\Delta_+}
\newcommand{\Ical}{\mathcal I}
\newcommand{\Mcal}{\mathcal M}

\renewcommand{\P}{\mathbb{P}}

\newcommand{\1}{\mathbbm{1}}

\newcommand{\p}{\text{plan}}

\renewcommand{\a}{\text{analysis}}

\newtheorem{theorem}{Theorem}[]
\newtheorem{corollary}{Corollary}[]

\newtheorem{lemma}{Lemma}

\newtheorem{proposition}{Proposition}

\usepackage{amssymb,amsthm,amsmath}
\usepackage{etoolbox}
\AtBeginDocument{%
    \let\orignewpage\newpage 
    \renewcommand\newpage{}
    \patchcmd{\clearpage}{\newpage}{\orignewpage}{}{}}

\makeatletter
\def\@settitle{\begin{center}\normalfont\Large\bfseries \@title\end{center}}
\makeatother
\makeatletter
\def\@setauthors{%
  \begingroup
    \let\MakeUppercase\relax
    \begin{center}
        \vspace{1em}
      \normalfont\normalsize\authors
    \end{center}%
  \endgroup
  \@setthanks
}
\makeatother

\makeatletter
\def\and{, }
\makeatother

\makeatletter
\def\@setkeywords{%
\\  \noindent \textit{Key words and phrases. }\@keywords%
}
\makeatother

\begin{document}
\def\spacingset#1{\renewcommand{\baselinestretch}%
{#1}\small\normalsize} \spacingset{1}
\begin{bibunit}

\title[Sensitivity analysis for observational studies with many outcomes]{\fontsize{15.8}{19}\selectfont Powerful Multivariate Sensitivity Analysis via Sample Splitting in an Observational Study of the Effects of Poverty on Cardiovascular Disease Risk Factors}

\author{William Bekerman$^{1}$} \email{\ignorespaces bekerman@wharton.upenn.edu ({\normalfont corresponding author})}
\author{Anurag Mehta$^{2}$}
\email{anurag.mehta@emory.edu}
\author{Rebecca E. Hasson$^{3}$}
\email{hassonr@umich.edu}
\author{Leah E. Robinson$^{3}$}
\email{lerobin@umich.edu}
\author{\newline Dylan S. Small$^{1}$}
\email{dsmall@wharton.upenn.edu}
\author{Colin B. Fogarty$^{4}$}
\email{fogartyc@umich.edu}

\dedicatory{$^{1}$Department of Statistics and Data Science, University of Pennsylvania, Philadelphia, PA, USA \\
$^{2}$Division of Cardiology, Emory University School of Medicine, Atlanta, GA, USA
\\
$^{3}$School of Kinesiology, University of Michigan, Ann Arbor, MI, USA
\\
$^{4}$Department of Statistics, University of Michigan, Ann Arbor, MI, USA}

\begin{abstract}
 When assessing the causal effect of an exposure on two or more outcomes in an observational study, a linear combination of outcomes may lessen the sensitivity of a test of the global null hypothesis to potential unmeasured biases. While all linear combinations of scored outcomes can be considered using Scheff\'e projections or constrained variants thereof, finding the combination that minimizes sensitivity to unmeasured biases requires corrections for multiple testing, which can erode power, especially when many outcomes are of interest. To mitigate this issue, we propose splitting the sample into a planning sample to identify an optimal linear combination and an analysis sample to conduct inference. We provide a novel characterization of the set of linear combinations for which this approach is guaranteed to achieve the same asymptotic power as full-sample alternatives and conduct extensive simulation studies that demonstrate enhanced power in finite samples. Finally, we apply our method to investigate the effects of poverty on the emergence of cardiovascular disease risk factors in children and adolescents. We discover adverse consequences on outcomes related to body composition, physical activity, and tobacco exposure. 
 Although the impact of poverty on elevated tobacco exposure shows some robustness to unmeasured confounding, the other findings remain sensitive to potential biases.
\end{abstract}

\keywords{Constrained statistical inference; Coherence; Exploratory data analysis; Multiple hypothesis testing; Sensitivity analysis\\\\
Bekerman was partially supported by the National Science Foundation Graduate Research Fellowship. Fogarty was partially funded by the National Science Foundation (NSF DMS-2413484). Mehta has received institutional research grants from Novartis and Amgen. The authors thank Jinho Bok for helpful suggestions. }

\maketitle
\spacingset{1.8}

\section{Socioeconomic Status and Cardiovascular Disease}

\subsection{How does poverty affect risk factors for cardiovascular disease?}

Cardiovascular disease (CVD) is a leading cause of mortality in the United States and around the world. Among the many factors that influence cardiovascular outcomes, socioeconomic status is a prominent, yet complex, underlying determinant of CVD risk that interacts closely with both behavioral and environmental risk factors. For instance, throughout the United States, children and adolescents from low-income families are less involved in sports and other fitness-oriented activities compared to their more affluent peers.  A study by the Centers for Disease Control and Prevention found that roughly 70 percent of children from families with incomes at least four times the poverty line took part in organized sports, compared to just around 31 percent among families at or below the poverty line \citep{black2022organized}. Similarly, 
\cite{rentschler2023global} documented that those in lower-income communities are particularly impacted by outdoor air pollution and disproportionately subjected to its downstream health complications. Socioeconomic disadvantage is well-established as a contributing factor to adverse health outcomes, yet growing up in or near poverty is relatively common in the United States compared to similarly affluent nations \citep{council2016poverty, shrider2023poverty}. Clarifying the role of poverty in the emergence of CVD risk factors in children and adolescents, therefore, offers significant potential for the development of new policies and interventions to mitigate negative health outcomes. 

Previous works have sought to better understand the mechanisms between growing up in poverty and the worsening risk of CVD. Typically, these studies have focused on evaluating how family income, a validated measure of income disparity that can be standardized against the federal poverty level to calculate family poverty income ratio (FPIR), affects intermediate, measurable risk factors like obesity, high cholesterol, and high blood pressure (BP) that indicate worsened risk of heart attack, stroke, and other complications. In particular, \cite{ali2011household}, \cite{jackson2018income}, and \cite{connolly2022social} used data from the National Health and Nutrition Examination Survey (NHANES) to study associations between FPIR and the emergence of CVD risk factors in children and adolescents.

\subsection{A new study design to reduce sensitivity to unobserved bias}

To investigate these impacts, we gather pre-pandemic NHANES waves from 1999 to 2016. We define exposure as growing up below the poverty line (FPIR $<$ 1) and compare ``treated'' children and adolescents to individuals with less socioeconomic disadvantage (FPIR $>$ 1). Throughout, we use the term ``poverty'' as shorthand to denote FPIR $<$ 1. Pair matching is used to adjust for measured covariates deemed plausible to affect poverty or intermediate CVD risk factors. We match exactly on the wave of data collection, participant gender, race/ethnicity, age group (8-11 years old; 12-17 years old), and health insurance status to examine possible stratum-specific effects, and match closely on age and the six-month time period when the examination was performed. Seeking to assess whether or not growing up in poverty causes a worsening risk of CVD, we collect outcomes related to nutrition, tobacco exposure, lipids, hypertension, body composition, kidney function, diabetes, and physical activity. In the full sample, we have a total of $4027$ matched pairs. These are divided into $1528$ pairs between the ages of 8 and 11 and $2499$ pairs between the ages of 12 and 17. In the younger age group, we have $736$ pairs of girls and $792$ pairs of boys, whereas the older group is divided into $1230$ pairs of girls and $1269$ pairs of boys. 

Related works provide descriptive statistics and discuss empirical patterns observed in their data, but few conduct formal statistical inference. Those analyses often use regression models that control for relevant demographic variables to test associations between low FPIR and each risk factor separately, reporting any $p$-values deemed significant. It is important to recognize that these are observational studies, so there is no reason to assume that children and adolescents from low-income families would be comparable to those from more affluent households in terms of covariates that were not measured. In practice, potential confounding variables such as parental health, inherited cardio-metabolic risk, family instability, neighborhood safety, and pubertal timing may be unmeasured or poorly measured. To alleviate concerns, it is typical to perform a sensitivity analysis to evaluate how much bias would need to be present from these covariates to alter the qualitative conclusions of an analysis that assumes we have measured all confounding variables. Examples of methods and illustrations of sensitivity analyses in observational studies include 
\cite{cornfield1959smoking} and
\cite{tan2006distributional}. 
The primary goal of our article is to clarify: Does growing up in poverty have a harmful causal effect on the emergence of CVD risk factors in children and adolescents? If so, how robust are our conclusions to potential unmeasured biases?

Relatively few strategies have been introduced to test this type of hypothesis, known alternatively as the global, intersection, or overall null hypothesis, in a sensitivity analysis framework. \cite{fogarty2016sensitivity} use the maximum of individual test statistics through a quadratically constrained linear program. \cite{rosenbaum1997signed} considers a linear combination of signed-rank statistics to create a coherent statistic, while \cite{rosenbaum2016using} extends this idea to all linear combinations using Scheff\'e projections (\citeyear{scheffe1953method}). Despite the advantages of these approaches in terms of potential gains in efficiency and insensitivity to unmeasured biases, finding the appropriate linear combination requires corrections for multiple testing that can erode power. \cite{cohen2020multivariate} look to diminish the stringency of this correction by restricting the combination vector to lie in the nonnegative orthant. In particular, this framework seems well-suited to our objective --- we worry that our exposure would affect each outcome in a pre-specified direction; that is, that poverty would induce a harmful effect on these intermediate CVD risk factors. However, the empirical performance of strategies that use Scheff\'e projections or constrained variants degrades with the number of outcomes in the global null hypothesis.
These methods utilize reference distributions whose quantiles depend upon the dimension of the linear combination to conduct valid inference. This forces practitioners to compare their deviate to a critical value that scales with the number of outcomes in the global null hypothesis, rendering it exceedingly challenging to reject the overall null hypothesis even when there is strong evidence against it.
Our observational study 
considers more outcomes than prior multivariate sensitivity analyses have handled --- can we still conduct powerful inference under concerns of unmeasured confounding?

We introduce a new sensitivity analysis framework for testing the global null hypothesis, in which we randomly split the data into a smaller pilot sample for planning the study and a larger analysis sample for making inferences. We estimate on the pilot sample an optimizing linear combination from the probability simplex and then treat this combination as pre-specified when conducting inference on the analysis sample using the Gaussian reference distribution. In principle, this approach could yield more powerful sensitivity analysis because the critical value of the Gaussian distribution does not depend on the number of outcomes included in the global null hypothesis. In defiance of traditional statistical wisdom \citep{cox1975note,wasserman2006weighted}, certain strategies that use data splitting have been shown to improve power in finite samples without suffering asymptotically in observational studies \citep{heller2009split,bekerman2026planning}. However, sample splitting does not uniformly improve performance in a sensitivity analysis in finite samples, nor does it uniformly do no harm asymptotically. As described in Section \ref{subsec:twoplayergame}, there exist techniques leveraging sample splitting that perform worse both in finite samples and asymptotically than the analogous strategy applied to the entire data. 
Currently, there is no theoretical guidance characterizing which sample splitting strategies preserve large-sample power, leaving unclear whether our approach sacrifices performance. 
To address this, we prove a novel minimax theorem which, as a corollary, guarantees that our procedure does not lose power when taking linear combinations from the probability simplex. As an additional benefit of sample splitting, exploratory data analysis can be performed on the pilot sample to incorporate data-driven insights into the study design and to refine the analysis plan \citep{rosenbaum2010design, small2024protocols}.

We randomly assign $1/4$ of the matched pairs in our study to the pilot sample and reserve the remaining pairs for formal analysis. We include in Table~\ref{tab:edatbl_summary_outcomes} summary statistics for outcomes of interest divided by FPIR stratum on the pilot sample.
Notably, we identify striking increases in cotinine levels, the primary biomarker used to assess recent tobacco use and secondhand smoke exposure, among children and adolescents growing up in poverty compared to those who do not. On the other hand, we observe either minimal discrepancies or differences that contradict our causal theory of poverty worsening each CVD risk factor, for non-HDL cholesterol levels, estimated glomerular filtration rate (eGFR), and glycated hemoglobin (HbA1c).

\begin{table}[!t]
\centering
\caption{Pilot sample summary statistics for selected outcomes by FPIR stratum, age group, and sex. Representative outcomes are chosen to assess tobacco exposure, lipids, hypertension, body composition, kidney function, diabetes, and physical activity, respectively. Cells marked with ``-'' indicate missing data for that stratum, which occurs because certain outcomes were not measured between ages 8 and 11.}
\label{tab:edatbl_summary_outcomes}
\fontsize{6.5}{7.8}\selectfont
\setlength{\tabcolsep}{4pt}
\renewcommand{\arraystretch}{1}
\setlength{\arrayrulewidth}{0.2pt}
\begin{tabular*}{\linewidth}{@{\extracolsep{\fill}}lcccc|cccc@{}}
\toprule
& \multicolumn{4}{c}{\shortstack{FPIR $<$ 1\\(N=1006)}} & \multicolumn{4}{c}{\shortstack{FPIR $>$ 1\\(N=1006)}} \\
\rule{0pt}{2.0em} & & & & & & & & \\
Outcome/Strata & \shortstack{Lower\\Quartile} & Median & \shortstack{Upper\\Quartile} & Trimean & \shortstack{Lower\\Quartile} & Median & \shortstack{Upper\\Quartile} & Trimean \\
\midrule
\hspace*{0.8em}- \textit{Cotinine (ng/mL)} & & & & & & & & \\
\hspace*{2.2em}8-11 & 0.03 & 0.10 & 0.71 & 0.24 & 0.01 & 0.04 & 0.16 & 0.06 \\
\hspace*{4.4em}Female & 0.02 & 0.07 & 0.58 & 0.18 & 0.01 & 0.04 & 0.21 & 0.08 \\
\hspace*{4.4em}Male & 0.03 & 0.14 & 0.80 & 0.28 & 0.01 & 0.03 & 0.10 & 0.04 \\
\hspace*{2.2em}12-17 & 0.04 & 0.19 & 1.27 & 0.42 & 0.02 & 0.05 & 0.32 & 0.11 \\
\hspace*{4.4em}Female & 0.03 & 0.15 & 1.03 & 0.34 & 0.02 & 0.04 & 0.29 & 0.10 \\
\hspace*{4.4em}Male & 0.04 & 0.26 & 2.12 & 0.67 & 0.02 & 0.05 & 0.33 & 0.11 \\
\rule{0pt}{1.1em} & & & & & & & & \\
\hspace*{0.8em}- \textit{Non-HDL Cholesterol (mmol/L)} & & & & & & & & \\
\hspace*{2.2em}8-11 & 2.27 & 2.71 & 3.24 & 2.74 & 2.37 & 2.75 & 3.23 & 2.77 \\
\hspace*{4.4em}Female & 2.26 & 2.81 & 3.17 & 2.76 & 2.38 & 2.70 & 3.06 & 2.71 \\
\hspace*{4.4em}Male & 2.27 & 2.65 & 3.27 & 2.71 & 2.34 & 2.81 & 3.39 & 2.84 \\
\hspace*{2.2em}12-17 & 2.17 & 2.63 & 3.26 & 2.67 & 2.27 & 2.70 & 3.21 & 2.72 \\
\hspace*{4.4em}Female & 2.22 & 2.65 & 3.23 & 2.69 & 2.23 & 2.69 & 3.20 & 2.70 \\
\hspace*{4.4em}Male & 2.15 & 2.59 & 3.29 & 2.65 & 2.30 & 2.71 & 3.23 & 2.74 \\
\rule{0pt}{1.1em} & & & & & & & & \\
\hspace*{0.8em}- \textit{Systolic BP (mm Hg)} & & & & & & & & \\
\hspace*{2.2em}8-11 & 96.00 & 101.17 & 108.00 & 101.58 & 94.75 & 100.67 & 106.00 & 100.52 \\
\hspace*{4.4em}Female & 94.50 & 101.33 & 107.50 & 101.17 & 94.67 & 101.67 & 106.67 & 101.17 \\
\hspace*{4.4em}Male & 96.00 & 100.67 & 108.00 & 101.33 & 95.08 & 100.67 & 106.00 & 100.60 \\
\hspace*{2.2em}12-17 & 102.00 & 108.67 & 115.33 & 108.67 & 102.67 & 109.33 & 116.00 & 109.33 \\
\hspace*{4.4em}Female & 101.33 & 106.67 & 113.67 & 107.08 & 100.67 & 106.00 & 112.67 & 106.33 \\
\hspace*{4.4em}Male & 103.33 & 111.33 & 116.67 & 110.67 & 104.67 & 112.00 & 120.00 & 112.17 \\
\rule{0pt}{1.1em} & & & & & & & & \\
\hspace*{0.8em}- \textit{Waist-to-Height} & & & & & & & & \\
\hspace*{2.2em}8-11 & 0.43 & 0.47 & 0.54 & 0.48 & 0.43 & 0.47 & 0.53 & 0.47 \\
\hspace*{4.4em}Female & 0.44 & 0.48 & 0.54 & 0.48 & 0.43 & 0.46 & 0.53 & 0.47 \\
\hspace*{4.4em}Male & 0.43 & 0.47 & 0.54 & 0.47 & 0.43 & 0.47 & 0.53 & 0.47 \\
\hspace*{2.2em}12-17 & 0.43 & 0.48 & 0.55 & 0.49 & 0.43 & 0.48 & 0.54 & 0.48 \\
\hspace*{4.4em}Female & 0.45 & 0.50 & 0.56 & 0.50 & 0.44 & 0.48 & 0.55 & 0.49 \\
\hspace*{4.4em}Male & 0.42 & 0.46 & 0.53 & 0.47 & 0.42 & 0.46 & 0.54 & 0.47 \\
\rule{0pt}{1.1em} & & & & & & & & \\
\hspace*{0.8em}- \textit{eGFR (mL/min/1.73m$^2$)} & & & & & & & & \\
\hspace*{2.2em}8-11 & - & - & - & - & - & - & - & - \\
\hspace*{4.4em}Female & - & - & - & - & - & - & - & - \\
\hspace*{4.4em}Male & - & - & - & - & - & - & - & - \\
\hspace*{2.2em}12-17 & 87.79 & 101.57 & 116.47 & 101.85 & 87.19 & 100.00 & 113.98 & 100.29 \\
\hspace*{4.4em}Female & 92.38 & 106.86 & 125.84 & 107.99 & 93.09 & 104.25 & 114.83 & 104.11 \\
\hspace*{4.4em}Male & 84.78 & 97.82 & 112.04 & 98.12 & 84.13 & 96.70 & 112.13 & 97.42 \\
\rule{0pt}{1.1em} & & & & & & & & \\
\hspace*{0.8em}- \textit{HbA1c (\%)} & & & & & & & & \\
\hspace*{2.2em}8-11 & - & - & - & - & - & - & - & - \\
\hspace*{4.4em}Female & - & - & - & - & - & - & - & - \\
\hspace*{4.4em}Male & - & - & - & - & - & - & - & - \\
\hspace*{2.2em}12-17 & 5.00 & 5.20 & 5.40 & 5.20 & 5.00 & 5.20 & 5.40 & 5.20 \\
\hspace*{4.4em}Female & 5.00 & 5.20 & 5.40 & 5.20 & 4.95 & 5.20 & 5.30 & 5.16 \\
\hspace*{4.4em}Male & 5.10 & 5.20 & 5.40 & 5.22 & 5.00 & 5.20 & 5.40 & 5.20 \\
\rule{0pt}{1.1em} & & & & & & & & \\
\hspace*{0.8em}- \textit{Vigorous LTPA (min/wk)} & & & & & & & & \\
\hspace*{2.2em}8-11 & - & - & - & - & - & - & - & - \\
\hspace*{4.4em}Female & - & - & - & - & - & - & - & - \\
\hspace*{4.4em}Male & - & - & - & - & - & - & - & - \\
\hspace*{2.2em}12-17 & 0.00 & 73.50 & 315.00 & 115.50 & 0.00 & 126.00 & 385.00 & 159.25 \\
\hspace*{4.4em}Female & 0.00 & 9.33 & 120.00 & 34.67 & 0.00 & 91.00 & 302.17 & 121.04 \\
\hspace*{4.4em}Male & 28.00 & 180.00 & 450.00 & 209.50 & 42.00 & 180.00 & 478.80 & 220.20 \\
\bottomrule
\end{tabular*}
\end{table}

\section{Hidden Bias in Matched Observational Studies} \label{sec: notation}

\subsection{Treatment effects and treatment assignments}

Suppose that we have an observational study with $I$ matched sets, where the $i$th set contains $n_i \geq 2$ individuals, among which one individual receives the treatment under investigation and $n_i-1$ receive the control. Denoting $Z_{ij}$ as the treatment indicator for individual $j$ in stratum $i$, we have $\sum_{j=1}^{n_i} Z_{ij} = 1$ and $N = \sum_{i=1}^{I} n_i$ total subjects. This framework readily extends to full-matched observational studies \citep{hansen2004full}. 
Individuals are matched on the basis of pretreatment covariates such that $\mathbf{x}_{ij} \approx \mathbf{x}_{ij'}$ for subjects $j \neq j'$ in stratum $i$, yet might fail to account for an unmeasured covariate $u_{ij} \in [0,1].$ In our study, we have $I = 4027$ pairs in the full sample and $n_i = 2$ for all $i \in [I].$%

There are $K$ outcome variables collected for each subject. For each outcome $k$, individual $j$ in matched set $i$ has two potential outcomes, ${r}_{T_{ijk}}$ under treatment and ${r}_{C_{ijk}}$ under control. Let $\mathbf{r}_{T_{ij}}$ and $\mathbf{r}_{C_{ij}}$ denote the $K$-dimensional vectors of potential outcomes for this subject. The observed response vector for this unit is $\mathbf{R}_{ij} = Z_{ij}\mathbf{r}_{Tij} + (1-Z_{ij})\mathbf{r}_{Cij}$ \citep{rubin1974estimating}. 
Since this individual receives treatment or control, never both, the treatment effect vector for this individual $\mathbf{r}_{T_{ij}} - \mathbf{r}_{C_{ij}}$ is never observed from the data \citep{neyman1923application,welch1937z}.

We define $\mathbf{Z} = (Z_{11}, \ldots, Z_{In_{I}})^\top$ as the lexicographically-ordered vector of treatment assignments of length $N$ and let the analogous hold for $\mathbf{u}$, $\mathbf{r}_{T_{k}}$, $\mathbf{r}_{C_{k}}$, and $\mathbf{R}_k$ ($k = 1, \ldots, K$). We take $\mathbf{R}$ as the $N \times K$ matrix of lexicographically-ordered rows containing $\mathbf{R}_{ij}^\top$. Additionally, let $\Omega$ be the set of $\prod_{i=1}^I n_i$ possible values of $\mathbf{Z}$ under the matched design and $\mathcal{Z} = \{ \mathbf{Z} \in \Omega \}$ be the event that the observed treatment assignment adheres to such a design. Our forthcoming developments consider inference conditional upon $\mathcal{Z}$ and the set $\mathcal{F}= \{ \mathbf{r}_{T_{ij}}, \mathbf{r}_{C_{ij}}, \mathbf{x}_{ij}, u_{ij} \mid i = 1, \ldots, I; \text{ } j = 1, \ldots, n_i \}$.
In a randomized experiment satisfying this design, each subject $j$ in matched set $i$ is chosen randomly to receive treatment such that $\P(Z_{ij} = 1 \mid \mathcal{F}, \mathcal{Z}) = 1/n_i$. Treatment allocation is independent across matched sets such that $\P(\mathbf{Z} = \mathbf{z} \mid \mathcal{F}, \mathcal{Z}) = 1/|\Omega|$, where we use $|\cdot|$ to denote the number of elements in a finite set.

\subsection{Randomization inference and sensitivity analysis} \label{sec: randomization-sensitivity}

We are interested in assessing whether poverty induces a harmful effect on CVD risk factors in children and adolescents, so our concern lies with testing multiple one-sided hypotheses. Without loss of generality, let us assume that our causal theory predicts positive values for treatment effects $\tau_{ijk}$ for each outcome variable $k$. We consider null hypotheses of the form
$H_k : r_{T_{ijk}} \leq r_{C_{ijk}}$ for all $ i \in [I]$ and $j \in [n_i].$
For example, the variable $k$ may represent waist-to-height ratio, a measurement that broadly reflects body composition. Then, $H_k$ posits that each individual's waist-to-height ratio under poverty is no larger than it would have been had they grown up above the poverty line.
To test each composite null $H_k$, we conduct inference on Fisher's sharp null of no treatment effect for outcome $k$ using effect-increasing sum statistics \citep[\S 4]{caughey2023randomisation}. Sum statistics take the form $T_k = \mathbf{Z}^\top\mathbf{q}_k$, where $\mathbf{q}_k$ is a pre-specified function of the observed responses $\mathbf{R}_k$ that is fixed under the null hypothesis. Sum statistics are effect-increasing if $(2z_{ij}-1)({r}^*_{ijk}-{r}_{ijk})$ implies $T_k(\mathbf{z},\mathbf{r}^*_k) \geq T_k(\mathbf{z},\mathbf{r}_k)$ for potential outcomes ${r}^*_{ijk} \neq {r}_{ijk}$ for all individuals in the study.

If our observed data were gathered from a randomized experiment and we assume that the sharp null hypothesis of no treatment effect is true for outcome $k$, then
the null distribution of its test statistic is given as its randomization distribution.
In an observational study, a researcher can conduct a sensitivity analysis of her results using the Rosenbaum model (\citeyear{rosenbaum2002covariance}, \S 4), which stipulates that individuals $j \neq j'$ in matched set $i$ may differ in their odds of assignment to treatment due to the presence of an unmeasured confounder by a factor no larger than $\Gamma$:
\begin{equation} \label{sensitivity-model1}
    \frac{1}{\Gamma} \leq \dfrac{ \P(Z_{ij}=1\mid\mathcal F)\, \P(Z_{ij'}=0\mid\mathcal F) }{ \P(Z_{ij'}=1\mid\mathcal F)\,\P(Z_{ij}=0\mid\mathcal F) } \leq \Gamma.
\end{equation}
The parameter $\Gamma$ controls the degree to which matching solely on observed covariates may return treatment assignment probabilities that differ from random allocation. Specifically, a value of $\Gamma = 1$ yields the usual randomized design, while $\Gamma > 1$ allows for a tilt in the randomization distribution by an amount controlled by $\Gamma$. Returning to the matched structure by conditioning on $\mathcal{Z}$, this is equivalent to assuming that:
\begin{equation} \label{sensitivity-model2}
    \P (\mathbf{Z} = \mathbf{z} \mid \mathcal{F}, \mathcal{Z}) = \frac{\exp(\gamma \mathbf{z}^T \mathbf{u})}{\sum_{\mathbf{b} \in \Omega} \exp(\gamma \mathbf{b}^T \mathbf{u})},
\end{equation}
where $\gamma = \log(\Gamma)$ and $\mathbf{u}$ is in the $N$-dimensional unit cube $\mathcal{U} = [0, 1]^N$. %

Following this sensitivity model and assuming that the sharp null hypothesis holds for outcome $k$, we can write the null distribution of the test statistic for $\Gamma > 1$ as:
\begin{equation} \label{null-distribution}
\P\left(T_k(\mathbf{Z},\mathbf{r}_{C_{k}}) \ge v \mid \mathcal{F}, \mathcal{Z} \right) = \sum_{\mathbf{z} \in \Omega} \1\left(T_k(\mathbf{Z}, \mathbf{r}_{C_{k}}) \geq v\right) \frac{ \exp(\gamma \mathbf{z}^T \mathbf{u})}{\sum_{\mathbf{b} \in \Omega} \exp(\gamma \mathbf{b}^T \mathbf{u})},
\end{equation}
where $\1(\cdot)$ denotes the indicator function. For $\Gamma > 1$, Equation \ref{null-distribution} is unknown due to its dependence on the nuisance vector $\mathbf{u}$. We proceed to conduct sensitivity analysis for a particular value of $\Gamma$ by finding the worst-case unmeasured confounder that maximizes the tail probability under the sensitivity model. 
A researcher often calculates a set of worst-case $p$-values under increasing values of $\Gamma$. The change-point value of $\Gamma$ at which the test fails to reject, also known as the sensitivity value \citep{zhao2019sensitivityval}, serves as a measure of the robustness of the results of an analysis to possible violations of the unconfoundedness assumption. Rosenbaum's (\citeyear{rosenbaum2004design}) design sensitivity is the stochastic limit of the sensitivity value as $I\rightarrow \infty$.

\section{A Framework for Multivariate Sensitivity Analysis} \label{sec: methods}

\subsection{Testing the global null hypothesis}

Since we intend to evaluate whether poverty causes a harmful effect on any of the CVD risk factors, we consider a directional global null hypothesis that all $K$ composite null hypotheses are true:
\begin{equation} \label{null-global}
    H_0: \bigcap_{k=1}^K H_k \qquad \text{ versus } \qquad H_1: \bigcup_{k=1}^K H^c_k.
\end{equation}
In our study, a rejection of $H_0$ would imply that at least one CVD risk factor is adversely affected by growing up in poverty.
Strategies such as closed testing and hierarchical testing procedures draw on a valid sensitivity analysis of (\ref{null-global}) to facilitate tests of individual $H_k$ while strongly controlling the family-wise error rate. This could inform us which CVD risk factors are negatively impacted by poverty following a rejection of the overall null.

\subsection{Weighted combinations of test statistics} \label{subsec:weighted-stats}

We first define the conditional treatment assignment probability $\varrho_{ij} = \P(Z_{ij} = 1 \mid \mathcal{F}, \mathcal{Z})$. Moving forward, expectations and variances for each $k$ are taken with respect to the randomization distribution of its test statistic and are implicitly conditional on $\mathcal{F}$ and $\mathcal{Z}$.
Under the global null (\ref{null-global}) that all $H_k$ are true, each $\mathbf{q}_k$ is a fixed quantity and we write the expectation $\boldsymbol{\mu}(\boldsymbol{\varrho})$ and covariance $\boldsymbol{\Sigma}(\boldsymbol{\varrho})$ for the vector of test statistics $\mathbf{T} = (T_1, \ldots, T_K)^\top$ as:
\begin{equation} \label{expectation-covariance}
\boldsymbol{\mu}(\boldsymbol{\varrho})_k 
= 
\sum_{i=1}^I \sum_{j=1}^{n_i} q_{ijk} \varrho_{ij}, 
\qquad
\boldsymbol{\Sigma}(\boldsymbol{\varrho})_{k,\ell} 
=
\sum_{i=1}^I \Biggl(
\sum_{j=1}^{n_i} q_{ijk} q_{ij\ell} \varrho_{ij}
-
\Bigl(\sum_{j=1}^{n_i} q_{ijk} \varrho_{ij}\Bigr)
\Bigl(\sum_{j=1}^{n_i} q_{ij\ell} \varrho_{ij}\Bigr)
\Biggr).
\end{equation}

Suppose that
$
I^{-1}\boldsymbol{\Sigma}(\boldsymbol{\varrho})\to \mathbf{S}(\boldsymbol{\varrho})
$
for some positive semidefinite matrix $\mathbf{S}(\boldsymbol{\varrho})$. The Lindeberg-Feller central limit theorem tells us that, for any fixed $\boldsymbol{\lambda} = (\lambda_1, \ldots, \lambda_K)^\top$ and probabilities $\boldsymbol{\varrho}$, the random variable $I^{-1/2}\,\boldsymbol{\lambda}^\top\bigl(\mathbf T-\boldsymbol{\mu}(\boldsymbol{\varrho})\bigr)$ converges to a Normal distribution with mean zero and variance $\boldsymbol{\lambda}^\top \mathbf{S}(\boldsymbol{\varrho})\boldsymbol{\lambda}$ as $I \rightarrow \infty$ under suitable regularity conditions; see, for instance, \cite[\S 8.10]{rosenbaum2025introduction}. An application of the Cramér--Wold device then implies that $I^{-1/2}\bigl( \mathbf{T}-\boldsymbol{\mu}(\boldsymbol{\varrho})\bigr)$ is asymptotically multivariate Normal with mean zero and covariance $\mathbf{S}(\boldsymbol{\varrho})$. Although we treat the vector of assignment probabilities $\boldsymbol{\varrho}$ as fixed, its actual values are unknown due to their dependence on hidden bias. Additionally, these unknown probabilities are constrained to elements of the polyhedral set $\mathcal{P}_{\Gamma}$. In our study, these probabilities encode possible hidden differences in the odds of growing up in poverty within matched pairs, while the vector $\mathbf T$ summarizes evidence across the CVD risk factors under investigation.

\subsection{Sensitivity analysis with multiple outcomes}

Tests of the global null hypothesis in a sensitivity analysis framework can be characterized as instances of a two-player adversarial game, in which the researcher determines the best choice of test statistics, subject to certain constraints, to combat nature’s presentation of the worst-case pattern of hidden bias; see, for instance, \cite{intriligator2002mathematical} and
\cite{yu2002competition} for more on these types of games.
Consider $\boldsymbol{\lambda}$ as belonging to some set $\Lambda \subseteq \mathbb{R}^K$. Taking $\Lambda = \{\mathbf{e}_k\}$, where $\mathbf{e}_k$ is the $k$th standard basis vector, returns a univariate sensitivity analysis for the $k$th outcome with a greater-than alternative. 
When $T_K$ are rank tests, defining $\Lambda = \{\mathbf{1}_K\}$, the $K$-vector of ones, returns the test of \cite{rosenbaum1997signed}. The method of \cite{fogarty2016sensitivity} with greater-than alternatives follows by choosing $\Lambda = \{\textbf{e}_1, \ldots, \textbf{e}_K\}$,
while those of \cite{rosenbaum2016using} and
\cite{cohen2020multivariate} use $\Lambda =\mathbb{R}^K \setminus\{\mathbf{0}_K\}$ and the nonnegative orthant, respectively.

Our investigation motivates the introduction of our method in the context of \cite{cohen2020multivariate}, as we seek to assess whether poverty has a harmful impact on any outcomes. Let $\mathbf{t} = (t_1, \ldots, t_K)^\top$ be the observed vector of test statistics. \cite{cohen2020multivariate} formulate their sensitivity analysis by considering the following problem:
\begin{equation} \label{cohen-sensanalysis}
\min_{\boldsymbol{\varrho} \in \mathcal{P}_{\Gamma}} 
\sup_{\boldsymbol{\lambda} \in \Lambda^+ \setminus\{\mathbf{0}\}} \; \max
\left \{
0,\;
\frac{\boldsymbol{\lambda}^\top\bigl(\mathbf{t} - \boldsymbol{\mu}(\boldsymbol{\varrho})\bigr) }
{ \sqrt{\boldsymbol{\lambda}^\top \boldsymbol{
\Sigma}(\boldsymbol{\varrho})\boldsymbol{\lambda}} }
\right \}^2,
\end{equation}
where $\Lambda^+ := \{\boldsymbol{\lambda} \mid \lambda_k \geq 0 \text{ } (k = 1,\ldots, K)\}$ denotes the positive orthant and the objective is a monotone non-decreasing transformation of the standardized deviate. As the value of $\Gamma$ increases to the sensitivity value, the outer minimization takes place over a sequence of growing feasible regions. In the sense of the two-player game, this corresponds to nature having more and more flexibility in assigning unfavorable treatment allocation distributions, which a researcher seeks to make up for by adjusting the linear combination of test statistics to try to maximize the objective value. \cite{cohen2020multivariate} prove that when replacing the optimization over $\boldsymbol{\varrho}$ by the value of the test statistic at the true treatment assignment probabilities $\tilde{\boldsymbol{\varrho}}$, the test statistic in (\ref{cohen-sensanalysis}) converges to a chi-bar-squared ($\bar{\chi}^2$) distribution under the sharp null hypothesis as $I \rightarrow \infty$. The $\bar{\chi}^2$ distribution is a common family of distributions that arises in order-restricted statistical inference \citep{silvapulle2011constrained}. To proceed with valid inference of the global null hypothesis in (\ref{null-global}), the practitioner could compare her standardized deviate to the square root of the $(1-\alpha)$--quantile of $\bar{\chi}^2_{\Lambda^+\setminus\{\mathbf{0}\}}(\tilde{\boldsymbol{\varrho}})$, the chi-bar-squared distribution over the $K$--dimensional nonnegative orthant with covariance matrix yielded $\tilde{\boldsymbol{\varrho}}$. 

Importantly, the critical value $c_{\alpha,\Lambda}$ used in multivariate sensitivity analyses depends upon the structure of $\Lambda$. When $\Lambda$ is a singleton, as is the case in \cite{rosenbaum1997signed}, the asymptotic reference distribution is standard Normal. However, if the combination $\boldsymbol{\lambda} \in \Lambda$ is not pre-specified by the researcher, comparing the optimal value of the standardized deviate to the $(1-\alpha)$--quantile of the standard Normal distribution may provide a valid level--$\alpha$ test of the global null hypothesis. In the approach proposed by \cite{cohen2020multivariate}, the authors use $c_{\alpha,\Lambda^+ \setminus\{\mathbf{0}\}} = \sqrt{\bar{\chi}^2_{\Lambda^+\setminus\{\mathbf{0}\}, 1-\alpha}(\tilde{\boldsymbol{\varrho}})}.$ Likewise, \cite{rosenbaum2016using} applies a classical result on quadratic forms of the multivariate Normal distribution to motivate $c_{\alpha,\mathbb{R}^K \setminus\{\mathbf{0}_K\}} = \sqrt{{\chi}^2_{K, 1-\alpha}}$. Notably, the (possibly conservative) critical values used by each of these procedures are increasing functions of the number of outcomes $K$ in the global null hypothesis.

Although choosing a richer set $\Lambda$ allows the researcher to select among combinations that may combat the adversarial treatment assignment more successfully, thus compensating for the diminished power resulting from inflated critical values, for a fixed sample size these benefits diminish as $K$ grows. In Table \ref{tab:quantiles}, we see that the critical values for various reference distributions increase markedly with $K$, forcing these previous strategies to compare the optimized test statistic with progressively more stringent critical values for valid level--$\alpha$ inference. 
\begin{table}[!t]
\centering
\caption{Critical values for relevant distributions at values of $K$. We show ($1-\alpha$)--quantiles of the standard Normal distribution and the square root of the ($1-\alpha$)--quantiles for remaining distributions at $\alpha = 0.05$. Quantiles of the $\bar{\chi}^2$ distribution are taken over the nonnegative orthant. $\rho$ denotes the constant correlation among all $K$ outcomes.} 
\resizebox{\textwidth}{!}{
\begin{tabular}{l c c c c}
& {Standard Normal} & {$\sqrt{\text{Chi-Bar-Squared ($\rho = 0.2$)}}$} & {$\sqrt{\text{Chi-Bar-Squared ($\rho = 0$)}}$} & {$\sqrt{\text{Chi-Squared}}$} \\
\midrule
$K=5$   & $1.645$ & $2.566$ & $2.735$ & $3.327$ \\
$K=15$  & $1.645$ & $3.301$ & $3.957$ & $5.000$ \\
$K=25$  & $1.645$ & $3.663$ & $4.773$ & $6.136$ \\
$K=100$ & $1.645$ & $4.642$ & $8.340$ & $11.151$ \\
\bottomrule
\end{tabular}
}
\label{tab:quantiles}
\end{table}
Our study of the effects of poverty on CVD risk factors considers more outcomes than prior multivariate sensitivity analyses have handled. For example, we initially examine $K = 9$ outcomes among adolescents, more than any prior applications of these methods; hence, the power loss from the stringent corrections used by existing techniques is especially concerning in our context. Moreover, the quantiles of the $\bar{\chi}^2$ distribution must typically be evaluated by time-consuming simulation, even for known correlation structures. The quantiles of the $\bar{\chi}^2$ distribution with $\rho = 0.2$ were approximated using one million Monte-Carlo iterations; $\rho=0$ is a special case that admits closed-form solutions. In practice, $\boldsymbol{\tilde{\varrho}}$ is unknown, so evaluating (\ref{cohen-sensanalysis}) involves approximating an upper bound for the worst-case critical value at each $\Gamma > 1$. %

\section{Sample Splitting for a Data-Driven Design} \label{sec: splitting}

\subsection{Method description} \label{sec: ourmethod}

Randomly splitting the data offers a straightforward approach to constructing a data-informed design that enables powerful sensitivity analysis across a wider range of $K$. We incorporate into our design stage the estimation of an optimizing linear combination on the pilot sample and carry over this estimate to the analysis sample for inference. 
For fixed values of $\alpha$ and $\Gamma$, we define the approach of sample splitting for multivariate sensitivity analysis as follows:
    \begin{enumerate}[label=\textbf{Step \arabic*: }, leftmargin=*,topsep=0pt]
    \item Randomly partition the dataset into planning and analysis samples, denoted as $\mathcal{G}_{\p}$ and $\mathcal{G}_{\a}$ respectively, such that $\mathcal{G}_{\p} \cap \mathcal{G}_{\a} = \emptyset$.
    \item Choose a single linear combination from $\G_{\p}$ without looking at $\G_{\a}$. In particular, solve on $\G_{\p}$ the following optimization problem:
    \begin{equation} \label{p1}
    \sup_{\boldsymbol{\lambda} \in \Delta_+} \min_{\boldsymbol{\varrho} \in \mathcal{P}_{\Gamma_\p}} \max\left \{0,\frac{\boldsymbol{\lambda}^{\top} (\textbf{t}-\boldsymbol{\mu}(\boldsymbol{\varrho}))}{\sqrt{\boldsymbol{\lambda}^{\top} \boldsymbol{\Sigma}(\boldsymbol{\varrho})\boldsymbol{\lambda}}}\right \}^2,
    \end{equation}
    where $
\Delta_+
:=
\left\{
\boldsymbol{\lambda}\in\mathbb R^K :
\lambda_k\ge 0 \text{ for all } k,
\ \mathbf 1^\top \boldsymbol{\lambda} =1
\right\}
$
denotes the probability simplex and $\Gamma_\p$ the sensitivity value on the planning sample computed by searching over $\Gamma$, with optimizer $\widehat{\boldsymbol{\lambda}}_\p$.
    \item Solve the following optimization problem on $\G_{\a}$:
    \begin{equation} \label{p2}
    \min_{\boldsymbol{\varrho} \in \mathcal{P}_{\Gamma}} \max\left \{0,\frac{\widehat{\boldsymbol{\lambda}}_\p^{\top} (\textbf{t}_\a-\boldsymbol{\mu}_\a(\boldsymbol{\varrho}))}{\sqrt{\widehat{\boldsymbol{\lambda}}_\p^{\top} \boldsymbol{\Sigma}_\a(\boldsymbol{\varrho})\widehat{\boldsymbol{\lambda}}_\p}}\right \}^2.
    \end{equation}
    Report the square root of the objective value and compare it to the ($1-\alpha$)--quantile of the standard Normal distribution. Reject $H_0$ if the square root of the objective value exceeds the critical value, otherwise fail to reject.
    \end{enumerate}
    
\begin{proposition}
    \label{validity}
    Sample splitting for multivariate sensitivity analysis strongly controls the family-wise error rate at the nominal level $\alpha$.
\end{proposition}
\textit{Proof.} The result is a consequence of independence between data splits and validity of the test in (\ref{p2}), following immediately by the law of iterated expectations. \hfill $\qed$

Data splitting offers numerous benefits for multivariate sensitivity analysis. By introducing a pilot sample, the researcher can obtain both qualitative and quantitative insights from her data without having to correct for data snooping \citep{tukey1953multiple, tukey1977exploratory}.
In our study, we gain various insights from the planning sample that help to inform our analysis plan. For example, we use the pilot sample to modify our measure of dietary intake that we examine on the analysis sample. 
Since our approach allows us to treat a linear combination estimated on the planning sample as pre-specified when conducting inference, a valid level--$\alpha$ test of the global null proceeds by comparing the square root of the objective value to the ($1-\alpha$)--quantile of the standard Normal distribution, which crucially does not depend on $K$. By circumventing use of the $\bar{\chi}^2$ distribution to conduct  inference, our approach often recovers much of the lost power due to reduced sample size. Moreover, computational burden is reduced and the separable algorithm of \cite{gastwirth2000asymptotic} can be used to identify the worst-case assignment probabilities in $O(I)$ time.

\subsection{Changing the two-player game}\label{subsec:twoplayergame}

A remaining issue is that this split-sample sensitivity analysis formulation differs from that of \cite{cohen2020multivariate}. In particular, the minimax inequality %
indicates that sample splitting procedures can do strictly worse than full-sample alternatives; specifically, the following holds for general $\Lambda$:
\begin{equation} \label{minmax-inequality}
\min_{\boldsymbol{\varrho} \in \mathcal{P}_{\Gamma}} 
\sup_{\boldsymbol{\lambda} \in \Lambda} \;
\max
\left \{
0,\;
\frac{\boldsymbol{\lambda}^\top\bigl(\mathbf{t} - \boldsymbol{\mu}(\boldsymbol{\varrho})\bigr) }
{ \sqrt{\boldsymbol{\lambda}^\top \boldsymbol{
\Sigma}(\boldsymbol{\varrho})\boldsymbol{\lambda}} }
\right \}^2
\;\geq \;
\sup_{\boldsymbol{\lambda} \in \Lambda} 
\min_{\boldsymbol{\varrho} \in \mathcal{P}_{\Gamma}} \;
\max
\left \{
0,\;
\frac{\boldsymbol{\lambda}^\top\bigl(\mathbf{t} - \boldsymbol{\mu}(\boldsymbol{\varrho})\bigr) }
{ \sqrt{\boldsymbol{\lambda}^\top \boldsymbol{
\Sigma}(\boldsymbol{\varrho})\boldsymbol{\lambda}} }
\right \}^2,
\end{equation}
In the sense of the two-player game, the left-hand side of (\ref{minmax-inequality}), similar to \cite{cohen2020multivariate}, sees nature give out adverse treatment assignments to minimize the objective, to which a researcher responds by adjusting her linear combination to increase the value of the standardized deviate. 
It can be regarded that the left-hand side allows the researcher to see the effects of the worst-case treatment allocation offered by nature, so that it can react to any choice made by nature when maximizing the objective.
The right-hand side of the game can be interpreted as a researcher first choosing combinations of test statistics that produce increasingly strong evidence against the null hypothesis, which nature then looks to diminish by adjusting the treatment allocation distributions to reduce the objective value. 
The order reflected on the right-hand side, with the practitioner first and nature second, is realized by our sample splitting approach, where the researcher estimates a linear combination via~\eqref{p1} on the pilot sample and nature responds at this fixed direction via~\eqref{p2} on the analysis sample. 
Intuitively, the right-hand side forces the researcher to first specify a linear combination that she hopes can withstand the opposition brought by nature, prior to nature acting adversarially and assigning unfavorable treatments in reaction to the practitioner's choice. As a result, the setting on the left-hand side will always be at least as favorable to the researcher as the right-hand side.

Indeed, for certain choices of $\Lambda$, solving the optimization problem on the right-hand side yields strictly lower power compared to that on the left-hand side. \cite[\S 3.2]{fogarty2016sensitivity} propose a scenario in which testing the most promising individual hypothesis after conducting multiple sensitivity analyses, corresponding to solving the problem on the right-hand side with $\Lambda$ the set of standard basis vectors, can have strictly worse power than the left-hand side. They note that, for each outcome $k$, there can exist a worst-case confounder that minimizes evidence against $H_k$ such that each $H_k$ fails to reject, yet there may not exist a single unmeasured confounder for which all $H_k$ simultaneously fail to reject. This phenomenon is plausible in our study because a hidden bias that could explain differences in cotinine levels, such as parental smoking, need not be the same hidden bias that could explain differences in waist-to-height ratio; we expand upon this example in the Supplementary Material.
In this situation, strict inequality in (\ref{minmax-inequality}) occurs: the right-hand side, which conducts multiple sensitivity analyses separately for each $k$ then corrects for multiple comparisons, yields an objective value smaller than the critical value, while solving directly for the left-hand side would lead to an objective value greater than the critical value and a rejection of the global null hypothesis. Similarly, \cite{heng2021increasing} provide an example in which two outcomes have test statistics that are perfectly negatively correlated. Although each individual $H_k$ has a design sensitivity of three, a test of the global null yields an infinite design sensitivity.%

Previous works have implicitly acknowledged this gap and formulated their sensitivity analyses  to solve the left-hand side of (\ref{minmax-inequality}). However, our method was developed specifically to circumvent non-Normal reference distributions and to facilitate sensitivity analysis for larger values of $K$. To do so, the researcher is directed to pre-specify a linear combination prior to conducting sensitivity analysis on the analysis sample. Yet forcing the practitioner to act first and allowing nature to respond compels our problem to be formulated as the right-hand side of the inequality. To ensure that our sample splitting approach does not lag behind due to being on the right-hand side of the minimax inequality, and hence does not sacrifice power in our study,
we must strengthen the result in (\ref{minmax-inequality}) to equality when $\Lambda$ is the probability simplex, which requires us to introduce a minimax theorem that is novel in the literature of sensitivity analysis with multiple outcomes. Since our procedure does not solve the right-hand side of (\ref{minmax-inequality}) on a single sample, instead splitting maximization and minimization across planning and analysis samples, respectively, we must also verify that our strategy retains the asymptotic power afforded by the minimax equality.

\subsection{Minimax theorem and design sensitivity} \label{sec: theory}

To accommodate positive semidefinite covariance matrices, we establish the minimax theorem with a new extended-real objective. The minimax theorem and subsequent design sensitivity result together justify using sample splitting to select combinations of CVD risk factors on a pilot sample for formal testing of global null hypotheses in our poverty study.

\begin{theorem} \label{minmax-thm}
Let $\Lambda \subseteq \mathbb R^K$ be a nonempty convex set. For every $\boldsymbol{\lambda} \in \Lambda$ and $\boldsymbol{\varrho} \in \mathcal{P}_{\Gamma}$, define the extended-real payoff
\[
F(\boldsymbol{\lambda},\boldsymbol{\varrho})
:=
\begin{cases}
\max
\left \{
0,\;
\frac{\boldsymbol{\lambda}^\top\bigl(\mathbf{t} - \boldsymbol{\mu}(\boldsymbol{\varrho})\bigr) }
{ \sqrt{\boldsymbol{\lambda}^\top \boldsymbol{
\Sigma}(\boldsymbol{\varrho})\boldsymbol{\lambda}} }
\right \}^2,
& \text{if } \boldsymbol{\lambda}^\top \boldsymbol{\Sigma}(\boldsymbol{\varrho})\boldsymbol{\lambda}>0,\\[2ex]
+\infty,
& \text{if } \boldsymbol{\lambda}^\top \boldsymbol{\Sigma}(\boldsymbol{\varrho})\boldsymbol{\lambda}=0
\text{ and } \boldsymbol{\lambda}^\top(\mathbf{t}-\boldsymbol{\mu}(\boldsymbol{\varrho}))>0,\\[1ex]
0,
& \text{if } \boldsymbol{\lambda}^\top \boldsymbol{\Sigma}(\boldsymbol{\varrho})\boldsymbol{\lambda}=0
\text{ and } \boldsymbol{\lambda}^\top(\mathbf{t}-\boldsymbol{\mu}(\boldsymbol{\varrho}))\le 0.
\end{cases}
\]
Then,
$$\min_{\boldsymbol{\varrho} \in \mathcal{P}_{\Gamma}} 
\sup_{\boldsymbol{\lambda} \in \Lambda} \;
F(\boldsymbol{\lambda},\boldsymbol{\varrho})
\;= \;
\sup_{\boldsymbol{\lambda} \in \Lambda} 
\min_{\boldsymbol{\varrho} \in \mathcal{P}_{\Gamma}} \;
F(\boldsymbol{\lambda},\boldsymbol{\varrho}).$$
\end{theorem}

\textit{Proof.} The result relies on a quadratic reformulation of the objective and use of Sion's Minimax Theorem \citep{sion1958general}. Details are provided in the Supplementary Material. \hfill $\qed$

\begin{corollary} \label{corr1}
Minimax equality holds for the extended-real payoff on $\Lambda^+$ and $\Delta_+$:
\begin{align} \label{minimax-eq-posorth}
\min_{\boldsymbol{\varrho} \in \mathcal{P}_{\Gamma}} 
\sup_{\boldsymbol{\lambda} \in \Lambda^+} \;
F(\boldsymbol{\lambda},\boldsymbol{\varrho})
\;= \;
\sup_{\boldsymbol{\lambda} \in \Lambda^+} 
\min_{\boldsymbol{\varrho} \in \mathcal{P}_{\Gamma}} \;
F(\boldsymbol{\lambda},\boldsymbol{\varrho}). \tag{a}
\end{align}
\vspace{-3em}
\begin{align} \label{minimax-eq-simplex}
\min_{\boldsymbol{\varrho} \in \mathcal{P}_{\Gamma}} 
\sup_{\boldsymbol{\lambda} \in \Delta_+} \;
F(\boldsymbol{\lambda},\boldsymbol{\varrho})
\;= \;
\sup_{\boldsymbol{\lambda} \in \Delta_+} 
\min_{\boldsymbol{\varrho} \in \mathcal{P}_{\Gamma}} \;
F(\boldsymbol{\lambda},\boldsymbol{\varrho}). \tag{b}
\end{align}
\end{corollary}

\textit{Proof.} 
The equality of Theorem~\ref{minmax-thm} yields the desired results in (\ref{minimax-eq-posorth}) and (\ref{minimax-eq-simplex}) since the positive orthant and probability simplex are both convex. \hfill $\qed$

Corollary \ref{corr1} reveals that the objective value returned by both formulations of the two-player game  are equivalent without any dependence on the sample size $I$. Theorem \ref{minmax-thm} relies on the linear combinations of test statistics comprising a convex set.
This result has not previously been characterized, yet fully accommodates the procedures of \cite{rosenbaum2016using} and \cite{cohen2020multivariate} while also justifying the counterexamples presented by \cite{fogarty2016sensitivity} and \cite{heng2021increasing}. 
Importantly, this is a fairly weak condition that will hopefully facilitate the further development of new approaches for multivariate sensitivity analysis that rely on sample splitting.

To evaluate the large-sample power of our sample splitting procedure, we rely on Rosenbaum's notion of design sensitivity (\citeyear{rosenbaum2004design}). Suppose that a researcher conducts an observational study under favorable conditions, whereas the treatment being investigated has a true causal effect in the direction of the alternative and there is no unmeasured confounding.
The power of a sensitivity analysis is defined under this favorable regime and is the probability of correctly rejecting the null hypothesis while allowing up to a pre-specified amount of bias $\Gamma$. \cite{rosenbaum2004design} introduces design sensitivity $\tilde{\Gamma}$ in this setting; under mild conditions, $H_0$ is rejected with probability one for $\Gamma < \tilde{\Gamma}$ and with probability zero for $\Gamma > \tilde{\Gamma}$ as $I \rightarrow \infty$. This value $\tilde{\Gamma}$ exists under mild conditions and is a concise indicator of large-sample power, closely related to the Bahadur efficiency of a sensitivity analysis \citep{rosenbaum2015bahadur}.
The minimax theorem allows us to assess the power of our method in this favorable regime and conclude that our procedure does not lose power asymptotically.

\begin{theorem} \label{convergence-thm}
    Under regularity conditions given in Appendix~\ref{appendix:proofs}, the design sensitivity achieved by conducting sensitivity analysis via  \eqref{cohen-sensanalysis} %
    with the whole sample equals the design sensitivity attained by our sample splitting approach given in \eqref{p2}.
\end{theorem}

The following result follows as a corollary to Theorem \ref{convergence-thm}:

\begin{corollary} \label{corr2}
    Suppose the regularity conditions of Theorem \ref{convergence-thm} hold. Then, the distance between any linear combination that maximizes the sensitivity value on the planning sample and the set of population maximizers of the design sensitivity converges to zero in probability.
\end{corollary}

Theorem \ref{convergence-thm} guarantees that the design sensitivity of our test equals that of \cite{cohen2020multivariate}, implying that we lose no power asymptotically from our sample splitting procedure.
Corollary \ref{corr2} justifies our use of a combination that maximizes the sensitivity value on the pilot sample to test on the analysis sample.

\section{Simulation results} \label{sec: sims}

We demonstrated that our test has the same limiting power as competitors, so it is now of primary interest to evaluate how our sample splitting method fares when applied to finite samples.
Along with full-sample alternatives, we also implement an oracle that tests on the full sample the combination of outcomes that maximizes the design sensitivity using the standard Normal reference distribution to assess the extent of what is possible. This corresponds to a practitioner (unrealistically) knowing and pre-specifying the linear combination that yields the largest design sensitivity as $I \rightarrow \infty$ before examining any data. We compare the power of the sensitivity analyses across various settings that allow us to assess how the number of outcomes and correlation among outcomes can affect their relative performance.

We consider $K \in \{5, 15, 25\}$ outcomes which are equi-correlated with constant correlation parameter $\rho \in \{0, 0.2\}$. We assume the data $Y_{ik}$ are taken from a pair-matched study and with total number of matched pairs $I ={300}$. These regimes mimic the primary statistical challenge of our NHANES analysis, with a moderate number of CVD risk factors and mild positive dependence, where relying on dimension-dependent critical values would likely incur a substantial loss in power.
For each simulation, $I$ observations of each outcome variable are drawn from a multivariate Normal distribution with mean $\boldsymbol{\tau}$ and covariance depending on $\rho$, where $\boldsymbol{\tau}$ is the unknown vector of treatment effects for each outcome $k$. For each value of $K$, we take $\boldsymbol{\tau}$ such that $80\%$ of the outcomes have a treatment effect of $0.1$ and the remaining $20\%$ have a treatment effect of $0.5$. This corresponds to the favorable setting with true causal effects and no hidden biases. Each outcome variable employs a test statistic $T_k = \sum_{i=1}^I \text{sign}(Y_{ik})\;\min(|Y_{ik}|/s_k, 2.5)$, where $s_k$ is the median of $|Y_{ik}|$ $(i = 1,\ldots, I)$. This amounts to a choice of $m$-statistic with Huber’s $\psi$-function \citep{rosenbaum2007sensitivity}. %

\begin{figure}[t!]
    \centering
    \includegraphics[width = \textwidth]{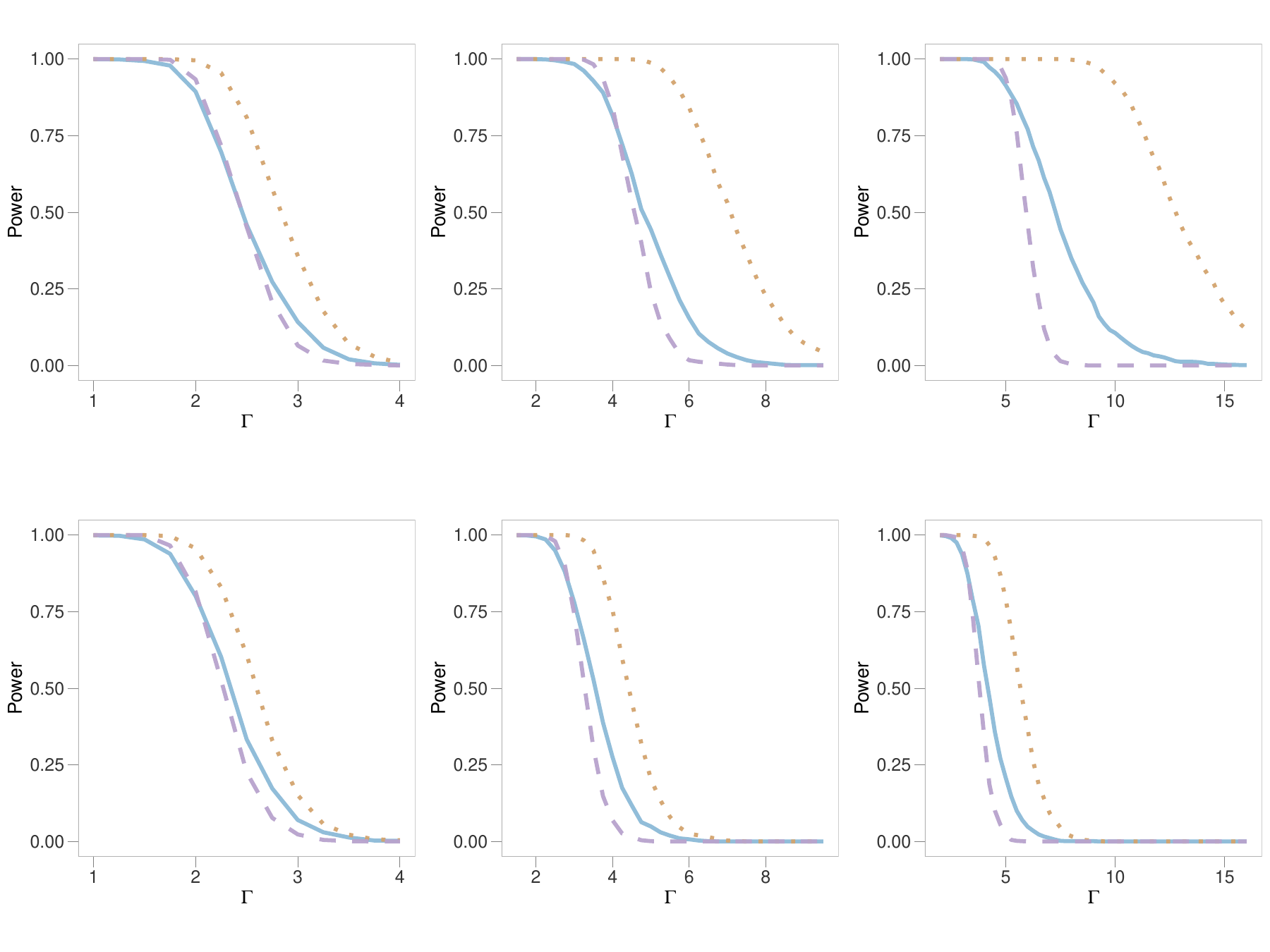}
    \caption{Power comparisons between the method of \cite{cohen2020multivariate} (dashed, purple), the sample splitting approach (solid, blue), and the oracle test (dotted, beige) as $\Gamma$ increases with $I = 300$. The left column has $K = 5$; the center has $K = 15$; and the right has $K = 25$. The first row has $\rho = 0$ and the second row has $\rho = 0.2$.}
    \label{sparse-I300}
\end{figure}

We display the estimated power curves for each test as a function of $\Gamma \geq 1$ with $1000$ simulations for each combination of parameters.
We use \cite[\S 4.1]{rosenbaum2016using} to evaluate the design sensitivity for a fixed test statistic in the pair-matched setting. Our method is implemented by splitting the data into $0.25I$ pairs to estimate the optimal linear combination and $(1-0.25)I$ samples for inference. %
Figure \ref{sparse-I300} illustrates the results, varying the number of outcomes $K$ across the columns and the correlation between paired differences down the rows. We observe across these various regimes that the oracle outperforms our method, while our method typically has higher power than the full-sample implementation of \cite{cohen2020multivariate}. Looking across the columns from left to right, we notice that there appears to be a more pronounced disparity in power among the various approaches as $K$ increases. In particular, we see that our procedure is much better suited to regimes with a higher number of outcomes than \cite{cohen2020multivariate}. However, it is also notable that these benefits start to manifest for smaller values of $K$ as well. Going from the top row of uncorrelated outcomes to the bottom row with mild to moderate positive dependence, we see that the performance gap tightens. While our procedure still outperforms the full-sample comparator, it does so by a smaller margin compared to when the outcomes are independent.
In the Supplementary Material, we compare our findings with $I=1000$ pairs. There, we see that the difference in performance appears to diminish as $I$ grows larger. This is related to Theorem \ref{convergence-thm}, which shows that all three methods will eventually approach the same limiting design sensitivity.

Despite reserving fewer samples for inference due to data splitting, our method typically makes up for the loss in power due to reduced sample size, and often surpasses the power of full-sample approaches in finite samples. This improvement is particularly evident as $K$ grows.
We draw a comparison between our work and recent literature in a more traditional setting of testing multivariate mean testing \citep{huang2015projection,liu2024projection}. Although the focus of these works differs, both use data splitting to find accurate projections in moderate and high-dimensional regimes to improve power. Likewise, \cite{barber2019knockoff} split their data to screen for potentially relevant variables before conducting inference over this reduced set when testing for associations in a high-dimensional linear model. Our work is also conceptually related to that of \cite{lei2018adapt}, who learn data-adaptive contrasts on a partially masked portion of the data to reduce the multiplicity burden and increase the power of their test, while preserving valid inference.

\section{Impacts of Poverty on Cardiovascular Disease Risk Factors} \label{sec: application}

\subsection{Data processing and preliminary exploration} \label{data-exam}

We aggregated pre-pandemic NHANES data between 1999 and 2016 and utilized publicly available demographic, dietary, examination, laboratory, and questionnaire data. In accordance with \cite{ali2011household}, we removed individuals who have been diagnosed with diabetes or are currently pregnant. We gathered variables related to personal demographics, health insurance coverage, and timing of participation in the study to identify and match similar individuals based on their measured covariates.
Before matching, we removed participants who were missing treatment information or any outcomes under investigation. We conducted nearest neighbor pair matching on the rank-based Mahalanobis distance with a propensity score caliper of $0.25$ standard deviations on the probit scale and obtained satisfactory covariate balance. %
We fixed our planning sample proportion to $0.25$ and partitioned the data into pilot and analysis samples.

Before proceeding with formal analysis, we conducted a comprehensive examination of the planning sample and consulted directly with domain specialists on our team. 
Together, we discussed our initial exploratory findings and formulated an appropriate analysis plan that was supported by our examinations of the data and the relevant literature. For example, in line with previous work by \cite{jackson2018income}, we initially sought to assess the effects of poverty on nutrition by constructing a proxy for the Healthy Eating Index (HEI) score, which registers overall adherence to the 2015 issue of Dietary Guidelines for Americans 
\citep{krebs2018update}. The HEI is divided into adequacy components, representing food groups and dietary elements that are encouraged, and moderation components, reflecting those for which there are recommended limits to consumption. Each component is assigned a standard to achieve a maximum score, and the components are weighted equally and summed to obtain the overall HEI score.
Although we observed minimal empirical evidence of a negative effect of poverty on the total HEI score in the pilot sample, we identified notable discrepancies in certain components, such as whole grains and fatty acids. This prompted us to continue to analyze a measure of dietary intake, but to construct a new, robust metric using our procedure to identify and quantify the combination of individual test statistics for each HEI component. Indeed, certain related works, such as \cite{hiza2013diet}, look to clarify the specific relationships between socioeconomic status in the pediatric population and the constituent components of the HEI score. We also originally considered a larger number of outcomes across our general categories of interest. We found on the pilot sample that certain outcomes, such as body-mass-index, diastolic BP, and creatinine, were highly correlated with other measured outcomes like waist-to-height ratio, systolic BP, and eGFR, respectively. This motivated us to maintain  the latter outcomes for investigation while removing redundant ones, since previous literature has suggested that these may be more clinically relevant to CVD prognosis \citep{kannel1971systolic, yoo2016waist}. On the other hand, we observed that certain outcomes, like vigorous and moderate leisure-time physical activity (LTPA), appeared to be differentially affected by poverty, so we continued to analyze these outcomes separately. Previous works have found that vigorous LTPA may have a greater influence on health outcomes compared to moderate activity \citep{owens2017case} and certain analyses consider them separately \citep{falese2021association}, while others pool them together \citep{nader2008moderate}. 
In addition, we discovered heterogeneity by age group and gender for certain outcomes. For instance, we observe that vigorous LTPA is lower among adolescents in poverty, yet this effect is more pronounced among girls compared to boys. 

\subsection{Data analysis}

We formulated age/gender-specific global null hypotheses by testing the intersection null for all individual outcomes recorded by matched subjects in that age/gender stratum. We also assessed age-specific global null hypotheses by pooling all matched subjects within each age stratum, and then considered an overall intersection null hypothesis that assumed no effect of poverty across age strata or their individual outcomes, using Fisher's method. Based on the logical structure of this problem, we devised our analysis as a problem of testing hypotheses on trees 
and used the inheritance procedure proposed by \cite{goeman2012inheritance} to exploit the logical structure encoded in our hypotheses to improve statistical power while controlling the family-wise error rate at level $\alpha=0.05$.

We ran our sensitivity analysis procedure and compared it with the full-sample approach of \cite{cohen2020multivariate} at multiple values of $\Gamma$. In addition to the exploratory insights provided by sample splitting,
we leveraged the pilot sample for numerous practical benefits. Based on the combination $\widehat{\boldsymbol{\lambda}}_\p$ estimated on the pilot sample, we chose not to test an individual hypothesis $j$ if $\widehat{{\lambda}}_{\p, j} = 0$. This reflected that, relative to other outcomes, the planning sample provided little evidence that this hypothesis would be robust to potential unmeasured confounding.
Similarly, for each individual and global null hypothesis, we searched on the planning sample over a class of $m$-statistics proposed by \cite{rosenbaum2007sensitivity} to identify the test statistic that yielded the largest sensitivity value. For each null hypothesis, we then employed this test statistic when conducting inference on the analysis sample. As full-sample sensitivity analysis does not enjoy the benefits of a pilot sample, the comparator used the original HEI proxy variable and $m$-statistic with Huber's $\psi$-function to test each hypothesis. 
For ease of comparison, we allowed the full-sample comparator to consider the same revised set of outcomes and stratified approach as our method, although both were refined during the planning stage of our analysis. We varied $\Gamma$ over the values in \{1.00, 1.25, 1.50, 2.00\} and illustrated our results in Table \ref{tab:rejections}. We include further details and additional results in the Supplementary Material.

\DeclareRobustCommand{\ultraboldcheck}{%
  \ooalign{%
    \hfil\ding{51}\hfil\cr
    \hfil\raisebox{0.3pt}{\ding{51}}\hfil\cr
    \hfil\raisebox{-0.3pt}{\ding{51}}\hfil\cr
  }%
}
\begin{table}[htbp]
\centering
\caption{Application of sample-splitting approach to study effects of poverty on cardiovascular disease risk factors using NHANES data. Cells marked with ``\ultraboldcheck'' denote that the particular global null hypothesis in the row was rejected by the method in the column at the given level of $\Gamma$, while those marked with ``\checkmark'' refer to rejections of individual hypothesis. The shaded column reflects proposed method given by \ref{p2}.}
\label{tab:rejections}
\fontsize{6.5}{7.8}\selectfont
\setlength{\tabcolsep}{5pt}
\renewcommand{\arraystretch}{0.85}

\resizebox{\textwidth-30pt}{!}{%
\begin{tabular}{lcccccccc}
\toprule
\noalign{\global\def\arraystretch{0.05}}
\rule{0pt}{0.01ex} & \rule{0pt}{0.01ex} & \rule{0pt}{0.01ex} & \rule{0pt}{0.01ex} & \rule{0pt}{0.01ex} & \rule{0pt}{0.01ex} & \rule{0pt}{0.01ex} & \rule{0pt}{0.01ex} & \rule{0pt}{0.01ex} \\
\noalign{\global\def\arraystretch{1}}
  & \multicolumn{2}{c}{\tiny$\mathbf{\Gamma = 1.00}$} & \multicolumn{2}{c}{\tiny$\mathbf{\Gamma = 1.25}$} & \multicolumn{2}{c}{\tiny$\mathbf{\Gamma = 1.50}$} & \multicolumn{2}{c}{\tiny$\mathbf{\Gamma = 2.00}$} \\
\noalign{\global\def\arraystretch{0.05}}
\rule{0pt}{0.01ex} & \rule{0pt}{0.01ex} & \rule{0pt}{0.01ex} & \rule{0pt}{0.01ex} & \rule{0pt}{0.01ex} & \rule{0pt}{0.01ex} & \rule{0pt}{0.01ex} & \rule{0pt}{0.01ex} & \rule{0pt}{0.01ex} \\
\noalign{\global\def\arraystretch{1}}
\cmidrule(lr){2-3} \cmidrule(lr){4-5} \cmidrule(lr){6-7} \cmidrule(lr){8-9}
\noalign{\global\def\arraystretch{0.05}}
\rule{0pt}{0.01ex} & \rule{0pt}{0.01ex} & \rule{0pt}{0.01ex} & \rule{0pt}{0.01ex} & \rule{0pt}{0.01ex} & \rule{0pt}{0.01ex} & \rule{0pt}{0.01ex} & \rule{0pt}{0.01ex} & \rule{0pt}{0.01ex} \\
\noalign{\global\def\arraystretch{1}}
\textbf{Method} & \textbf{Full} & \cellcolor{gray!25}\textbf{Split} & \textbf{Full} & \cellcolor{gray!25}\textbf{Split} & \textbf{Full} & \cellcolor{gray!25}\textbf{Split} & \textbf{Full} & \cellcolor{gray!25}\textbf{Split} \\
\noalign{\global\def\arraystretch{0.05}}
\rule{0pt}{0.01ex} & \rule{0pt}{0.01ex} & \rule{0pt}{0.01ex} & \rule{0pt}{0.01ex} & \rule{0pt}{0.01ex} & \rule{0pt}{0.01ex} & \rule{0pt}{0.01ex} & \rule{0pt}{0.01ex} & \rule{0pt}{0.01ex} \\
\noalign{\global\def\arraystretch{1}}
\midrule
\noalign{\global\def\arraystretch{0.05}}
\rule{0pt}{0.01ex} & \rule{0pt}{0.01ex} & \cellcolor{gray!25}\rule{0pt}{0.01ex} & \rule{0pt}{0.01ex} & \cellcolor{gray!25}\rule{0pt}{0.01ex} & \rule{0pt}{0.01ex} & \cellcolor{gray!25}\rule{0pt}{0.01ex} & \rule{0pt}{0.01ex} & \cellcolor{gray!25}\rule{0pt}{0.01ex} \\
\noalign{\global\def\arraystretch{1}}
\textbf{All Outcomes} & {\tiny\ultraboldcheck} & \cellcolor{gray!25}{\tiny\ultraboldcheck} & {\tiny\ultraboldcheck} & \cellcolor{gray!25}{\tiny\ultraboldcheck} & {\tiny\ultraboldcheck} & \cellcolor{gray!25}{\tiny\ultraboldcheck} &  & \cellcolor{gray!25}{\tiny\ultraboldcheck} \\
\noalign{\global\def\arraystretch{0.05}}
\rule{0pt}{0.01ex} & \rule{0pt}{0.01ex} & \cellcolor{gray!25}\rule{0pt}{0.01ex} & \rule{0pt}{0.01ex} & \cellcolor{gray!25}\rule{0pt}{0.01ex} & \rule{0pt}{0.01ex} & \cellcolor{gray!25}\rule{0pt}{0.01ex} & \rule{0pt}{0.01ex} & \cellcolor{gray!25}\rule{0pt}{0.01ex} \\
\noalign{\global\def\arraystretch{1}}
\specialrule{0.1pt}{0pt}{0pt}
\noalign{\global\def\arraystretch{0.05}}
\rule{0pt}{0.01ex} & \rule{0pt}{0.01ex} & \cellcolor{gray!25}\rule{0pt}{0.01ex} & \rule{0pt}{0.01ex} & \cellcolor{gray!25}\rule{0pt}{0.01ex} & \rule{0pt}{0.01ex} & \cellcolor{gray!25}\rule{0pt}{0.01ex} & \rule{0pt}{0.01ex} & \cellcolor{gray!25}\rule{0pt}{0.01ex} \\
\noalign{\global\def\arraystretch{1}}
\textbf{Ages 8-11} & {\tiny\ultraboldcheck} & \cellcolor{gray!25}{\tiny\ultraboldcheck} & {\tiny\ultraboldcheck} & \cellcolor{gray!25}{\tiny\ultraboldcheck} & {\tiny\ultraboldcheck} & \cellcolor{gray!25}{\tiny\ultraboldcheck} &  & \cellcolor{gray!25}{\tiny\ultraboldcheck} \\
\noalign{\global\def\arraystretch{0.05}}
\rule{0pt}{0.01ex} & \rule{0pt}{0.01ex} & \cellcolor{gray!25}\rule{0pt}{0.01ex} & \rule{0pt}{0.01ex} & \cellcolor{gray!25}\rule{0pt}{0.01ex} & \rule{0pt}{0.01ex} & \cellcolor{gray!25}\rule{0pt}{0.01ex} & \rule{0pt}{0.01ex} & \cellcolor{gray!25}\rule{0pt}{0.01ex} \\
\noalign{\global\def\arraystretch{1}}
\specialrule{0.1pt}{0pt}{0pt}
\noalign{\global\def\arraystretch{0.05}}
\rule{0pt}{0.01ex} & \rule{0pt}{0.01ex} & \cellcolor{gray!25}\rule{0pt}{0.01ex} & \rule{0pt}{0.01ex} & \cellcolor{gray!25}\rule{0pt}{0.01ex} & \rule{0pt}{0.01ex} & \cellcolor{gray!25}\rule{0pt}{0.01ex} & \rule{0pt}{0.01ex} & \cellcolor{gray!25}\rule{0pt}{0.01ex} \\
\noalign{\global\def\arraystretch{1}}
\textbf{Ages 12-17} & {\tiny\ultraboldcheck} & \cellcolor{gray!25}{\tiny\ultraboldcheck} &  & \cellcolor{gray!25}{\tiny\ultraboldcheck} &  & \cellcolor{gray!25}{\tiny\ultraboldcheck} &  & \cellcolor{gray!25} \\
\noalign{\global\def\arraystretch{0.05}}
\rule{0pt}{0.01ex} & \rule{0pt}{0.01ex} & \cellcolor{gray!25}\rule{0pt}{0.01ex} & \rule{0pt}{0.01ex} & \cellcolor{gray!25}\rule{0pt}{0.01ex} & \rule{0pt}{0.01ex} & \cellcolor{gray!25}\rule{0pt}{0.01ex} & \rule{0pt}{0.01ex} & \cellcolor{gray!25}\rule{0pt}{0.01ex} \\
\noalign{\global\def\arraystretch{1}}
\specialrule{0.1pt}{0pt}{0pt}
\noalign{\global\def\arraystretch{0.05}}
\rule{0pt}{0.01ex} & \rule{0pt}{0.01ex} & \cellcolor{gray!25}\rule{0pt}{0.01ex} & \rule{0pt}{0.01ex} & \cellcolor{gray!25}\rule{0pt}{0.01ex} & \rule{0pt}{0.01ex} & \cellcolor{gray!25}\rule{0pt}{0.01ex} & \rule{0pt}{0.01ex} & \cellcolor{gray!25}\rule{0pt}{0.01ex} \\
\noalign{\global\def\arraystretch{1}}
\textbf{Ages 8-11, Female} &  & \cellcolor{gray!25}{\tiny\ultraboldcheck} &  & \cellcolor{gray!25}{\tiny\ultraboldcheck} &  & \cellcolor{gray!25}{\tiny\ultraboldcheck} &  & \cellcolor{gray!25}{\tiny\ultraboldcheck} \\ \rule{0pt}{0pt} & \rule{0pt}{0pt} & \cellcolor{gray!25}\rule{0pt}{0pt} & \rule{0pt}{0pt} & \cellcolor{gray!25}\rule{0pt}{0pt} & \rule{0pt}{0pt} & \cellcolor{gray!25}\rule{0pt}{0pt} & \rule{0pt}{0pt} & \cellcolor{gray!25}\rule{0pt}{0pt} \vspace{-0.805em} \\
\hspace{1.2em}Cotinine &  & \cellcolor{gray!25}{\tiny$\checkmark$} &  & \cellcolor{gray!25}{\tiny$\checkmark$} &  & \cellcolor{gray!25}{\tiny$\checkmark$} &  & \cellcolor{gray!25}{\tiny$\checkmark$} \\
\hspace{1.2em}Waist-to-Height &  & \cellcolor{gray!25}{\tiny$\checkmark$} &  & \cellcolor{gray!25} &  & \cellcolor{gray!25} &  & \cellcolor{gray!25} \\
\hspace{1.2em}Diet &  & \cellcolor{gray!25} &  & \cellcolor{gray!25} &  & \cellcolor{gray!25} &  & \cellcolor{gray!25} \\
\hspace{1.2em}Non-HDL Cholesterol &  & \cellcolor{gray!25} &  & \cellcolor{gray!25} &  & \cellcolor{gray!25} &  & \cellcolor{gray!25} \\
\hspace{1.2em}Systolic BP &  & \cellcolor{gray!25} &  & \cellcolor{gray!25} &  & \cellcolor{gray!25} &  & \cellcolor{gray!25} \\
\noalign{\global\def\arraystretch{0.05}}
\rule{0pt}{0.01ex} & \rule{0pt}{0.01ex} & \cellcolor{gray!25}\rule{0pt}{0.01ex} & \rule{0pt}{0.01ex} & \cellcolor{gray!25}\rule{0pt}{0.01ex} & \rule{0pt}{0.01ex} & \cellcolor{gray!25}\rule{0pt}{0.01ex} & \rule{0pt}{0.01ex} & \cellcolor{gray!25}\rule{0pt}{0.01ex} \\
\noalign{\global\def\arraystretch{1}}
\specialrule{0.1pt}{0pt}{0pt}
\noalign{\global\def\arraystretch{0.05}}
\rule{0pt}{0.01ex} & \rule{0pt}{0.01ex} & \cellcolor{gray!25}\rule{0pt}{0.01ex} & \rule{0pt}{0.01ex} & \cellcolor{gray!25}\rule{0pt}{0.01ex} & \rule{0pt}{0.01ex} & \cellcolor{gray!25}\rule{0pt}{0.01ex} & \rule{0pt}{0.01ex} & \cellcolor{gray!25}\rule{0pt}{0.01ex} \\
\noalign{\global\def\arraystretch{1}}
\textbf{Ages 8-11, Male} & {\tiny\ultraboldcheck} & \cellcolor{gray!25}{\tiny\ultraboldcheck} &  & \cellcolor{gray!25}{\tiny\ultraboldcheck} &  & \cellcolor{gray!25}{\tiny\ultraboldcheck} &  & \cellcolor{gray!25}{\tiny\ultraboldcheck} \\ \rule{0pt}{0pt} & \rule{0pt}{0pt} & \cellcolor{gray!25}\rule{0pt}{0pt} & \rule{0pt}{0pt} & \cellcolor{gray!25}\rule{0pt}{0pt} & \rule{0pt}{0pt} & \cellcolor{gray!25}\rule{0pt}{0pt} & \rule{0pt}{0pt} & \cellcolor{gray!25}\rule{0pt}{0pt} \vspace{-0.805em} \\
\hspace{1.2em}Cotinine & {\tiny$\checkmark$} & \cellcolor{gray!25}{\tiny$\checkmark$} &  & \cellcolor{gray!25}{\tiny$\checkmark$} &  & \cellcolor{gray!25}{\tiny$\checkmark$} &  & \cellcolor{gray!25}{\tiny$\checkmark$} \\
\hspace{1.2em}Waist-to-Height &  & \cellcolor{gray!25} &  & \cellcolor{gray!25} &  & \cellcolor{gray!25} &  & \cellcolor{gray!25} \\
\hspace{1.2em}Diet &  & \cellcolor{gray!25} &  & \cellcolor{gray!25} &  & \cellcolor{gray!25} &  & \cellcolor{gray!25} \\
\hspace{1.2em}Non-HDL Cholesterol &  & \cellcolor{gray!25} &  & \cellcolor{gray!25} &  & \cellcolor{gray!25} &  & \cellcolor{gray!25} \\
\hspace{1.2em}Systolic BP &  & \cellcolor{gray!25} &  & \cellcolor{gray!25} &  & \cellcolor{gray!25} &  & \cellcolor{gray!25} \\
\noalign{\global\def\arraystretch{0.05}}
\rule{0pt}{0.01ex} & \rule{0pt}{0.01ex} & \cellcolor{gray!25}\rule{0pt}{0.01ex} & \rule{0pt}{0.01ex} & \cellcolor{gray!25}\rule{0pt}{0.01ex} & \rule{0pt}{0.01ex} & \cellcolor{gray!25}\rule{0pt}{0.01ex} & \rule{0pt}{0.01ex} & \cellcolor{gray!25}\rule{0pt}{0.01ex} \\
\noalign{\global\def\arraystretch{1}}
\specialrule{0.1pt}{0pt}{0pt}
\noalign{\global\def\arraystretch{0.05}}
\rule{0pt}{0.01ex} & \rule{0pt}{0.01ex} & \cellcolor{gray!25}\rule{0pt}{0.01ex} & \rule{0pt}{0.01ex} & \cellcolor{gray!25}\rule{0pt}{0.01ex} & \rule{0pt}{0.01ex} & \cellcolor{gray!25}\rule{0pt}{0.01ex} & \rule{0pt}{0.01ex} & \cellcolor{gray!25}\rule{0pt}{0.01ex} \\
\noalign{\global\def\arraystretch{1}}
\textbf{Ages 12-17, Female} &  & \cellcolor{gray!25}{\tiny\ultraboldcheck} &  & \cellcolor{gray!25}{\tiny\ultraboldcheck} &  & \cellcolor{gray!25}{\tiny\ultraboldcheck} &  & \cellcolor{gray!25} \\ \rule{0pt}{0pt} & \rule{0pt}{0pt} & \cellcolor{gray!25}\rule{0pt}{0pt} & \rule{0pt}{0pt} & \cellcolor{gray!25}\rule{0pt}{0pt} & \rule{0pt}{0pt} & \cellcolor{gray!25}\rule{0pt}{0pt} & \rule{0pt}{0pt} & \cellcolor{gray!25}\rule{0pt}{0pt} \vspace{-0.805em} \\
\hspace{1.2em}Cotinine &  & \cellcolor{gray!25}{\tiny$\checkmark$} &  & \cellcolor{gray!25}{\tiny$\checkmark$} &  & \cellcolor{gray!25}{\tiny$\checkmark$} &  & \cellcolor{gray!25} \\
\hspace{1.2em}Vigorous LTPA &  & \cellcolor{gray!25}{\tiny$\checkmark$} &  & \cellcolor{gray!25} &  & \cellcolor{gray!25} &  & \cellcolor{gray!25} \\
\hspace{1.2em}Waist-to-Height &  & \cellcolor{gray!25}{\tiny$\checkmark$} &  & \cellcolor{gray!25} &  & \cellcolor{gray!25} &  & \cellcolor{gray!25} \\
\hspace{1.2em}eGFR &  & \cellcolor{gray!25} &  & \cellcolor{gray!25} &  & \cellcolor{gray!25} &  & \cellcolor{gray!25} \\
\hspace{1.2em}Diet &  & \cellcolor{gray!25} &  & \cellcolor{gray!25} &  & \cellcolor{gray!25} &  & \cellcolor{gray!25} \\
\hspace{1.2em}HbA1c &  & \cellcolor{gray!25} &  & \cellcolor{gray!25} &  & \cellcolor{gray!25} &  & \cellcolor{gray!25} \\
\hspace{1.2em}Moderate LTPA &  & \cellcolor{gray!25} &  & \cellcolor{gray!25} &  & \cellcolor{gray!25} &  & \cellcolor{gray!25} \\
\hspace{1.2em}Non-HDL Cholesterol &  & \cellcolor{gray!25} &  & \cellcolor{gray!25} &  & \cellcolor{gray!25} &  & \cellcolor{gray!25} \\
\hspace{1.2em}Systolic BP &  & \cellcolor{gray!25} &  & \cellcolor{gray!25} &  & \cellcolor{gray!25} &  & \cellcolor{gray!25} \\
\noalign{\global\def\arraystretch{0.05}}
\rule{0pt}{0.01ex} & \rule{0pt}{0.01ex} & \cellcolor{gray!25}\rule{0pt}{0.01ex} & \rule{0pt}{0.01ex} & \cellcolor{gray!25}\rule{0pt}{0.01ex} & \rule{0pt}{0.01ex} & \cellcolor{gray!25}\rule{0pt}{0.01ex} & \rule{0pt}{0.01ex} & \cellcolor{gray!25}\rule{0pt}{0.01ex} \\
\noalign{\global\def\arraystretch{1}}
\specialrule{0.1pt}{0pt}{0pt}
\noalign{\global\def\arraystretch{0.05}}
\rule{0pt}{0.01ex} & \rule{0pt}{0.01ex} & \cellcolor{gray!25}\rule{0pt}{0.01ex} & \rule{0pt}{0.01ex} & \cellcolor{gray!25}\rule{0pt}{0.01ex} & \rule{0pt}{0.01ex} & \cellcolor{gray!25}\rule{0pt}{0.01ex} & \rule{0pt}{0.01ex} & \cellcolor{gray!25}\rule{0pt}{0.01ex} \\
\noalign{\global\def\arraystretch{1}}
\textbf{Ages 12-17, Male} &  & \cellcolor{gray!25}{\tiny\ultraboldcheck} &  & \cellcolor{gray!25}{\tiny\ultraboldcheck} &  & \cellcolor{gray!25}{\tiny\ultraboldcheck} &  & \cellcolor{gray!25} \\ \rule{0pt}{0pt} & \rule{0pt}{0pt} & \cellcolor{gray!25}\rule{0pt}{0pt} & \rule{0pt}{0pt} & \cellcolor{gray!25}\rule{0pt}{0pt} & \rule{0pt}{0pt} & \cellcolor{gray!25}\rule{0pt}{0pt} & \rule{0pt}{0pt} & \cellcolor{gray!25}\rule{0pt}{0pt} \vspace{-0.805em} \\
\hspace{1.2em}Cotinine &  & \cellcolor{gray!25}{\tiny$\checkmark$} &  & \cellcolor{gray!25}{\tiny$\checkmark$} &  & \cellcolor{gray!25} &  & \cellcolor{gray!25} \\
\hspace{1.2em}Vigorous LTPA &  & \cellcolor{gray!25}{\tiny$\checkmark$} &  & \cellcolor{gray!25} &  & \cellcolor{gray!25} &  & \cellcolor{gray!25} \\
\hspace{1.2em}Waist-to-Height &  & \cellcolor{gray!25} &  & \cellcolor{gray!25} &  & \cellcolor{gray!25} &  & \cellcolor{gray!25} \\
\hspace{1.2em}eGFR &  & \cellcolor{gray!25} &  & \cellcolor{gray!25} &  & \cellcolor{gray!25} &  & \cellcolor{gray!25} \\
\hspace{1.2em}Diet &  & \cellcolor{gray!25} &  & \cellcolor{gray!25} &  & \cellcolor{gray!25} &  & \cellcolor{gray!25} \\
\hspace{1.2em}HbA1c &  & \cellcolor{gray!25} &  & \cellcolor{gray!25} &  & \cellcolor{gray!25} &  & \cellcolor{gray!25} \\
\hspace{1.2em}Moderate LTPA &  & \cellcolor{gray!25} &  & \cellcolor{gray!25} &  & \cellcolor{gray!25} &  & \cellcolor{gray!25} \\
\hspace{1.2em}Non-HDL Cholesterol &  & \cellcolor{gray!25} &  & \cellcolor{gray!25} &  & \cellcolor{gray!25} &  & \cellcolor{gray!25} \\
\hspace{1.2em}Systolic BP &  & \cellcolor{gray!25} &  & \cellcolor{gray!25} &  & \cellcolor{gray!25} &  & \cellcolor{gray!25} \\
\noalign{\global\def\arraystretch{0.05}}
\rule{0pt}{0.01ex} & \rule{0pt}{0.01ex} & \cellcolor{gray!25}\rule{0pt}{0.01ex} & \rule{0pt}{0.01ex} & \cellcolor{gray!25}\rule{0pt}{0.01ex} & \rule{0pt}{0.01ex} & \cellcolor{gray!25}\rule{0pt}{0.01ex} & \rule{0pt}{0.01ex} & \cellcolor{gray!25}\rule{0pt}{0.01ex} \\
\noalign{\global\def\arraystretch{1}}
\bottomrule
\end{tabular}%
}

\end{table}

As expected, both the split-sample and full-sample procedures rejected the largest number of individual and global null hypotheses at $\Gamma=1$. %
Each method identified significant overall differences in CVD risk profiles between children and adolescents below and above the poverty line, as reflected by the rejections of the global null hypothesis for each age stratum. However, the full-sample comparator could only reject the intersection null for boys between ages 8 and 11, while our method rejected all global nulls for each age/gender stratum. Additionally, the full-sample approach could only reject one individual null hypothesis, finding a significant difference in cotinine levels between young boys below and above the poverty line. On the other hand, our approach found evidence of this worsening tobacco exposure in every age/gender stratum.
This finding is consistent with structural accounts linking poverty to tobacco use and
secondhand smoke exposure through reduced access to health care and cessation
resources \citep{marbin2021health}.
Our method also discovered significant decreases in vigorous LTPA for both adolescent boys and girls in poverty, which warrants
concern given the well-established links between childhood physical activity,
fitness, obesity, and later-life health \citep{stodden2008developmental,owens2017case}. 
We also found a higher waist-to-height ratio in the exposed group for both younger and older girls.
This finding shows how growing up in poverty can exacerbate health disparities early in life and perpetuate a cycle of poor health in girls, with potential later-life health consequences 
\citep{kelsey2014age}. Although our new diet composite outcome was not statistically significant in any stratum after correcting for multiple comparisons, most of the individual $p$-values for this new composite were markedly lower than the standard HEI variable.%

One of our novel contributions
is acknowledging that these data are inherently observational in nature and explicitly addressing the robustness of these findings to potential unobserved confounding. Interestingly, we find that many of these results are not very resilient to potential biases. Allowing for a relatively small amount of potential unmeasured confounding at $\Gamma = 1.05$, we can no longer conclude a significant decrease in vigorous LTPA among older boys. When we set $\Gamma=1.15$, the only results we can establish among the individual hypotheses are increased cotinine levels among poorer children and adolescents across each age/gender stratum. As $\Gamma$ increases, the performance of the full-sample approach similarly degrades. At $\Gamma=1.15$, the comparator cannot reject any individual hypothesis or age/gender-specific global null hypothesis. On the other hand, our strategy identifies a harmful causal effect of poverty on CVD risk among 8 to 11-year-olds and increased tobacco exposure for both boys and girls through $\Gamma = 2.25$.  In fact, the sensitivity value for the root node, corresponding to the overall hypothesis of no effect of poverty on any outcome for any age or gender, using our procedure is approximately $3.05$. Although there is indeed some robustness of our conclusions to potential unmeasured biases, it appears to be largely driven by the outsized effect of growing up in poverty on increased tobacco exposure. If we had conducted a full-sample sensitivity analysis, the sensitivity value for the root node would only be around $1.67$. The robustness of elevated cotinine levels to potential unmeasured confounding highlights the importance of programs that aim to reduce tobacco exposure among lower-income children.%

\section{Reproducibility}
NHANES data is publicly available (\url{https://wwwn.cdc.gov/nchs/nhanes/}). An \texttt{R} package will be made available to implement our method and reproduce results.

\renewcommand{\refname}{References}
\putbib

\end{bibunit}

\orignewpage
\appendix

\begin{bibunit}

\spacingset{1.8}

\section{Additional Remarks} \label{sec: discussion}

\subsection{Data application}
We applied our method to investigate the effects of poverty on the emergence of CVD risk factors in children and adolescents. We discovered various age/gender specific effects of poverty on outcomes related to body composition, physical activity, and tobacco exposure. Our split-sample sensitivity analysis demonstrated substantially higher power compared to full-sample alternatives. However, many of these results, outside tobacco exposure, are not robust to potential unmeasured confounding, a novel finding in this line of research. We emphasize that our study design, shaped by the use of a pilot sample, yielded a markedly higher sensitivity value for the overall null hypothesis compared to the full-sample comparator. That is, by splitting our data to estimate an optimizing combination of test statistics and to conduct exploratory analysis, we greatly increased the robustness of our observational study to potential unmeasured confounding. Furthermore, our study 
provided causal interpretations on the effects of poverty in pediatric populations, which directly supports anti-poverty policy as contributing to CVD prevention. While many related works, such as \cite{ali2011household} and \cite{jackson2018income}, employ cutoffs to create binary indicators of CVD risk factors, we analyzed each variable over its full domain to quantify the underlying effect in its entirety, rather than reporting the proportion of children who cross an arbitrary clinical threshold. Finally, our analysis provided a framework for global null testing that, to the best of our knowledge, is unique in this line of research. Although individual findings can be clinically relevant depending on the focus of the study, many analyses wish to assess evidence of whether or not there is a broader effect of socioeconomic disadvantage on worsening CVD risk factors. Our new formulation of this problem may provide more clarity and detection power compared to previous works.

\subsection{A contextual counterexample}

\cite{fogarty2016sensitivity} propose a scenario in which testing the most promising individual hypothesis after conducting multiple sensitivity analyses, corresponding to solving the problem on the right-hand side, can have strictly worse power than the left-hand side. They remark that, for each outcome $k$, there can exist a worst-case confounder that minimizes evidence against $H_k$ such that each $H_k$ fails to reject, yet there may not exist a single unmeasured confounder for which all $H_k$ simultaneously fail to reject. For example, consider a more narrowly focused version of our study in which we examine the effects of poverty on cotinine and waist-to-height ratio in young girls. Both of these outcomes appear to be impacted on the pilot sample; see Table~\ref{tab:edatbl_summary_outcomes}. It could be the case that the observed difference in cotinine levels between girls with low FPIR and those out of poverty might fail to reject at some level $\Gamma$ if the worst-case unmeasured confounder is parental smoking, whereas the gap in waist-to-height ratio might separately fail to reject when the worst-case confounder reflects inherited cardiometabolic risk. However, in this scenario, parental smoking may not affect adiposity and inherited metabolic risk may not elevate cotinine levels. The global null could therefore be rejected for larger values of $\Gamma$ at which each marginal test would fail. In such situations, strict inequality in (\ref{minmax-inequality}) occurs: the right-hand side, which conducts multiple sensitivity analyses separately for each $k$ then corrects for multiple comparisons, yields an objective value smaller than the critical value, while solving directly for the left-hand side would lead to an objective value greater than the critical value and a rejection of the global null hypothesis.

\subsection{Methodological considerations}

Taking into account the advice of \cite{heller2009split} and \cite{small2024protocols}, which advocate splitting the data into a smaller planning sample for designing the study and a larger analysis sample for making inferences, we randomly partitioned our NHANES data so that we could gain exploratory insights from the data and still conduct a valid sensitivity analysis of our results. It is worth mentioning that previous frameworks also implicitly recognize the benefits of peeking at the data before testing, the difference being that they maintain the whole sample for testing and use corrections for multiple comparisons to account for their data-driven choice of test statistic.

This novel approach for sensitivity analysis both facilitates the development of flexible, data-informed designs for observational studies and, in many regimes, improves statistical power when evaluating causal theories under concerns of unmeasured confounding. Numerous remarks are now in order.

We start by randomly splitting the sample, so that we can estimate a maximizing linear combination on the planning sample in Step 2 and plug it into the analysis sample for testing in Step 3.  Often times, the procedure would be used to find the sensitivity value on the analysis sample; that is, the largest value of $\Gamma$ for which $H_0$ fails to reject in Step 3. To do so, in Step 2, we typically maximize the sensitivity value for the optimization problem on the planning sample, $\Gamma_\p$. This can be computed by a one-dimensional search over $\Gamma$ (e.g., bisection on a grid). We recommend comparing the square root of the objective value of this problem to the standard Normal distribution. The critical value to which we compare this quantity affects the sensitivity value on the planning sample and may consequently have a downstream effect on an optimizing value $\widehat{{\lambda}}_\p$. While this estimated linear combination may vary depending on the reference distribution being used, we recommend the standard Normal since it provides the same stringency as will be faced when conducting inference on the analysis sample and is far less computationally expensive than iteratively computing quantiles of the $\bar{\chi}^2$ distribution. The directions for each individual alternative hypothesis can either be pre-specified and used to assess an elaborate causal theory in the spirit of \cite{cochran1965planning}, or can be modified after looking at the planning sample in an attempt to maximize formal discoveries on the analysis sample.
In our observational study, we sought to assess whether poverty has an adverse effect on the emergence of CVD risk factors, which motivated us to pre-specify the directions for each individual alternative hypothesis $H_k^c$. Similarly, we note that the overall procedure maintains its validity if the researcher wishes to adapt the values of $\alpha$ and $\Gamma$ to be used in Step 3 after examining $\Gamma_\p$. This may be useful if, for instance, the overall strength of signals appears strong on the pilot data, possibly leading the practitioner to increase $\Gamma$ to assess robustness to greater levels of potential unobserved biases, or very weak, which may lead her to increase the overall $\alpha$--level of her analysis in an attempt to reject $H_0$ at $\Gamma=1$. We remark that Step 3 is a typical univariate sensitivity analysis, which can be solved quickly by the separable algorithm \citep{gastwirth2000asymptotic}.

We point out that our approach considers linear combinations drawn from the probability simplex, unlike \cite{cohen2020multivariate} who draw directions over the nonnegative orthant. Due to the homogeneity of the objective function, this difference is inconsequential and returns the same objective value, yet this new formulation provides certain benefits in terms of technical convenience and interpretation of linear combinations. Notably, it is crucial to the performance of the overall procedure that a good estimate for a maximizing linear combination is rendered by the planning sample. If the pilot data are too small, particularly with respect to $K$, or are unrepresentative of the analysis sample, perhaps due to the presence of few, but very notable outliers in the sample, an estimated optimal direction will be unable to effectively combat adversarial treatment allocation and will not yield powerful inference. Similarly, an estimate may be substantially different from a true optimizing direction when the true objective is highly non-smooth. However, when combinations are estimated with reasonable accuracy, often due to a large enough planning sample, their individual components can lend meaningful scientific interpretation \citep{ye2022dimensions} and be used to construct principled composite measures; see an example of this in Section \ref{data-exam}. 

Finally, we note that, while our method is amenable to closed testing \citep{marcus1976closed} for smaller values of $K$, it quickly becomes computationally infeasible to search through all possible intersection hypotheses, as the burden grows exponentially in $K$. We used a hierarchical testing method proposed by \cite{goeman2012inheritance} in our data application to conduct powerful testing of both global nulls and individual nulls. This allowed us to reject certain stratum-specific global null hypotheses in our study, for example, that growing up in poverty does not have a negative impact on the emergence of CVD risk factors in young girls, while also facilitating tests of its constituent individual nulls, such as assessing the effect of poverty on waist-to-height ratio in young girls.

\subsection{Future directions}
It would be interesting to study more recent NHANES waves of data collection and see how these impacts have evolved over time, especially after the COVID-19 pandemic. Notably, COVID-era expansions of the child tax credit and other supports contributed to a large decline in child poverty in 2021, yet its expiration the following year resulted in a marked rise in child poverty \citep{bitler2023effects}. It may also be of interest to examine potential mediation among the outcomes under investigation, for instance, whether the observed decrease in vigorous LTPA or dietary intake mediated the effect of poverty on subsequent increases in the waist-to-height ratio among girls. A future study may also look to decouple the effects of poverty on passive smoking versus active smoking among teens and young adults. It could also be worthwhile to examine possible heterogeneity in the effects of poverty across race/ethnicity.

One question in the area of multivariate sensitivity analysis that remains unresolved is how to quantify the closeness of particular linear combinations and what distance metrics would be most appropriate. New insights could lead to finite sample guarantees for the performance of our method and new means to evaluate competing methods. It could also be useful to extend our framework to conjunction or partial conjunction tests in the spirit of \cite{benjamini2008screening} and \cite{karmakar2020assessment} in order to assess the extent of corroboration of an elaborate causal theory, or to guarantee control of the false discovery rate \citep{benjamini1995controlling} in these studies with many hypotheses of interest.

\section{Technical Proofs} \label{appendix:proofs}

\begin{lemma}
\label{lem:scalar_positive_part_identity}
For every $a\in\mathbb R$ and every $V>0$,
\[
\frac{\max\{0,a\}^2}{V}
=
\sup_{s\ge 0}\left(2sa-s^2V\right).
\]
\end{lemma}

\begin{proof}
Fix $a\in\mathbb R$ and $V>0$. For $s\ge 0$, consider the function
$
s\longmapsto 2sa-s^2V.%
$
This is a quadratic polynomial in $s$ with coefficient $-V<0$ on $s^2$, so it
is strictly concave on $[0,\infty)$.

If $a\le 0$, then for every $s\ge 0$ we have
$
2sa-s^2V \le 0,
$
where  equality holds at $s=0$. Therefore
$
\sup_{s\ge 0}(2sa-s^2V)=0.
$
Since $a\le 0$, we also have $\max\{0,a\}=0$, so
$
\frac{\max\{0,a\}^2}{V}=0.
$
Thus the desired identity holds when $a\le 0$.

Now suppose $a>0$. The derivative of $2sa-s^2V$ with respect to $s$ is
$
2a-2sV.
$
Setting this equal to $0$ gives the unique critical point
$
s^\star=\frac{a}{V}.
$
Because $a>0$ and $V>0$, we have $s^\star>0$, so $s^\star$ belongs to the
feasible set $[0,\infty)$. Since the function is strictly concave, this critical
point is the unique maximizer. Evaluating at $s^\star$ yields
\[
2s^\star a-(s^\star)^2V
=
2\frac{a}{V}a-\left(\frac{a}{V}\right)^2V
=
\frac{a^2}{V}.
\]
Since $a>0$, we have $\max\{0,a\}=a$, so
$
\frac{\max\{0,a\}^2}{V}=\frac{a^2}{V}.
$
Thus the desired identity also holds when $a>0$.
\end{proof}

\begin{lemma}
\label{lem:cone_reformulation}
Let $\Lambda\subseteq\mathbb R^K$ be a nonempty convex set. Define its
conic hull as
\[
K_\Lambda := \{s\lambda : \lambda\in\Lambda,\ s\ge 0\}.
\]
Then for every $\varrho\in P_\Gamma$,
\[
\sup_{\lambda\in\Lambda} F(\lambda,\varrho)
=
\sup_{B\in K_\Lambda}
\left\{
2B^\top(t-\mu(\varrho))-B^\top\Sigma(\varrho)B
\right\}.
\]
\end{lemma}

\begin{proof}
Fix $\varrho\in P_\Gamma$. We first prove that for every fixed
$\lambda\in\Lambda$,
\[
F(\lambda,\varrho)
=
\sup_{s\ge 0}
\left\{
2s\,\lambda^\top(t-\mu(\varrho))
-
s^2\,\lambda^\top\Sigma(\varrho)\lambda
\right\}.
\]
If $\lambda=0$, then by definition $F(0,\varrho)=0$. Also,
\[
2s\,0^\top(t-\mu(\varrho))-s^2\,0^\top\Sigma(\varrho)0=0
\qquad\text{for every } s\ge 0,
\]
so the supremum on the right-hand side is also $0$. Thus the identity holds
when $\lambda=0$.

Now suppose $\lambda\neq 0$.

If $\lambda^\top\Sigma(\varrho)\lambda>0$, then Lemma
\ref{lem:scalar_positive_part_identity} applied with
$a=\lambda^\top(t-\mu(\varrho))$ and $V=\lambda^\top\Sigma(\varrho)\lambda$
gives
\[
\frac{\max\{0,\lambda^\top(t-\mu(\varrho))\}^2}
{\lambda^\top\Sigma(\varrho)\lambda}
=
\sup_{s\ge 0}
\left\{
2s\,\lambda^\top(t-\mu(\varrho))
-
s^2\,\lambda^\top\Sigma(\varrho)\lambda
\right\}.
\]
This is exactly the desired identity.

If $\lambda^\top\Sigma(\varrho)\lambda=0$, then for every $s\ge 0$,
\[
2s\,\lambda^\top(t-\mu(\varrho))
-
s^2\,\lambda^\top\Sigma(\varrho)\lambda
=
2s\,\lambda^\top(t-\mu(\varrho)).
\]
If $\lambda^\top(t-\mu(\varrho))>0$, then the supremum over $s\ge 0$ of the
right-hand side is $+\infty$, which agrees with the definition of
$F(\lambda,\varrho)$. If $\lambda^\top(t-\mu(\varrho))\le 0$, then the right-hand
side is at most $0$ for every $s\ge 0$, and equals $0$ at $s=0$, so its
supremum is $0$, again agreeing with the definition of $F(\lambda,\varrho)$.
Thus the desired identity holds for every $\lambda\in\Lambda$.

We now take the supremum over $\lambda\in\Lambda$. Recall we have 
$
\{s\lambda : \lambda\in\Lambda,\ s\ge 0\} := K_\Lambda.$ Therefore,
\begin{align*}
\sup_{\lambda\in\Lambda}F(\lambda,\varrho)
&=
\sup_{\lambda\in\Lambda}\sup_{s\ge 0}
\left\{
2s\,\lambda^\top(t-\mu(\varrho))
-
s^2\,\lambda^\top\Sigma(\varrho)\lambda
\right\}\\
&=
\sup_{B\in K_\Lambda}
\left\{
2B^\top(t-\mu(\varrho))-B^\top\Sigma(\varrho)B
\right\}.
\end{align*}
This proves the lemma.
\end{proof}

\noindent \textit{Theorem \ref{minmax-thm}:} Let $\Lambda \subseteq \mathbb R^K$ be a nonempty convex set with conic hull
$K_\Lambda$. Then, %

\noindent\textup{(i)} The function
$
\varrho\longmapsto \sup_{\lambda\in\Lambda}F(\lambda,\varrho)
$
is convex and lower semicontinuous on $P_\Gamma$.

\noindent\textup{(ii)} The minimum
$
\min_{\varrho\in P_\Gamma}\sup_{\lambda\in\Lambda}F(\lambda,\varrho)
$
is attained.

\noindent\textup{(iii)} There is an exact minimax equality
\begin{align*}
\min_{\varrho\in P_\Gamma}\sup_{\lambda\in\Lambda}F(\lambda,\varrho)
&=
\min_{\varrho\in P_\Gamma}\sup_{B\in K_\Lambda}
\left\{
2B^\top(t-\mu(\varrho))-B^\top\Sigma(\varrho)B
\right\}\\
&=
\sup_{B\in K_\Lambda}\min_{\varrho\in P_\Gamma}
\left\{
2B^\top(t-\mu(\varrho))-B^\top\Sigma(\varrho)B
\right\}\\
&= 
\sup_{\lambda\in\Lambda}
\min_{\varrho\in P_\Gamma}
F(\lambda,\varrho).
\end{align*}

\begin{proof}
We begin with part (i). By Lemma \ref{lem:cone_reformulation},
\[
\sup_{\lambda\in\Lambda}F(\lambda,\varrho)
=
\sup_{B\in K_\Lambda}
\left\{
2B^\top(t-\mu(\varrho))-B^\top\Sigma(\varrho)B
\right\}
\qquad\text{for every } \varrho\in P_\Gamma.
\]
Fix $B\in K_\Lambda$. We show that the map
$
\varrho\longmapsto 2B^\top(t-\mu(\varrho))-B^\top\Sigma(\varrho)B
$
is convex and continuous on $P_\Gamma$.

The term $2B^\top t$ is constant in $\varrho$. The map
$\varrho\mapsto B^\top\mu(\varrho)$ is linear in $\varrho$, hence both convex
and continuous. By Lemma 1 of \cite{cohen2020multivariate}, the
map $\varrho\mapsto B^\top\Sigma(\varrho)B$ is concave, so its negative is
convex. It is also continuous because it is a polynomial in coordinates
$\varrho_{ij}$. Therefore
$
\varrho\longmapsto 2B^\top(t-\mu(\varrho))-B^\top\Sigma(\varrho)B
$
is convex and continuous for each fixed $B\in K_\Lambda$.

Recall that the pointwise supremum of any family of convex functions is convex. Likewise, the
pointwise supremum of any family of lower semicontinuous functions is lower
semicontinuous. Since each of the functions above is continuous, hence lower
semicontinuous, it follows that
$
\varrho\longmapsto \sup_{B\in K_\Lambda}
\left\{
2B^\top(t-\mu(\varrho))-B^\top\Sigma(\varrho)B
\right\}
$
is convex and lower semicontinuous. By Lemma
\ref{lem:cone_reformulation}, this is exactly
$
\varrho\longmapsto \sup_{\lambda\in\Lambda}F(\lambda,\varrho).
$
This proves part (i).

For part (ii), $P_\Gamma$ is compact by its definition in Section \ref{subsec:weighted-stats}. Also,  part (i) shows that
$
\varrho\longmapsto \sup_{\lambda\in\Lambda}F(\lambda,\varrho)
$
is lower semicontinuous. A lower semicontinuous function on a compact set
attains its minimum. Therefore, the minimum
$
\min_{\varrho\in P_\Gamma}\sup_{\lambda\in\Lambda}F(\lambda,\varrho)
$
is attained.

We now prove part (iii). The first equality is simply Lemma
\ref{lem:cone_reformulation} applied pointwise in $\varrho$, followed by taking
the minimum over $\varrho\in P_\Gamma$
\[
\min_{\varrho\in P_\Gamma}\sup_{\lambda\in\Lambda}F(\lambda,\varrho)
=
\min_{\varrho\in P_\Gamma}\sup_{B\in K_\Lambda}
\left\{
2B^\top(t-\mu(\varrho))-B^\top\Sigma(\varrho)B
\right\}.
\]

For the second equality, define
$
f(\varrho,B)
:=
2B^\top(t-\mu(\varrho))-B^\top\Sigma(\varrho)B.
$
Observe that, since we move from an extended-real ratio payoff to a finite
quadratic payoff, we can apply the standard version of Sion's minimax theorem
to the latter objective. We now verify the conditions of Sion's minimax theorem
on the domain $P_\Gamma \times K_\Lambda$.

$P_\Gamma$ is compact and convex by definition. $K_\Lambda$ is a cone, and in fact a convex cone: closure under nonnegative scaling is immediate from the definition. For convexity, take $s_1\lambda_1, s_2\lambda_2 \in K_\Lambda$ with $\lambda_1,\lambda_2 \in \Lambda$, $s_1,s_2 \ge 0$, and $\theta \in [0,1]$. If
$
\theta s_1 + (1-\theta)s_2 = 0,
$
then $\theta s_1 = (1-\theta)s_2 = 0$, so
$
\theta s_1 \lambda_1 + (1-\theta)s_2 \lambda_2 = 0,
$
which lies in $K_\Lambda$. Otherwise, let
$
S = \theta s_1 + (1-\theta)s_2 > 0.
$
Observe that
\[
\theta s_1 \lambda_1 + (1-\theta)s_2 \lambda_2
=
S\left(\frac{\theta s_1}{S}\lambda_1 + \frac{(1-\theta)s_2}{S}\lambda_2\right),
\]
where the term in parentheses is a convex combination of $\lambda_1$ and $\lambda_2$. Since $\Lambda$ is convex, that term lies in $\Lambda$. Therefore,
$
\theta s_1 \lambda_1 + (1-\theta)s_2 \lambda_2 \in K_\Lambda.
$

Next, fix $B\in K_\Lambda$. As shown above, the map $\varrho\mapsto f(\varrho,B)$ is convex and continuous on $P_\Gamma$. Then, fix $\varrho\in P_\Gamma$. We show that the map $B\mapsto f(\varrho,B)$ is concave and continuous on $K_\Lambda$. The term $B\mapsto 2B^\top(t-\mu(\varrho))$ is linear in $B$, hence both concave and continuous. The covariance matrix $\Sigma(\varrho)$ is positive semidefinite. Thus, the quadratic form $B\mapsto B^\top\Sigma(\varrho)B$ is convex in $B$, and its negative is concave. It is also continuous. Hence,
$
B\longmapsto 2B^\top(t-\mu(\varrho))-B^\top\Sigma(\varrho)B
$
is concave and continuous on $K_\Lambda$.

All conditions of Sion's minimax theorem are satisfied, so we can conclude that
\begin{align*}
&\min_{\varrho\in P_\Gamma}\sup_{B\in K_\Lambda}
\left\{
2B^\top(t-\mu(\varrho))-B^\top\Sigma(\varrho)B
\right\}\\
=
&\sup_{B\in K_\Lambda}\min_{\varrho\in P_\Gamma}
\left\{
2B^\top(t-\mu(\varrho))-B^\top\Sigma(\varrho)B
\right\}.
\end{align*}

For the third and final equality, take $\lambda \in \Lambda$ and $\varrho \in P_\Gamma$. We showed in Lemma~\ref{lem:cone_reformulation}
that
\[
F(\lambda,\varrho)
=
\sup_{s\ge 0}
\Bigl\{
2s\,\lambda^\top(t-\mu(\varrho))
-
s^2\,\lambda^\top\Sigma(\varrho)\lambda
\Bigr\}.
\]
For each fixed $s\ge0$, the map $\varrho \mapsto 2s\,\lambda^\top(t-\mu(\varrho))
-
s^2\,\lambda^\top\Sigma(\varrho)\lambda$ is convex and continuous on $P_\Gamma$, since $\mu(\varrho)$ is linear in $\varrho$ and
$\lambda^\top\Sigma(\varrho)\lambda$ is concave in $\varrho$.
For each fixed $\varrho$, the map $s\mapsto 2s\,\lambda^\top(t-\mu(\varrho))
-
s^2\,\lambda^\top\Sigma(\varrho)\lambda$ is a concave
quadratic polynomial on $[0,\infty)$.
Hence Sion's minimax theorem applies on the convex sets
$[0,\infty)$ and $P_\Gamma$, yielding
\begin{align*}
&\sup_{s\ge0}
\min_{\varrho\in P_\Gamma}
\Bigl\{2s\,\lambda^\top(t-\mu(\varrho))
-
s^2\,\lambda^\top\Sigma(\varrho)\lambda\Bigr\}\\
=
&\min_{\varrho\in P_\Gamma}
\sup_{s\ge0}
\Bigl\{2s\,\lambda^\top(t-\mu(\varrho))
-
s^2\,\lambda^\top\Sigma(\varrho)\lambda\Bigr\}.
\end{align*}

By definition of $F$, the right-hand side equals
$\min_{\varrho\in P_\Gamma} F(\lambda,\varrho)$.
Therefore,
\[
\sup_{s\ge0}
\min_{\varrho\in P_\Gamma}
\Bigl\{
2s\,\lambda^\top(t-\mu(\varrho))
-
s^2\,\lambda^\top\Sigma(\varrho)\lambda
\Bigr\}
=
\min_{\varrho\in P_\Gamma}
F(\lambda,\varrho).
\]

Finally, we have by the definition of $K_\Lambda$ that
\begin{align*}
&\sup_{B\in K_\Lambda}
\min_{\varrho\in P_\Gamma}
\Bigl\{
2B^\top(t-\mu(\varrho))
-
B^\top\Sigma(\varrho)B
\Bigr\}\\
=
&\sup_{\lambda\in\Lambda}
\sup_{s\ge0}
\min_{\varrho\in P_\Gamma}
\Bigl\{2s\,\lambda^\top(t-\mu(\varrho))
-
s^2\,\lambda^\top\Sigma(\varrho)\lambda\Bigr\} \\
=
&\sup_{\lambda\in\Lambda}
\min_{\varrho\in P_\Gamma}
F(\lambda,\varrho),
\end{align*}
which confirms part (iii). This completes the proof.
\end{proof}

{\noindent \textit{Theorem \ref{convergence-thm} and Corollary \ref{corr2}:}
Before beginning the proof, we first introduce some notation. For each \(\lambda\in\Delta_+\), we define
$
T_i(\lambda):=\lambda^\top T_i
$ and $
\theta_i(\lambda):=\mathbb E[T_i(\lambda)].
$
We use \(\{p,a,w\}\) as shorthand notation to denote planning, analysis, and whole samples, respectively. For each \(s\in\{p,a,w\}\), let \(S_s\) be the set of matched-set indices assigned to sample \(s\) with cardinality \(I_s:=|S_s|\). We assume that the full matched-set random elements indexed by \(i\in S_s\) are independent across \(i\), though not necessarily identically distributed. In particular, \(T_i(\lambda)\) is a measurable function of these matched-set random elements. All arguments below may be read conditionally on the realized split; the random splitting is assumed independent of the matched-set random elements. Throughout, expectations are taken over the random treatment assignments and outcomes within matched sets, holding the matching configuration and random sample partition fixed. We also take \(I_s\to\infty\) for each \(s\in\{p,a,w\}\) throughout.  Define the sample averages
$
T_s(\lambda):=\frac1{I_s}\sum_{i\in S_s}T_i(\lambda)
$ and $
\theta_s(\lambda):=\frac1{I_s}\sum_{i\in S_s}\theta_i(\lambda).
$ We also write
$
\bar T_s:=\frac1{I_s}\sum_{i\in S_s}T_i$ and $
\bar\theta_s:=\frac1{I_s}\sum_{i\in S_s}\mathbb E[T_i],
$
so that $T_s(\lambda)=\lambda^\top \bar T_s$ and $\theta_s(\lambda)=\lambda^\top \bar\theta_s$.

Fix \(\Gamma\ge 1\). Let \(\mu_{\Gamma,i}(\lambda)\) and \(\nu_{\Gamma,i}^2(\lambda)\) denote the worst-case mean and variance returned by the separable algorithm of \cite{gastwirth2000asymptotic} when applied to the scalar scores induced by \(\lambda\). We assume that \(\mu_{\Gamma,i}(\lambda)\) and \(\nu_{\Gamma,i}^2(\lambda)\) are likewise measurable functions of the matched-set random elements, and that the maps $(\omega,\lambda,\Gamma)\mapsto\mu_{\Gamma,i}(\lambda;\omega)$ and $(\omega,\lambda,\Gamma)\mapsto\nu_{\Gamma,i}^2(\lambda;\omega)$ are jointly measurable and continuous in $(\lambda,\Gamma)$ for each $\omega$. Define the random averages
$
\mu_{\Gamma,s}(\lambda)
:=
\frac1{I_s}\sum_{i\in S_s}\mu_{\Gamma,i}(\lambda)$ and $
\nu_{\Gamma,s}^2(\lambda)
:=
\frac1{I_s}\sum_{i\in S_s}\nu_{\Gamma,i}^2(\lambda).
$
We then introduce the corresponding deterministic averages
$
\bar\mu_{\Gamma,s}(\lambda)
:=
\frac1{I_s}\sum_{i\in S_s}\mathbb E[\mu_{\Gamma,i}(\lambda)]
$ and $
\bar\nu_{\Gamma,s}^2(\lambda)
:=
\frac1{I_s}\sum_{i\in S_s}\mathbb E[\nu_{\Gamma,i}^2(\lambda)].
$
For the limiting deterministic criterion, we take
$
\nu_\Gamma(\lambda):=\sqrt{\nu_\Gamma^2(\lambda)}.
$ Throughout, $\theta$, $\mu_\Gamma$, and $\nu_\Gamma^2$ denote deterministic population limits whose existence and regularity are postulated in Assumption~\textup{(A1)} below.

For each \(s\in\{p,a,w\}\), let \(c_s\) denote the fixed, finite critical value used in the sensitivity analysis on sample $s$. Throughout, \(c_s\) denotes the critical value on the standardized-deviate scale (i.e., if a procedure is written in terms of a squared statistic with squared-scale critical 
value \(d_s\), then \(c_s=\sqrt{d_s}\)). For example, the fixed-combination Gaussian test has 
\(c_s=z_{1-\alpha}\), while a chi-squared or chi-bar-squared test with squared-scale 
critical value \(d_s\) corresponds to \(c_s=\sqrt{d_s}\). For simplicity of notation, we write the critical value as \(c_s\), although our argumentation still applies
if \(c_s\) is replaced by a measurable critical-value function \(c_s(\lambda,\Gamma)\), possibly depending on the sample, after appropriate modifications.

Finally, let \(\mathcal I:=[\Gamma_L,\Gamma_U]\subset[1,\infty)\) be a compact interval over which sensitivity analysis is conducted. The results presented here concern sensitivity values restricted to \(\mathcal I\) -- if the unrestricted sensitivity values lie in \(\mathcal I\) with probability tending to one, then the same conclusions apply to the unrestricted sensitivity values. Throughout, we use the Euclidean norm on \(\mathbb R^K\) to state continuity and Lipschitz conditions on the finite-dimensional set \(\Delta_+\).%

For each \(s\in\{p,a,w\}\), let
\[
B_s^\nu
:=
\left\{
\inf_{(\lambda,\Gamma)\in\Delta_+\times\mathcal I}
\nu_{\Gamma,s}^2(\lambda)>0
\right\}.
\]
On \(B_s^\nu\), define \(\Psi_s(\lambda,\Gamma)\) by
\begin{align*}
\Psi_s(\lambda,\Gamma)
&:=
\frac{T_s(\lambda)-\theta_s(\lambda)}{\sqrt{\nu_{\Gamma,s}^2(\lambda)}}
+
\frac{\theta_s(\lambda)-\mu_{\Gamma,s}(\lambda)}{\sqrt{\nu_{\Gamma,s}^2(\lambda)}}
-
\frac{c_s}{\sqrt{I_s}}.
\end{align*}
Off \(B_s^\nu\), set \(\Psi_s(\lambda,\Gamma):=0\) for all
\((\lambda,\Gamma)\in\Delta_+\times\mathcal I\). All root and maximization
arguments below are carried out on events of probability tending to one (specifically, the event \(E_s\), as defined in the proof). On \(B_s^\nu\), for a fixed \(\lambda\), rejection at level \(\Gamma\) is 
equivalent to
\[
\frac{\sqrt{I_s}\bigl(T_s(\lambda)-\mu_{\Gamma,s}(\lambda)\bigr)}
{\sqrt{\nu_{\Gamma,s}^2(\lambda)}}
> c_s,
\]
which is equivalent to
$
\Psi_s(\lambda,\Gamma)>0.$ Moreover, define
\[
\Psi(\lambda,\Gamma)
:=
\frac{\theta(\lambda)-\mu_\Gamma(\lambda)}{\nu_\Gamma(\lambda)}.
\]

For each \(\lambda\in\Delta_+\), let \(\widetilde\Gamma(\lambda)\) denote the design sensitivity when using linear combination \(\lambda\). Since $c_s/\sqrt{I_s}\to 0$ as $I_s\to\infty$, the critical value vanishes asymptotically and the design sensitivity $\widetilde\Gamma(\lambda)$ is characterized as the unique solution to $\Psi(\lambda,\Gamma)=0$ in the population limit \citep{rosenbaum2004design}). The proof shows that \(\widetilde\Gamma\) is continuous on \(\Delta_+\) and attains its maximum. We write
$
\Gamma^\star:=\max_{\lambda\in\Delta_+}\widetilde\Gamma(\lambda)
$
and
$
\mathcal M:=\arg\max_{\lambda\in\Delta_+}\widetilde\Gamma(\lambda).
$ For each \(s\in\{p,a,w\}\), let \(\widehat\Gamma_s(\lambda)\) denote the sample sensitivity value for the fixed linear combination \(\lambda\), defined as the unique zero of
$
\Psi_s(\lambda,\Gamma)=0
$
on the event $E_s$ defined in the proof, and set to $\Gamma_L$ on $E_s^c$. Let \(\widehat\lambda_s\) be any measurable selection from
$
\arg\max_{\lambda\in\Delta_+}\widehat\Gamma_s(\lambda),
$
using a fixed deterministic tie-breaking rule whenever necessary.

We assume the following regularity conditions:

\noindent\textup{(A1)} {[Regularity for design sensitivity]}
Assume that there exist deterministic functions
$
\theta:\Delta_+\to\mathbb R,
$ as well as $
\mu:\Delta_+\times\mathcal I\to\mathbb R
$ and $
\nu^2:\Delta_+\times\mathcal I\to(0,\infty),
$
written as \(\mu_\Gamma(\lambda)\) and \(\nu_\Gamma^2(\lambda)\), such that for each \(s\in\{p,a,w\}\),
\begin{enumerate}[label=(\roman*)]
\item For every fixed \(\lambda\in\Delta_+\),
$
\theta_s(\lambda)\to\theta(\lambda)
$ as $I_s\to\infty;
$
\item For every fixed \((\lambda,\Gamma)\in\Delta_+\times\mathcal I\),
$
\bar\mu_{\Gamma,s}(\lambda)\to\mu_\Gamma(\lambda)
$ and $
\bar\nu_{\Gamma,s}^2(\lambda)\to\nu_\Gamma^2(\lambda)
$ as $I_s\to\infty;
$
\item Letting
$
g(\lambda,\Gamma):=\theta(\lambda)-\mu_\Gamma(\lambda),
$
the map \(g\) is continuous on \(\Delta_+\times\mathcal I\);
\item For every fixed \(\lambda\in\Delta_+\), the map
$
\Gamma\mapsto g(\lambda,\Gamma)
$
is strictly decreasing on \(\mathcal I\);

\item
$
\inf_{\lambda\in\Delta_+} g(\lambda,\Gamma_L)>0
$ and $
\sup_{\lambda\in\Delta_+} g(\lambda,\Gamma_U)<0.
$
\end{enumerate}

\noindent\textup{(A2)} {[Regularity for sensitivity values]}
Assume that there exists a common \(\delta\in(0,1]\) such that
\begin{enumerate}[label=(\roman*)]
\item \((\lambda,\Gamma)\mapsto \nu_\Gamma^2(\lambda)\) is continuous on \(\Delta_+\times\mathcal I\);
\item \(\nu_\Gamma^2(\lambda)>0\) for every \(\lambda\in\Delta_+\) and every \(\Gamma\in\mathcal I\);
\item For each \(s\in\{p,a,w\}\), with probability tending to one, for every \(\lambda\in\Delta_+\), the map
$
\Gamma\mapsto \Psi_s(\lambda,\Gamma)
$
is strictly decreasing on \(\mathcal I\);

\textit{Remark:} At a high level, this condition requires that increasing $\Gamma$ moves the standardized criterion $\Psi_s$ downward, so that the sample sensitivity value $\widehat\Gamma_s(\lambda)$ is uniquely defined. Although this is the sample analogue of \textup{(A1)(iv)}, it is not implied solely by monotonicity of $\mu_{\Gamma,s}(\lambda)$, as $\Psi_s$ also depends on $\Gamma$ through $\nu_{\Gamma,s}^2(\lambda)$.

\item There exist nonnegative random variables
$
M_i^\mu,\ M_i^\nu,\ L_i^\mu,\ L_i^\nu
$
such that, almost surely, for all \((\lambda,\Gamma),(\lambda',\Gamma')\in\Delta_+\times\mathcal I\),
\[
|\mu_{\Gamma,i}(\lambda)|\le M_i^\mu,
\qquad
|\nu_{\Gamma,i}^2(\lambda)|\le M_i^\nu,
\]
and
\[
|\mu_{\Gamma,i}(\lambda)-\mu_{\Gamma',i}(\lambda')|
\le
L_i^\mu\bigl(\|\lambda-\lambda'\|+|\Gamma-\Gamma'|\bigr),
\]
\[
|\nu_{\Gamma,i}^2(\lambda)-\nu_{\Gamma',i}^2(\lambda')|
\le
L_i^\nu\bigl(\|\lambda-\lambda'\|+|\Gamma-\Gamma'|\bigr),
\]
with
\[
\sup_s \frac1{I_s}\sum_{i\in S_s}\mathbb E\bigl[(M_i^\mu)^{1+\delta}\bigr]<\infty,
\qquad
\sup_s \frac1{I_s}\sum_{i\in S_s}\mathbb E\bigl[(M_i^\nu)^{1+\delta}\bigr]<\infty,
\]
\[
\sup_s \frac1{I_s}\sum_{i\in S_s}\mathbb E\bigl[(L_i^\mu)^{1+\delta}\bigr]<\infty,
\qquad
\sup_s \frac1{I_s}\sum_{i\in S_s}\mathbb E\bigl[(L_i^\nu)^{1+\delta}\bigr]<\infty.
\]

\textit{Remark:} A sufficient condition for the Lipschitz property of \((\lambda,\Gamma)\mapsto \nu_{\Gamma,i}^2(\lambda)\) is that, for each matched set \(i\), either the worst-case mean maximizer is unique at every \((\lambda,\Gamma)\in\Delta_+\times\mathcal I\), or whenever two candidates tie in worst-case mean at a point \((\lambda,\Gamma)\), they also yield the same candidate worst-case variance at that point. In that case, the tie-breaking rule used to define \(\nu_{\Gamma,i}^2(\lambda)\) does not create discontinuities, and the selected worst-case variance inherits Lipschitz regularity from the underlying candidate variance functions;

\item $
\sup_s \frac1{I_s}\sum_{i\in S_s}\mathbb E\|T_i\|^{1+\delta}<\infty.
$
\end{enumerate}

\medskip

Then,

\noindent\textup{(i)} For each \(s\in\{p,a,w\}\),
$
\sup_{\lambda\in\Delta_+}\widehat\Gamma_s(\lambda)
\xrightarrow{p}
\Gamma^\star.
$
Moreover, any measurable maximizer \(\widehat\lambda_s\in
\arg\max_{\lambda\in\Delta_+}\widehat\Gamma_s(\lambda)\) satisfies
$
\operatorname{dist}(\widehat\lambda_s,\mathcal M)\xrightarrow{p}0.
$

\noindent\textup{(ii)} The procedures \eqref{cohen-sensanalysis} and \eqref{p2} have the same design sensitivity.

\begin{proof} \text{  }

\medskip

\noindent\textit{\underline{Properties of the population design sensitivity $\widetilde\Gamma$.}}
We first establish some preliminary results which will be used to derive properties of $\widetilde\Gamma.$

Since \(F(c\lambda,\varrho)=F(\lambda,\varrho)\) for every \(c>0\), we can rescale any \(\lambda\in\Lambda^+ \setminus \{0\} \) to lie in \(\Delta_+\) without changing the value of \(F\). We therefore restrict our attention to \(\lambda\in\Delta_+\) for \eqref{cohen-sensanalysis}; procedure \eqref{p2} was formulated to explicitly confine the set of linear combinations to $\Delta_+$.

We now prove that the deterministic averages $\theta_s$, $\bar\mu_{\Gamma,s}$, and $\bar\nu_{\Gamma,s}^2$ converge uniformly to their limits $\theta$, $\mu_\Gamma$, and $\nu_\Gamma^2$. For \(\theta_s\), note that
$
\theta_s(\lambda)=\lambda^\top\left(\frac1{I_s}\sum_{i\in S_s}\mathbb E[T_i]\right).
$
Hence, for any \(\lambda,\lambda'\in\Delta_+\),
\[
|\theta_s(\lambda)-\theta_s(\lambda')|
\le
\|\lambda-\lambda'\|
\left\|\frac1{I_s}\sum_{i\in S_s}\mathbb E[T_i]\right\|.
\]
By Jensen's inequality and Assumption \textup{(A2)(v)},
\[
\left\|\frac1{I_s}\sum_{i\in S_s}\mathbb E[T_i]\right\|
\le
\frac1{I_s}\sum_{i\in S_s}\mathbb E\|T_i\|
\le
\left(\frac1{I_s}\sum_{i\in S_s}\mathbb E\|T_i\|^{1+\delta}\right)^{1/(1+\delta)},
\]
so the Lipschitz constants of \(\theta_s\) are uniformly bounded by a constant that does not depend on \(s\). The same uniform Lipschitz bound implies that the pointwise limit \(\theta\) is Lipschitz, hence uniformly continuous, on \(\Delta_+\). Since \(\theta_s(\lambda)\to\theta(\lambda)\) pointwise on the compact set \(\Delta_+\) by Assumption \textup{(A1)(i)}, a standard finite-net argument \citep[Ch. 5]{wainwright2019high} -- pointwise convergence at finitely many net points combined with the uniform Lipschitz bound to control the gap to arbitrary $\lambda$ -- yields
$
\sup_{\lambda\in\Delta_+}|\theta_s(\lambda)-\theta(\lambda)|\to 0.
$

For \(\bar\mu_{\Gamma,s}\) and \(\bar\nu_{\Gamma,s}^2\), Assumption \textup{(A2)(iv)} implies
\[
|\bar\mu_{\Gamma,s}(\lambda)-\bar\mu_{\Gamma',s}(\lambda')|
\le
\left(\frac1{I_s}\sum_{i\in S_s}\mathbb EL_i^\mu\right)
\bigl(\|\lambda-\lambda'\|+|\Gamma-\Gamma'|\bigr),
\]
\[
|\bar\nu_{\Gamma,s}^2(\lambda)-\bar\nu_{\Gamma',s}^2(\lambda')|
\le
\left(\frac1{I_s}\sum_{i\in S_s}\mathbb EL_i^\nu\right)
\bigl(\|\lambda-\lambda'\|+|\Gamma-\Gamma'|\bigr).
\]
By Hölder's inequality, the averages of \(\mathbb EL_i^\mu\) and \(\mathbb EL_i^\nu\) are uniformly bounded by constants that do not depend on \(s\). Since the Lipschitz constants are uniformly bounded in \(s\), the pointwise limits \(\mu\) and \(\nu^2\) are Lipschitz, hence uniformly continuous, on \(\Delta_+\times\mathcal I\). Since \(\bar\mu_{\Gamma,s}(\lambda)\to\mu_\Gamma(\lambda)\) and \(\bar\nu_{\Gamma,s}^2(\lambda)\to\nu_\Gamma^2(\lambda)\) pointwise by Assumption \textup{(A1)(ii)}, another finite-net argument on the compact set \(\Delta_+\times\mathcal I\) yields
\[
\sup_{(\lambda,\Gamma)\in\Delta_+\times\mathcal I}
|\bar\mu_{\Gamma,s}(\lambda)-\mu_\Gamma(\lambda)|\to 0,
\]
and
\[
\sup_{(\lambda,\Gamma)\in\Delta_+\times\mathcal I}
|\bar\nu_{\Gamma,s}^2(\lambda)-\nu_\Gamma^2(\lambda)|\to 0.
\]

Next, we define
$
g(\lambda,\Gamma):=\theta(\lambda)-\mu_\Gamma(\lambda).
$
For each fixed \(\lambda\in\Delta_+\), Assumption \textup{(A1)(v)} gives
$
g(\lambda,\Gamma_L)>0$ and $
g(\lambda,\Gamma_U)<0.
$
Since \(\Gamma\mapsto g(\lambda,\Gamma)\) is continuous on \(\mathcal I\) by Assumption \textup{(A1)(iii)}, the intermediate value theorem implies that there exists \(\widetilde\Gamma(\lambda)\in(\Gamma_L,\Gamma_U)\) such that
$
g(\lambda,\widetilde\Gamma(\lambda))=0.
$
Moreover, Assumption \textup{(A1)(iv)} implies that \(\Gamma\mapsto g(\lambda,\Gamma)\) is strictly decreasing on \(\mathcal I\), so this root is unique.

We next show that \(\lambda\mapsto \widetilde\Gamma(\lambda)\) is continuous and attains its maximum on \(\Delta_+\). To prove continuity, let \(\lambda_n\to\lambda\) in \(\Delta_+\). Since \(\mathcal I\) is compact, any subsequence of \(\widetilde\Gamma(\lambda_n)\) has a further subsequence converging to some \(\bar\Gamma\in\mathcal I\). By continuity of \(g\), we have that
$
0
=
\lim_{m\to\infty} g\bigl(\lambda_{n_m},\widetilde\Gamma(\lambda_{n_m})\bigr)
=
g(\lambda,\bar\Gamma).
$
Since \(\Gamma\mapsto g(\lambda,\Gamma)\) is strictly decreasing and has unique zero at \(\widetilde\Gamma(\lambda)\), it must hold that \(\bar\Gamma=\widetilde\Gamma(\lambda)\). Hence \(\widetilde\Gamma(\lambda_n)\to\widetilde\Gamma(\lambda)\). Since \(\DeltaP\) is compact and \(\widetilde\Gamma\) is continuous, \(\widetilde\Gamma\) attains its maximum. 
Then \(\Mcal\) is nonempty and compact, but need not be a singleton without further assumptions.

\medskip

\noindent\textit{\underline{Uniform stochastic convergence of the sample criterion $\Psi_s$ to its population limit $\Psi$.}}%

We first need two auxiliary positive constants from the population criterion $\Psi$. Assumptions \textup{(A1)(iii)}, \textup{(A2)(i)}, and \textup{(A2)(ii)} imply that \(\Psi\) is continuous on \(\Delta_+\times\mathcal I\). Since \(\Delta_+\) is compact and \(\Psi(\cdot,\Gamma_L)\) and \(\Psi(\cdot,\Gamma_U)\) are continuous, the extrema
$
\eta_L:=\inf_{\lambda\in\Delta_+}\Psi(\lambda,\Gamma_L)$ and $
\eta_U:=-\sup_{\lambda\in\Delta_+}\Psi(\lambda,\Gamma_U)
$
are attained. As \(g(\lambda,\Gamma_L)\) is uniformly bounded away from zero, \(g(\lambda,\Gamma_U)\) is uniformly bounded away from zero in the negative direction, and \(\nu_\Gamma(\lambda)\) is continuous and positive on the compact set \(\Delta_+\times\mathcal I\), it follows that
$
\eta_L>0$ and $
\eta_U>0.
$

Let us now consider \(T_s(\lambda)-\theta_s(\lambda)\). For each \(i\in S_s\), define
$
X_i(\lambda):=T_i(\lambda)-\theta_i(\lambda).
$
Then \(\mathbb EX_i(\lambda)=0\), and by linearity,
\[
|X_i(\lambda)-X_i(\lambda')|
\le
\|T_i-\mathbb ET_i\|\,\|\lambda-\lambda'\|.
\]
Also, since \(\Delta_+\) is compact, there exists \(C<\infty\) such that
\[
\sup_{\lambda\in\Delta_+}|X_i(\lambda)|
\le C \|T_i-\mathbb ET_i\|.
\]
Here and below, $C$ denotes a generic finite constant whose value may change from line to line. By the triangle inequality, the convexity bound
$(x+y)^{1+\delta}\le 2^\delta(x^{1+\delta}+y^{1+\delta})$ for
$x,y\ge 0$, and Jensen's inequality,
\[
\begin{aligned}
\mathbb E\|T_i-\mathbb ET_i\|^{1+\delta}
&\le
2^\delta\left\{
\mathbb E\|T_i\|^{1+\delta}
+
\|\mathbb ET_i\|^{1+\delta}
\right\}  \\
&\le
2^\delta\left\{
\mathbb E\|T_i\|^{1+\delta}
+
\mathbb E\|T_i\|^{1+\delta}
\right\}  \\
&=
2^{1+\delta}\mathbb E\|T_i\|^{1+\delta}
\end{aligned}
\]
for all \(i\). Hence Assumption \textup{(A2)(v)} yields
\[
\sup_s \frac1{I_s}\sum_{i\in S_s}\mathbb E\Bigl[\sup_{\lambda\in\Delta_+}|X_i(\lambda)|^{1+\delta}\Bigr]<\infty,
\]
and similarly the sample-average Lipschitz constants of \(X_i\) have uniformly bounded \((1+\delta)\)-moments.

Fix \(\varepsilon>0\). Let \(\{\lambda_1,\dots,\lambda_N\}\subset\Delta_+\) be a finite \(\rho\)-net. Then
\[
\sup_{\lambda\in\Delta_+}\left|\frac1{I_s}\sum_{i\in S_s}X_i(\lambda)\right|
\le
\max_{1\le m\le N}\left|\frac1{I_s}\sum_{i\in S_s}X_i(\lambda_m)\right|
+
\rho\,\frac1{I_s}\sum_{i\in S_s}\|T_i-\mathbb ET_i\|.
\]
For each fixed \(m\), the variables \(X_i(\lambda_m)\) are independent and mean zero. Since \(\delta\in(0,1]\), the von Bahr--Esseen inequality (\citeyear{von1965inequalities}) gives
\[
\mathbb E\left|
\frac1{I_s}\sum_{i\in S_s}X_i(\lambda_m)
\right|^{1+\delta}
\le
\frac{2}{I_s^{1+\delta}}\sum_{i\in S_s}\mathbb E|X_i(\lambda_m)|^{1+\delta}
\le
\frac{C}{I_s^\delta},
\]
for some constant \(C<\infty\) independent of \(s\) and \(m\). Hence
$
\frac1{I_s}\sum_{i\in S_s}X_i(\lambda_m)\xrightarrow{p}0
$
for each fixed \(m\), and therefore
$
\max_{1\le m\le N}\left|\frac1{I_s}\sum_{i\in S_s}X_i(\lambda_m)\right|\xrightarrow{p}0.
$
Also, by Jensen's inequality and Assumption \textup{(A2)(v)},
\[
\mathbb E\left[\left(\frac1{I_s}\sum_{i\in S_s}\|T_i-\mathbb ET_i\|\right)^{1+\delta}\right]
\le
\frac1{I_s}\sum_{i\in S_s}\mathbb E\|T_i-\mathbb ET_i\|^{1+\delta},
\]
which is uniformly bounded in $s$, hence
\[
\frac1{I_s}\sum_{i\in S_s}\|T_i-\mathbb ET_i\|=O_p(1)
\]
via Markov's inequality. Choosing \(\rho\) small and then taking \(I_s\) sufficiently large gives
$
\sup_{\lambda\in\Delta_+}|T_s(\lambda)-\theta_s(\lambda)|\xrightarrow{p}0.
$

Next, we consider
$
Y_i(\lambda,\Gamma):=\mu_{\Gamma,i}(\lambda)-\mathbb E[\mu_{\Gamma,i}(\lambda)].
$
By Assumption \textup{(A2)(iv)},
$
|Y_i(\lambda,\Gamma)|\le M_i^\mu+\mathbb EM_i^\mu,
$
and
\[
|Y_i(\lambda,\Gamma)-Y_i(\lambda',\Gamma')|
\le
\bigl(L_i^\mu+\mathbb EL_i^\mu\bigr)\bigl(\|\lambda-\lambda'\|+|\Gamma-\Gamma'|\bigr).
\]
Moreover, using \((x+y)^{1+\delta}\le 2^\delta(x^{1+\delta}+y^{1+\delta})\) for \(x,y\ge 0\) together with Jensen's inequality,
\[
\mathbb E|M_i^\mu+\mathbb EM_i^\mu|^{1+\delta}\le C_\mu\,\mathbb E[(M_i^\mu)^{1+\delta}],
\]
and
\[
\mathbb E|L_i^\mu+\mathbb EL_i^\mu|^{1+\delta}\le C_\mu'\,\mathbb E[(L_i^\mu)^{1+\delta}],
\]
for constants \(C_\mu,C_\mu'<\infty\) independent of \(i\). Thus, these centered uniform bounds and Lipschitz constants have uniformly bounded average \((1+\delta)\)-moments over each \(S_s\). Repeating the same finite-net argument on the compact set \(\Delta_+\times\mathcal I\), using von Bahr--Esseen inequality at each net point, yields
\[
\sup_{(\lambda,\Gamma)\in\Delta_+\times\mathcal I}
\left|
\mu_{\Gamma,s}(\lambda)-\bar\mu_{\Gamma,s}(\lambda)
\right|
\xrightarrow{p}0.
\]
By the deterministic uniform convergence already proved,
\[
\sup_{(\lambda,\Gamma)\in\Delta_+\times\mathcal I}
\left|
\bar\mu_{\Gamma,s}(\lambda)-\mu_\Gamma(\lambda)
\right|
\to 0.
\]
Therefore,
\[
\sup_{(\lambda,\Gamma)\in\Delta_+\times\mathcal I}
|\mu_{\Gamma,s}(\lambda)-\mu_\Gamma(\lambda)|
\xrightarrow{p}0.
\]

The same argument, with
$
Z_i(\lambda,\Gamma):=\nu_{\Gamma,i}^2(\lambda)-\mathbb E[\nu_{\Gamma,i}^2(\lambda)],
$
gives
\[
|Z_i(\lambda,\Gamma)|\le M_i^\nu+\mathbb EM_i^\nu,
\]
\[
|Z_i(\lambda,\Gamma)-Z_i(\lambda',\Gamma')|
\le
\bigl(L_i^\nu+\mathbb EL_i^\nu\bigr)\bigl(\|\lambda-\lambda'\|+|\Gamma-\Gamma'|\bigr),
\]
and consequently
\[
\sup_{(\lambda,\Gamma)\in\Delta_+\times\mathcal I}
\left|
\nu_{\Gamma,s}^2(\lambda)-\bar\nu_{\Gamma,s}^2(\lambda)
\right|
\xrightarrow{p}0.
\]
Using again the deterministic uniform convergence already proved,
\[
\sup_{(\lambda,\Gamma)\in\Delta_+\times\mathcal I}
|\nu_{\Gamma,s}^2(\lambda)-\nu_\Gamma^2(\lambda)|
\xrightarrow{p}0.
\]

Since \(\nu_\Gamma^2(\lambda)\) is continuous and strictly positive on the compact set \(\Delta_+\times\mathcal I\), it is bounded away from zero:
\[
\inf_{(\lambda,\Gamma)\in\Delta_+\times\mathcal I}\nu_\Gamma^2(\lambda)>0.
\]
Therefore, the preceding uniform convergence implies
\[
\inf_{(\lambda,\Gamma)\in\Delta_+\times\mathcal I}\nu_{\Gamma,s}^2(\lambda)>0
\qquad\text{with probability tending to one.}
\]

Combining the preceding uniform convergences and using \(c_s/\sqrt{I_s}\to0\), the continuous mapping theorem implies that
\[
\sup_{(\lambda,\Gamma)\in\Delta_+\times\mathcal I}
|\Psi_s(\lambda,\Gamma)-\Psi(\lambda,\Gamma)|
\xrightarrow{p}0.
\]

Hence, with probability tending to one,
\[
\sup_{(\lambda,\Gamma)\in\Delta_+\times\mathcal I}
|\Psi_s(\lambda,\Gamma)-\Psi(\lambda,\Gamma)|
<
\frac12\min\{\eta_L,\eta_U\}.
\]
On this event, for every \(\lambda\in\Delta_+\),
\[
\Psi_s(\lambda,\Gamma_L)
\ge
\Psi(\lambda,\Gamma_L)-\frac12\min\{\eta_L,\eta_U\}
\ge
\frac{\eta_L}{2}>0,
\]
and
\[
\Psi_s(\lambda,\Gamma_U)
\le
\Psi(\lambda,\Gamma_U)+\frac12\min\{\eta_L,\eta_U\}
\le
-\frac{\eta_U}{2}<0.
\]
Therefore, for each \(s\in\{p,a,w\}\), with probability tending to one,
$
\inf_{\lambda\in\Delta_+}\Psi_s(\lambda,\Gamma_L)>0$ and $
\sup_{\lambda\in\Delta_+}\Psi_s(\lambda,\Gamma_U)<0.
$

\medskip

\noindent\textit{\underline{Existence, uniqueness, and uniform consistency of the sample sensitivity value $\widehat\Gamma_s$.}} Let $E_s$ denote the high-probability event on which the sample criterion $\Psi_s$ admits a unique root in $\Gamma$ for every $\lambda \in \Delta_+$. Specifically, $E_s$ is the event that:
\begin{enumerate}[label=(\alph*)]
\item \(\inf_{(\lambda,\Gamma)\in\Delta_+\times\mathcal I}\nu_{\Gamma,s}^2(\lambda)>0\);
\item For every \(\lambda\in\Delta_+\), the map
$
\Gamma\mapsto \Psi_s(\lambda,\Gamma)
$
is strictly decreasing on \(\mathcal I\);
\item
$
\inf_{\lambda\in\Delta_+}\Psi_s(\lambda,\Gamma_L)>0
$ and $
\sup_{\lambda\in\Delta_+}\Psi_s(\lambda,\Gamma_U)<0.
$
\end{enumerate}
As a consequence of the preceding uniform convergence argument, along with the positivity of the limiting variance in Assumption~\textup{(A2)(ii)}, Assumption~\textup{(A2)(iii)}, and the endpoint sign separation in Assumption~\textup{(A1)(v)}, we deduce that $\mathbb P(E_s)\to 1$.

Moreover, on \(E_s\), the map
\((\lambda,\Gamma)\mapsto\Psi_s(\lambda,\Gamma)\) is continuous on
\(\Delta_+\times\mathcal I\). Indeed,
$
T_s(\lambda)-\theta_s(\lambda)
=
\lambda^\top(\bar T_s-\bar\theta_s)
$
is linear in \(\lambda\), and
\((\lambda,\Gamma)\mapsto\mu_{\Gamma,s}(\lambda)\) and
\((\lambda,\Gamma)\mapsto\nu_{\Gamma,s}^2(\lambda)\) are finite averages of
Lipschitz functions by Assumption~\textup{(A2)(iv)}, hence continuous. On
\(E_s\), positivity gives a denominator bounded away from zero, and therefore
\(\Psi_s\) is continuous.

On \(E_s\), for each \(\lambda\in\Delta_+\), the map \(\Gamma\mapsto\Psi_s(\lambda,\Gamma)\) is continuous and strictly decreasing on \(\mathcal I\), with opposite signs at \(\Gamma_L\) and \(\Gamma_U\). Hence the equation
$
\Psi_s(\lambda,\Gamma)=0
$
has a unique solution \(\widehat\Gamma_s(\lambda)\in(\Gamma_L,\Gamma_U)\). On the complement \(E_s^c\), define \(\widehat\Gamma_s(\lambda):=\Gamma_L\) for all \(\lambda\in\Delta_+\).

We now prove uniform root consistency. Fix \(\varepsilon>0\), and define
\[
A_\varepsilon
:=
\Bigl\{
(\lambda,\Gamma)\in\Delta_+\times\mathcal I:
|\Gamma-\widetilde\Gamma(\lambda)|\ge \varepsilon
\Bigr\}.
\]
If \(A_\varepsilon=\varnothing\), the desired conclusion is immediate. Otherwise, because \(\widetilde\Gamma\) is continuous and \(\Delta_+\times\mathcal I\) is compact, \(A_\varepsilon\) is compact. Moreover, by Assumption \textup{(A1)(iv)} and positivity of \(\nu_\Gamma(\lambda)\) from Assumption \textup{(A2)(ii)},
$
\Psi(\lambda,\Gamma)=0
$ if and only if $
\Gamma=\widetilde\Gamma(\lambda).
$
Hence \(\Psi\) does not vanish on \(A_\varepsilon\), so continuity of \(\Psi\) yields
$
\delta_\varepsilon
:=
\inf_{(\lambda,\Gamma)\in A_\varepsilon}|\Psi(\lambda,\Gamma)|
>0.
$
Therefore, on the event
\[
E_s
\cap
\left\{
\sup_{(\lambda,\Gamma)\in\Delta_+\times\mathcal I}
|\Psi_s(\lambda,\Gamma)-\Psi(\lambda,\Gamma)|
<
\delta_\varepsilon
\right\},
\]
no root of \(\Psi_s(\lambda,\cdot)\) can lie outside
$
(\widetilde\Gamma(\lambda)-\varepsilon,\widetilde\Gamma(\lambda)+\varepsilon).
$
It follows that
$
\sup_{\lambda\in\Delta_+}
|\widehat\Gamma_s(\lambda)-\widetilde\Gamma(\lambda)|
<\varepsilon.
$
Since \(\mathbb P(E_s)\to1\) and the displayed uniform convergence holds in probability, we conclude that
$
\sup_{\lambda\in\Delta_+}
|\widehat\Gamma_s(\lambda)-\widetilde\Gamma(\lambda)|
\xrightarrow{p}0$ for each $s \in\{p,a,w\}.
$

It remains to show that \(\lambda\mapsto\widehat\Gamma_s(\lambda)\) is continuous on \(\Delta_+\) with probability tending to one. Fix \(s\in\{p,a,w\}\), work on the event $E_s$, and let \(\lambda_n\to\lambda\). Because \(\mathcal I\) is compact, any subsequence of \(\widehat\Gamma_s(\lambda_n)\) has a further subsequence converging to some \(\bar\Gamma\in\mathcal I\). By continuity of \(\Psi_s\), we have
$
0
=
\lim_{m\to\infty}
\Psi_s(\lambda_{n_m},\widehat\Gamma_s(\lambda_{n_m}))
=
\Psi_s(\lambda,\bar\Gamma).
$
Since \(\Gamma\mapsto\Psi_s(\lambda,\Gamma)\) is strictly decreasing on \(\mathcal I\), this zero is unique, so \(\bar\Gamma=\widehat\Gamma_s(\lambda)\). Hence \(\widehat\Gamma_s(\lambda_n)\to\widehat\Gamma_s(\lambda)\). Thus \(\widehat\Gamma_s\) is continuous on \(\Delta_+\) on \(E_s\).

The joint measurability of \(\Psi_s\), its continuity in \((\lambda,\Gamma)\), and the uniqueness of the zero in \(\Gamma\) on \(E_s\) imply that
$
(\omega,\lambda)\mapsto \widehat\Gamma_s(\omega,\lambda)
$
is measurable on \(E_s\). Together with the definition \(\widehat\Gamma_s(\lambda)=\Gamma_L\) on \(E_s^c\), this gives a jointly measurable objective that is continuous in \(\lambda\) on the compact set \(\Delta_+\) for every sample realization. Hence the Weierstrass theorem guarantees existence of a maximizer. We take \(\widehat\lambda_s\) to be any measurable selection from
$
\arg\max_{\lambda\in\Delta_+}\widehat\Gamma_s(\lambda),
$
which exists by the measurable maximum theorem \citep[Theorem~18.19]{aliprantis2006infinite}. On \(E_s^c\), where \(\widehat\Gamma_s(\lambda)\equiv \Gamma_L\), take \(\widehat\lambda_s:=\mathbf e_1\).

For $s  \in \{p, a, w\}$, note that
\[
\left|
\sup_{\lambda\in\DeltaP}\widehat\Gamma_s(\lambda)
-
\sup_{\lambda\in\DeltaP}\widetilde\Gamma(\lambda)
\right|
\le
\sup_{\lambda\in\DeltaP}
\left|\widehat\Gamma_s(\lambda)-\widetilde\Gamma(\lambda)\right|
\overset{p}{\rightarrow}0.
\]
Therefore,
$
\sup_{\lambda\in\Delta_+}\widehat\Gamma_s(\lambda)
\xrightarrow{p}
\Gamma^\star.
$

We next establish that the distance between \(\widehat\lambda_s\) and \(\mathcal M\) converges to zero. Fix any \(\lambda^\star\in\mathcal M\). Since \(\widehat\lambda_s\) maximizes \(\widehat\Gamma_s\),
$
\widehat\Gamma_s(\widehat\lambda_s)
\ge
\widehat\Gamma_s(\lambda^\star).
$
Therefore,
\begin{align*}
\widetilde\Gamma(\widehat\lambda_s)
&\ge
\widehat\Gamma_s(\widehat\lambda_s)
-
\sup_{\lambda\in\Delta_+}
\left|\widehat\Gamma_s(\lambda)-\widetilde\Gamma(\lambda)\right| \\
&\ge
\widehat\Gamma_s(\lambda^\star)
-
\sup_{\lambda\in\Delta_+}
\left|\widehat\Gamma_s(\lambda)-\widetilde\Gamma(\lambda)\right| \\
&\ge
\widetilde\Gamma(\lambda^\star)
-
2\sup_{\lambda\in\Delta_+}
\left|\widehat\Gamma_s(\lambda)-\widetilde\Gamma(\lambda)\right| \\
&=
\Gamma^\star
-
2\sup_{\lambda\in\Delta_+}
\left|\widehat\Gamma_s(\lambda)-\widetilde\Gamma(\lambda)\right|.
\end{align*}
Since \(\widetilde\Gamma(\widehat\lambda_s)\le\Gamma^\star\), it follows that
$
0
\le
\Gamma^\star-\widetilde\Gamma(\widehat\lambda_s)
\le
2\sup_{\lambda\in\Delta_+}
\left|\widehat\Gamma_s(\lambda)-\widetilde\Gamma(\lambda)\right|
\xrightarrow{p}0.
$
Hence
$
\widetilde\Gamma(\widehat\lambda_s)
\xrightarrow{p}
\Gamma^\star.
$

Finally, fix \(\varepsilon>0\), and define
$
A_\varepsilon^\lambda
:=
\left\{\lambda\in\Delta_+:
\operatorname{dist}(\lambda,\mathcal M)\ge\varepsilon
\right\},$ where $
\operatorname{dist}(\lambda,\mathcal M)
:=
\inf_{m\in\mathcal M}\|\lambda-m\|.
$
If \(A_\varepsilon^\lambda=\varnothing\), the conclusion is immediate. If \(A_\varepsilon^\lambda\) is nonempty, it is compact and disjoint from \(\mathcal M\). Since \(\widetilde\Gamma\) is continuous and \(\mathcal M\) is its argmax set,
$
\sup_{\lambda\in A_\varepsilon^\lambda}\widetilde\Gamma(\lambda)
<
\Gamma^\star.
$
Thus there exists \(\kappa_\varepsilon>0\) such that
$
\sup_{\lambda\in A_\varepsilon^\lambda}\widetilde\Gamma(\lambda)
\le
\Gamma^\star-\kappa_\varepsilon.
$
Consequently,
\begin{align*}
\mathbb P\left\{\operatorname{dist}(\widehat\lambda_s,\mathcal M)\ge\varepsilon\right\}
&\le
\mathbb P\left\{
\widetilde\Gamma(\widehat\lambda_s)
\le
\Gamma^\star-\kappa_\varepsilon
\right\}
\to 0.
\end{align*}
Therefore,
$
\operatorname{dist}(\widehat\lambda_s,\mathcal M)\xrightarrow{p}0.
$
This proves \textup{(i)}.

\medskip

\noindent\textit{\underline{Equal design sensitivity of the whole-sample and split-sample procedures.}} By Theorem~\ref{minmax-thm} and Corollary~\ref{corr1}~\ref{minimax-eq-simplex}, the whole-sample procedure \eqref{cohen-sensanalysis} rejects at level $\Gamma$ if and only if $\sup_{\lambda\in\Delta_+}\Psi_w(\lambda,\Gamma)>0$. On the event $E_w$, for each fixed $\lambda$, the map $\Gamma\mapsto\Psi_w(\lambda,\Gamma)$ is continuous and strictly decreasing with unique zero $\widehat\Gamma_w(\lambda)$, so $\Psi_w(\lambda,\Gamma)>0$ if and only if $\Gamma<\widehat\Gamma_w(\lambda)$. Therefore, the sensitivity value of \eqref{cohen-sensanalysis} is
$
\sup_{\lambda\in\Delta_+}\widehat\Gamma_w(\lambda)
=
\widehat\Gamma_w(\widehat\lambda_w),
$
where the last equality follows from the definition of $\widehat\lambda_w$. Hence,
\[
\left|
\widehat\Gamma_w(\widehat\lambda_w)-\Gamma^\star
\right|
=
\left|
\sup_{\lambda\in\DeltaP}\widehat\Gamma_w(\lambda)
-
\sup_{\lambda\in\DeltaP}\widetilde\Gamma(\lambda)
\right|
\le
\sup_{\lambda\in\DeltaP}
\left|\widehat\Gamma_w(\lambda)-\widetilde\Gamma(\lambda)\right|
\xrightarrow{p}0.
\]
Thus, 
$
\widehat\Gamma_w(\widehat\lambda_w)
\xrightarrow{p}
\Gamma^\star.
$

The split-sample procedure \eqref{p2} estimates \(\widehat\lambda_p\) on the planning sample and evaluates the sensitivity value on the analysis sample, so its sensitivity value is \(\widehat\Gamma_a(\widehat\lambda_p)\). We have
\begin{align*}
\left|
\widehat\Gamma_a(\widehat\lambda_p)-\Gamma^\star
\right|
&\le
\left|
\widehat\Gamma_a(\widehat\lambda_p)-\widetilde\Gamma(\widehat\lambda_p)
\right|
+
\left|
\widetilde\Gamma(\widehat\lambda_p)-\Gamma^\star
\right| \\
&\le
\sup_{\lambda\in\DeltaP}
\left|\widehat\Gamma_a(\lambda)-\widetilde\Gamma(\lambda)\right|
+
\left|
\widetilde\Gamma(\widehat\lambda_p)-\Gamma^\star
\right|.
\end{align*}
The first term converges to zero in probability by uniform root consistency for the analysis sample. The second term converges to zero in probability by the preceding value-consistency conclusion applied to the planning-sample maximizer \(\widehat\lambda_p\). Hence,
$
\widehat\Gamma_a(\widehat\lambda_p)
\xrightarrow{p}
\Gamma^\star.
$

Therefore both sensitivity values converge in probability to \(\Gamma^\star\). Since, on events whose probabilities tend to one, rejection at a fixed \(\Gamma\in\Ical\) occurs if and only if the corresponding sensitivity value exceeds \(\Gamma\), it follows that, for every fixed \(\Gamma\in\mathcal I\) with
\(\Gamma<\Gamma^\star\), both procedures reject with probability tending to one,
while for every fixed \(\Gamma\in\mathcal I\) with
\(\Gamma>\Gamma^\star\), both procedures reject with probability tending to zero. Thus both procedures have design sensitivity
$
\Gamma^\star
=
\max_{\lambda\in\DeltaP}\widetilde\Gamma(\lambda)
$, and we are done.
\end{proof}

}

\clearpage

\section{Additional Simulations}

We emulate the setting described in Section \ref{sec: sims} of the main text, increasing the number of matched pairs $I$ from $300$ to $1000$. Results are provided in Figure \ref{sparse-I1000}. 

\begin{figure}[t!]
    \centering
    \includegraphics[width = \textwidth]{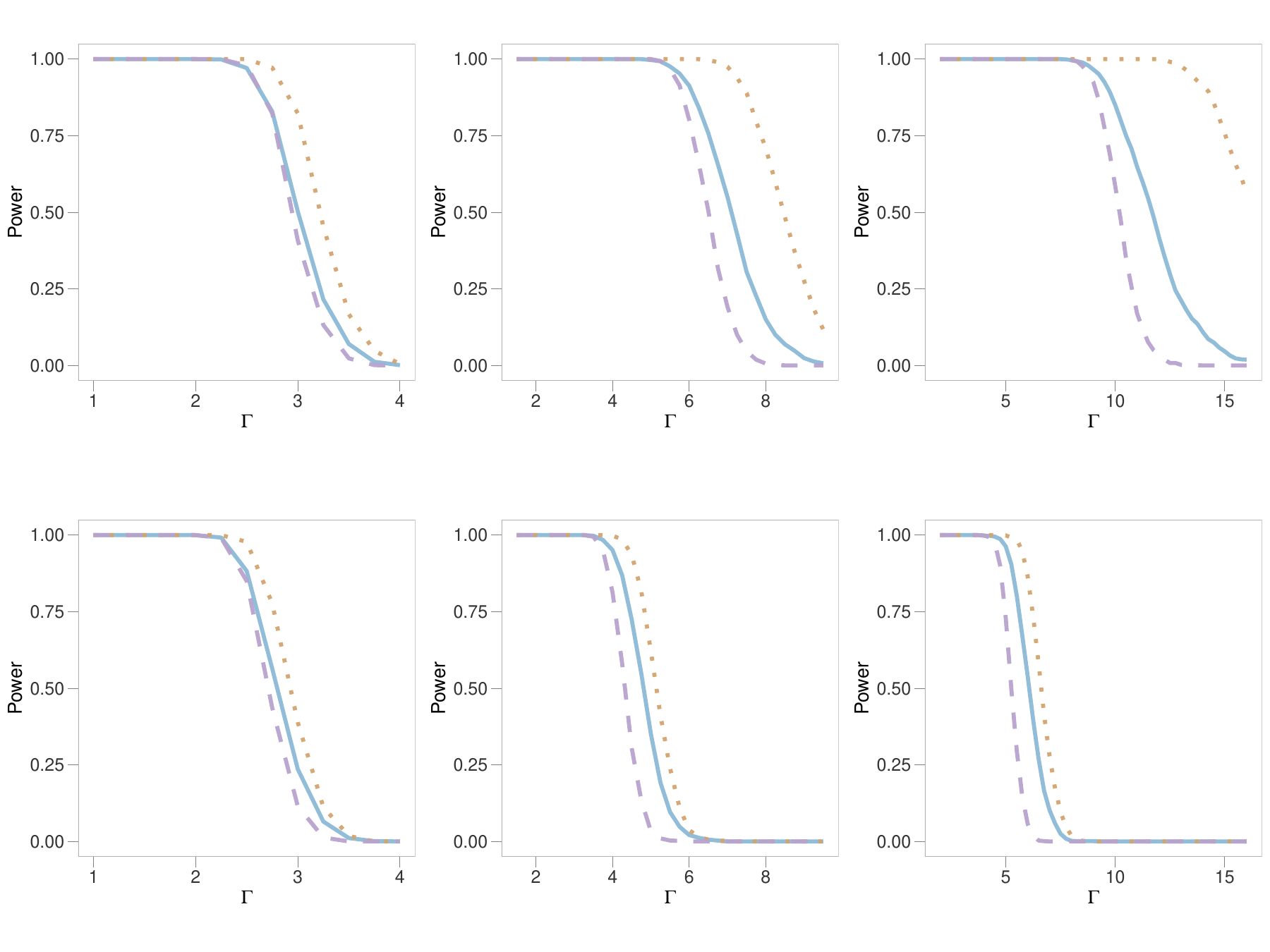}
    \caption{Power comparisons between the method of \cite{cohen2020multivariate} (dashed, purple), the sample splitting approach (solid, blue), and the oracle test (dotted, beige) as $\Gamma$ increases with $I = 1000$. The left column has $K = 5$; the center has $K = 15$; and the right has $K = 25$. The first row has $\rho = 0$ and the second row has $\rho = 0.2$.}
    \label{sparse-I1000}
\end{figure}

Comparing Figures \ref{sparse-I300} and \ref{sparse-I1000}, the difference in performance appears to diminish as $I$ grows larger. This finding is related to Theorem \ref{convergence-thm}, which shows that all three methods will eventually approach the same limiting value of design sensitivity.

\section{Details of Data Application} \label{appendix:data-app}

\subsection{Variable details}

We describe the variables from the NHANES datasets which were used for processing, defining the exposure, gathering confounders, and creating outcomes of interest. We list the corresponding variable name in the NHANES dataset in parentheses. Multiple variable names are listed for certain variables since some naming conventions changed over the course of the NHANES.\\

-- Data Processing:
\begin{itemize}
    \item Diagnosed with diabetes (\texttt{DIQ010})
    \item Currently pregnant (\texttt{RIDEXPRG})
\end{itemize}
\vspace{1em}
-- Exposure:
\begin{itemize}
    \item Ratio of family income to poverty (FPIR;  \texttt{INDFMPIR})
\end{itemize}
\vspace{1em}
-- Measured Confounders:
\begin{itemize}
    \item Six-month time period of examination (\texttt{RIDEXMON})
    \item Gender of respondent (\texttt{RIAGENDR})
    \item Age of respondent (\texttt{RIDAGEYR})
    \item Race/ethnicity of respondent (\texttt{RIDRETH1})
    \item Insurance status of respondent (\texttt{HID010, HIQ011})
\end{itemize}
\vspace{1em}
-- Outcome Variables:
\begin{itemize}
    \item Diet (discussed in further detail in Appendix~\ref{appendix:diet})
    \item Cotinine (\texttt{LBXCOT})
    \item Glycated hemoglobin (HbA1c; \texttt{LBXGH})
    \item Estimated glomerular filtration rate (eGFR)
    \begin{itemize}
        \item Measured via Schwartz bedside formula \citep{schwartz2009new}) using height (\texttt{BMXHT}) and creatinine (\texttt{LBXSCR, LBDSCR})
    \end{itemize}
    \item Moderate leisure-time physical activity (LTPA)
        \begin{itemize}
        \item Measured using frequency (\texttt{PADTIMES} in early waves, \texttt{PAQ670} in later waves) and duration (\texttt{PADDURAT} in early waves, \texttt{PAD675} in later waves)
    \end{itemize}
    \item Non-HDL cholesterol
        \begin{itemize}
        \item Measured as difference between \texttt{LBDTCSI} and either \texttt{LBDHDLSI} or \texttt{LBDHDDSI}
    \end{itemize}
    \item Systolic blood pressure (BP)
        \begin{itemize}
        \item Measured as average of three or four BP measurements (\texttt{BPXSY1, BPXSY2, BPXSY3, BPXSY4})
        \end{itemize}
    \item Vigorous leisure-time physical activity (LTPA)
        \begin{itemize}
        \item Measured using frequency (\texttt{PADTIMES} in early waves, \texttt{PAQ655} in later waves) and duration (\texttt{PADDURAT} in early waves, \texttt{PAD660} in later waves)
    \end{itemize}
    \item Waist-to-height ratio (\texttt{BMXWAIST / BMXHT})
\end{itemize}

\subsection{Covariate balance} \label{appendix:smdplt} Matching was done using the \texttt{MatchIt} package in R \citep{ho2011matchit}. Satisfactory covariate balance after matching can be seen in Figure~\ref{fig:SMDplot}.

\begin{figure}[H]
    \centering
    \includegraphics[width = \textwidth]{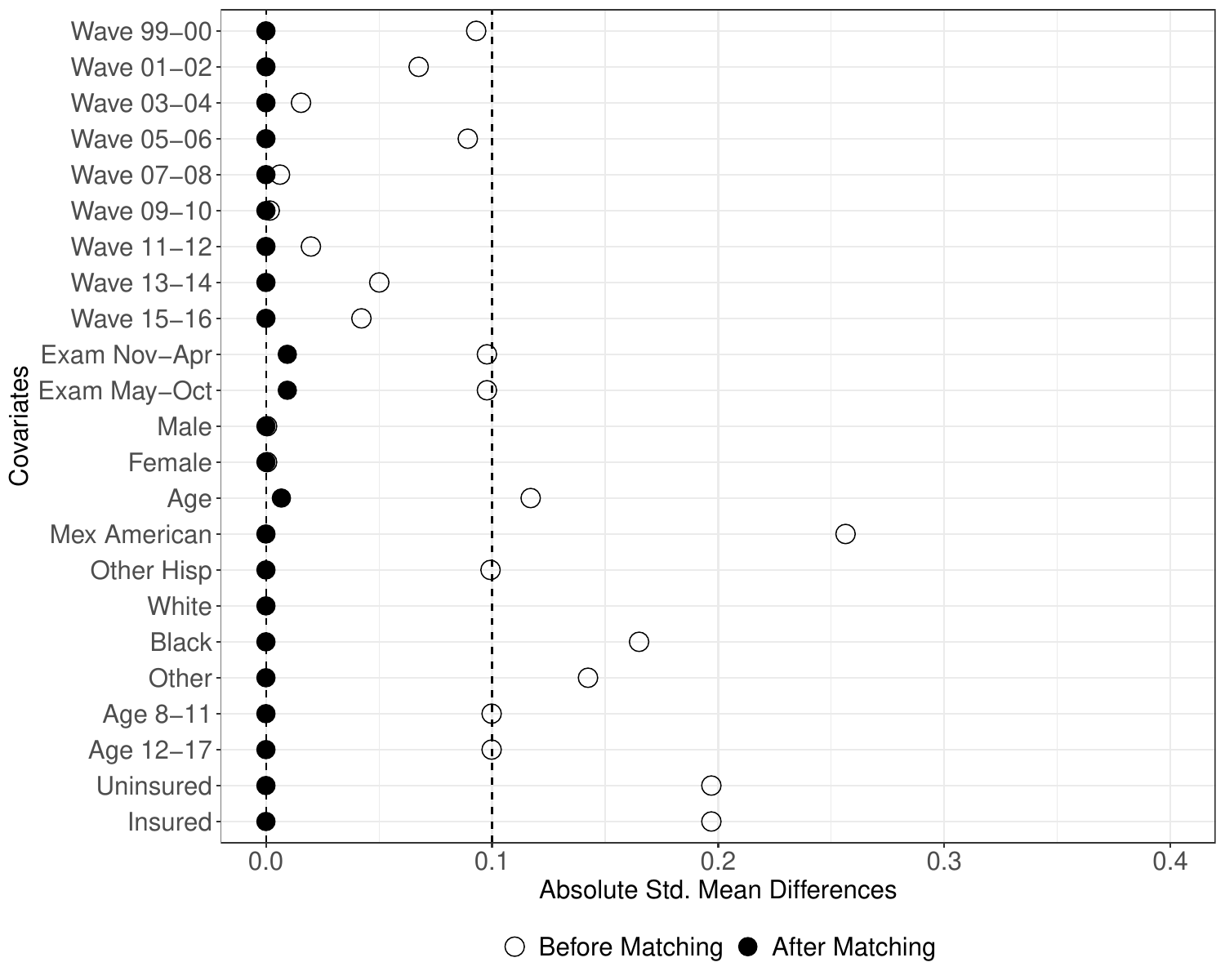}
    \caption{Love plot displaying the absolute standardized mean differences across covariates.}
    \label{fig:SMDplot}
\end{figure}

\subsection{Creating a robust nutrition composite} \label{appendix:diet}
We constructed a harmonized Healthy Eating Index--2015 (HEI-2015) measure from NHANES 1999--2016 Day~1 24-hour dietary recalls. Individual food records were linked to their corresponding food-pattern equivalency files (MPED for 1999--2004; FPED for 2005--2016) at the respondent--line level using respondent identifier and line number, aggregated to respondent-level cup/ounce equivalents, and converted to HEI-2015 component densities per 1{,}000~kcal using total Day~1 energy from the dietary totals file. Legumes were allocated between vegetable and protein components following HEI-2015 rules. In NHANES 1999--2004, MPED respondent--line linkages are incomplete and unmatched lines carry non-negligible energy; we therefore computed HEI only when $\ge 95\%$ of Day~1 line-level kcal matched to MPED, whereas in 2005--2016 FPED linkage gaps were rare and carried essentially zero energy in our QC diagnostics. The resulting variable (\texttt{HEI\_Total}, range 0--100) reflects overall adherence to the Dietary Guidelines for Americans, with higher scores indicating higher diet quality.

We constructed a line-level intake table containing \texttt{SEQN}, line number, food code, and per-line energy (\texttt{kcal\_line}), merged respondent--line food-pattern equivalents from MPED/FPED, and quantified join failures at the line level. A food line was classified as ``unmatched'' if all HEI mapping variables were missing after the merge. We computed (i) the unmatched line fraction and (ii) the unmatched energy fraction, defined as the fraction of total line-level energy sitting on unmatched lines. For matched lines, food-pattern equivalent variables were converted to numeric and missing values were set to zero (interpreting missing within a matched line as zero contribution for that component). We retained unmatched lines as missing at the line level, quantified their energy share, and restricted scoring in 1999–2004 to respondents with at least 95\% of Day-1 energy on matched lines. We aggregated to the respondent level by summing food-pattern equivalents across food lines and simultaneously computed:
(i) \texttt{kcal\_dr\_sum}, the sum of \texttt{kcal\_line} across all Day~1 lines,
(ii) \texttt{kcal\_mapped}, the sum of \texttt{kcal\_line} across matched lines only,
and (iii) counts of total and matched lines.
We then merged respondent-level nutrient totals from the Dietary Totals File (including total energy, sodium, saturated fat, MUFA, and PUFA) and retained the Day~1 dietary recall weight (\texttt{WTDRD1}) for later survey-weighted analyses.

For HEI density calculations (per 1{,}000~kcal), we used total Day~1 energy from the Dietary Totals File as the primary denominator, with a fallback to \texttt{kcal\_dr\_sum} when total energy was missing. Let $kcal$ denote the chosen total energy and define $per1000 = kcal/1000$. Nutrient-derived quantities were computed as:
(i) sodium density (g/1{,}000~kcal), $(\text{sodium}_{mg}/1000)/per1000$;
(ii) saturated fat as percent energy, $100 \times (9 \cdot \text{SFA}_{g})/kcal$;
and (iii) the fatty acid ratio $(\text{MUFA}_{g} + \text{PUFA}_{g})/\text{SFA}_{g}$, treated as missing when any fat input was missing or when all fat grams were zero.

Because MPED (1999--2004) and FPED (2005--2016) differ in variable names and sometimes component definitions, we harmonized a common set of food-group totals used for scoring via cycle-appropriate coalescing rules. For fruits, we used total fruit and total fruit juice where available, and otherwise used the MPED proxy for whole fruit (sum of citrus/melons/berries and other fruit). Whole fruit was computed as $\max(\text{Total fruit}-\text{Fruit juice},0)$ when both were available, and was capped not to exceed total fruit when both were observed. Analogous harmonization was performed for vegetables, dark green vegetables, whole and refined grains, dairy, and protein-related components. Added sugars were harmonized across cycles and converted to percent energy under the assumption that the added sugars variable is in teaspoons; we used 1~tsp $\approx 4.2$~g and 4~kcal/g (i.e., 16.8~kcal per tsp). HEI-2015 treats legumes (beans and peas) as counting toward total protein foods until the standard is met; remaining legumes count toward vegetables and contribute to the ``greens and beans'' component. For FPED cycles, we used explicit vegetable-legume (cup-equivalent) and protein-legume (oz-equivalent) variables. For MPED cycles, we used the MPED legume cup-equivalent variable and imputed protein-legume ounce-equivalents using the FPED conversion 1~cup legumes $=$ 4~oz-eq protein foods. We computed non-legume protein density using (i) FPED protein-food totals (which exclude legumes by design) when available, and (ii) a conservative MPED proxy based on meat/poultry/fish plus nuts and soy when a direct total was unavailable. Let $\text{Prot}_{\text{nonleg}}$ denote non-legume protein density (oz-eq/1{,}000~kcal) and $\text{Leg}_{\text{prot}}$ denote legume-as-protein density. We allocated legumes to protein up to the deficit relative to the HEI target (2.5~oz-eq/1{,}000~kcal):
$
\text{need} = \max\{2.5 - \text{Prot}_{\text{nonleg}}, 0\},$ and $
\text{leg\_as\_protein} = \min\{\text{Leg}_{\text{prot}}, \text{need}\}.
$
Remaining legumes were assigned to vegetables. For FPED cycles, total vegetables already exclude legumes, so we added vegetable-allocated legume remainder to total vegetables. For MPED cycles, where total vegetables may include legumes, we removed total legumes then added back the vegetable-allocated remainder. ``Greens and beans'' was computed as dark green vegetables plus vegetable-allocated legumes and was capped not to exceed total vegetables. We computed HEI-2015 component scores using standard cutpoints and linear interpolation. Adequacy components were scored proportionally from 0 to the maximum points and capped at the maximum; moderation components were scored using a piecewise linear rule that assigns maximum points at or below the ``good'' cutpoint and 0 points at or above the ``bad'' cutpoint. The total score was computed as the sum of 13 component scores, yielding \texttt{HEI\_Total} on a 0--100 scale.

Our primary estimand is the {mean of person-level HEI scores} computed from each respondent's Day~1 recall
$
\theta_{\text{ind}} \;=\; \mathbb{E}\!\left[\text{HEI}_i\right],
$
where $\text{HEI}_i$ is the HEI-2015 score computed from respondent $i$'s Day~1 intake.
This differs from the alternative {population-ratio} (or ``HEI of the average diet'') approach sometimes used for NHANES surveillance summaries, which first aggregates food groups and energy over individuals (typically with survey weights) and then computes a single HEI score from those aggregated densities. Because HEI is a nonlinear function of component densities with capping and piecewise linear scoring, $\mathbb{E}[\text{HEI}_i]$ generally differs from $\text{HEI}(\mathbb{E}[\text{densities}_i])$, and the direction of the difference depends on the distribution of intakes relative to component cutpoints. Our analysis uses the individual-first estimand because it aligns with person-level modeling and individual-level covariates/outcomes. A respondent was scored only if a valid energy denominator was available ($kcal>0$ and $per1000>0$). For 1999--2004 (MPED cycles), we additionally required high linkage coverage to avoid biased food-group densities when substantial energy fell on unmatched lines. Specifically, we computed a coherent mapping coverage metric using the same line-level energy system:
\[
\text{map\_coverage\_dr} \;=\; \frac{\texttt{kcal\_mapped}}{\texttt{kcal\_dr\_sum}},
\]
and set \texttt{score\_ok} to \texttt{TRUE} only when $\text{map\_coverage\_dr} \ge 0.95$.
If MUFA/PUFA/SFA inputs required for the fatty acid ratio were missing, the fatty acids component (and thus \texttt{HEI\_Total}) was set to missing by design rather than imputing nutrient values.

\subsection{HEI components on the pilot sample} \label{appendix:diettbl} Pilot sample summary statistics for HEI outcomes by FPIR stratum, age group, and gender. Outcomes are designated as either adequacy or moderation components of the HEI, or the total HEI score.  A higher value reflects closer alignment with recommendations set forth by the the Dietary Guidelines for Americans.

\setlength{\tabcolsep}{4pt}
\renewcommand{\arraystretch}{1}
\setlength{\arrayrulewidth}{0.2pt}
\setlength{\LTleft}{0pt}
\setlength{\LTright}{0pt}
\tiny
\begin{longtable}{@{\extracolsep{\fill}}lcccc|cccc@{}}
\label{tab:edatbl_summary_hei}\\
\toprule
& \multicolumn{4}{c}{\shortstack{FPIR $<$ 1\\(N=1006)}} & \multicolumn{4}{c}{\shortstack{FPIR $>$ 1\\(N=1006)}} \\
\rule{0pt}{2.0em} & & & & & & & & \\
Outcome/Strata & \shortstack{Lower\\Quartile} & Median & \shortstack{Upper\\Quartile} & Trimean & \shortstack{Lower\\Quartile} & Median & \shortstack{Upper\\Quartile} & Trimean \\
\midrule
\endfirsthead
\multicolumn{9}{l}{\textit{Table \thetable\ (continued)}}\\
\toprule
& \multicolumn{4}{c}{\shortstack{FPIR $<$ 1\\(N=1006)}} & \multicolumn{4}{c}{\shortstack{FPIR $>$ 1\\(N=1006)}} \\
\rule{0pt}{2.0em} & & & & & & & & \\
Outcome/Strata & \shortstack{Lower\\Quartile} & Median & \shortstack{Upper\\Quartile} & Trimean & \shortstack{Lower\\Quartile} & Median & \shortstack{Upper\\Quartile} & Trimean \\
\midrule
\endhead
\bottomrule
\endfoot
\bottomrule
\endlastfoot
\hspace*{0.8em}\textbf{Adequacy} & & & & & & & & \\
\rule{0pt}{0.35em} & & & & & & & & \\
\hspace*{0.8em}- \textit{Total Fruits} & & & & & & & & \\
\hspace*{2.2em}8-11 & 0.13 & 2.22 & 4.83 & 2.35 & 0.12 & 2.22 & 5.00 & 2.39 \\
\hspace*{4.4em}Female & 0.21 & 2.19 & 4.83 & 2.36 & 0.10 & 2.02 & 4.77 & 2.23 \\
\hspace*{4.4em}Male & 0.09 & 2.25 & 4.91 & 2.37 & 0.30 & 2.47 & 5.00 & 2.56 \\
\hspace*{2.2em}12-17 & 0.00 & 1.37 & 4.35 & 1.77 & 0.00 & 1.30 & 4.46 & 1.77 \\
\hspace*{4.4em}Female & 0.00 & 1.46 & 4.25 & 1.79 & 0.06 & 1.36 & 5.00 & 1.95 \\
\hspace*{4.4em}Male & 0.00 & 1.12 & 4.46 & 1.67 & 0.00 & 1.27 & 3.79 & 1.58 \\
\rule{0pt}{1.1em} & & & & & & & & \\
\hspace*{0.8em}- \textit{Whole Fruits} & & & & & & & & \\
\hspace*{2.2em}8-11 & 0.00 & 1.30 & 5.00 & 1.90 & 0.00 & 2.10 & 5.00 & 2.30 \\
\hspace*{4.4em}Female & 0.00 & 2.36 & 5.00 & 2.43 & 0.00 & 1.95 & 5.00 & 2.23 \\
\hspace*{4.4em}Male & 0.00 & 0.55 & 5.00 & 1.53 & 0.00 & 2.76 & 5.00 & 2.63 \\
\hspace*{2.2em}12-17 & 0.00 & 1.19 & 5.00 & 1.85 & 0.00 & 1.04 & 5.00 & 1.77 \\
\hspace*{4.4em}Female & 0.00 & 1.63 & 5.00 & 2.07 & 0.00 & 1.32 & 5.00 & 1.91 \\
\hspace*{4.4em}Male & 0.00 & 1.02 & 5.00 & 1.76 & 0.00 & 0.69 & 5.00 & 1.60 \\
\rule{0pt}{1.1em} & & & & & & & & \\
\hspace*{0.8em}- \textit{Total Vegetables} & & & & & & & & \\
\hspace*{2.2em}8-11 & 0.77 & 1.74 & 3.05 & 1.82 & 0.74 & 1.51 & 2.68 & 1.61 \\
\hspace*{4.4em}Female & 0.97 & 1.99 & 3.28 & 2.06 & 0.75 & 1.61 & 2.54 & 1.63 \\
\hspace*{4.4em}Male & 0.67 & 1.50 & 2.78 & 1.61 & 0.73 & 1.46 & 2.78 & 1.61 \\
\hspace*{2.2em}12-17 & 0.75 & 1.83 & 3.18 & 1.90 & 1.00 & 1.95 & 3.34 & 2.06 \\
\hspace*{4.4em}Female & 0.64 & 1.83 & 3.18 & 1.87 & 1.03 & 2.06 & 3.56 & 2.18 \\
\hspace*{4.4em}Male & 0.83 & 1.82 & 3.17 & 1.91 & 0.96 & 1.85 & 3.10 & 1.94 \\
\rule{0pt}{1.1em} & & & & & & & & \\
\hspace*{0.8em}- \textit{Greens and Beans} & & & & & & & & \\
\hspace*{2.2em}8-11 & 0.00 & 0.00 & 0.00 & 0.00 & 0.00 & 0.00 & 0.00 & 0.00 \\
\hspace*{4.4em}Female & 0.00 & 0.00 & 0.00 & 0.00 & 0.00 & 0.00 & 0.00 & 0.00 \\
\hspace*{4.4em}Male & 0.00 & 0.00 & 0.00 & 0.00 & 0.00 & 0.00 & 0.00 & 0.00 \\
\hspace*{2.2em}12-17 & 0.00 & 0.00 & 0.00 & 0.00 & 0.00 & 0.00 & 0.00 & 0.00 \\
\hspace*{4.4em}Female & 0.00 & 0.00 & 0.00 & 0.00 & 0.00 & 0.00 & 0.00 & 0.00 \\
\hspace*{4.4em}Male & 0.00 & 0.00 & 0.00 & 0.00 & 0.00 & 0.00 & 0.00 & 0.00 \\
\rule{0pt}{1.1em} & & & & & & & & \\
\hspace*{0.8em}- \textit{Whole Grains} & & & & & & & & \\
\hspace*{2.2em}8-11 & 0.00 & 0.42 & 2.67 & 0.87 & 0.00 & 0.97 & 3.44 & 1.35 \\
\hspace*{4.4em}Female & 0.00 & 1.16 & 3.32 & 1.41 & 0.00 & 1.09 & 3.60 & 1.44 \\
\hspace*{4.4em}Male & 0.00 & 0.00 & 2.15 & 0.54 & 0.00 & 0.79 & 3.31 & 1.22 \\
\hspace*{2.2em}12-17 & 0.00 & 0.00 & 1.71 & 0.43 & 0.00 & 0.00 & 2.52 & 0.63 \\
\hspace*{4.4em}Female & 0.00 & 0.00 & 2.16 & 0.54 & 0.00 & 0.00 & 2.72 & 0.68 \\
\hspace*{4.4em}Male & 0.00 & 0.00 & 1.57 & 0.39 & 0.00 & 0.00 & 2.29 & 0.57 \\
\rule{0pt}{1.1em} & & & & & & & & \\
\hspace*{0.8em}- \textit{Dairy} & & & & & & & & \\
\hspace*{2.2em}8-11 & 4.55 & 7.78 & 10.00 & 7.53 & 3.88 & 7.10 & 10.00 & 7.02 \\
\hspace*{4.4em}Female & 4.59 & 8.02 & 10.00 & 7.66 & 3.81 & 6.89 & 10.00 & 6.90 \\
\hspace*{4.4em}Male & 4.50 & 7.58 & 10.00 & 7.42 & 4.07 & 7.21 & 10.00 & 7.12 \\
\hspace*{2.2em}12-17 & 2.35 & 5.47 & 8.87 & 5.54 & 2.63 & 6.08 & 9.50 & 6.07 \\
\hspace*{4.4em}Female & 2.16 & 4.99 & 8.69 & 5.21 & 2.38 & 5.93 & 9.23 & 5.87 \\
\hspace*{4.4em}Male & 2.49 & 5.77 & 9.22 & 5.82 & 2.93 & 6.12 & 9.82 & 6.25 \\
\rule{0pt}{1.1em} & & & & & & & & \\
\hspace*{0.8em}- \textit{Total Protein Foods} & & & & & & & & \\
\hspace*{2.2em}8-11 & 2.33 & 4.01 & 5.00 & 3.84 & 2.28 & 4.10 & 5.00 & 3.87 \\
\hspace*{4.4em}Female & 2.29 & 3.91 & 5.00 & 3.78 & 2.53 & 4.12 & 5.00 & 3.94 \\
\hspace*{4.4em}Male & 2.57 & 4.02 & 5.00 & 3.90 & 2.05 & 3.98 & 5.00 & 3.75 \\
\hspace*{2.2em}12-17 & 2.45 & 4.15 & 5.00 & 3.94 & 2.50 & 4.23 & 5.00 & 3.99 \\
\hspace*{4.4em}Female & 2.46 & 3.97 & 5.00 & 3.85 & 2.44 & 4.11 & 5.00 & 3.91 \\
\hspace*{4.4em}Male & 2.35 & 4.32 & 5.00 & 4.00 & 2.64 & 4.43 & 5.00 & 4.13 \\
\rule{0pt}{1.1em} & & & & & & & & \\
\hspace*{0.8em}- \textit{Seafood and Plant Proteins} & & & & & & & & \\
\hspace*{2.2em}8-11 & 0.00 & 0.01 & 3.01 & 0.76 & 0.00 & 0.06 & 2.45 & 0.64 \\
\hspace*{4.4em}Female & 0.00 & 0.01 & 3.25 & 0.82 & 0.00 & 0.00 & 2.13 & 0.53 \\
\hspace*{4.4em}Male & 0.00 & 0.00 & 2.89 & 0.72 & 0.00 & 0.12 & 2.65 & 0.72 \\
\hspace*{2.2em}12-17 & 0.00 & 0.00 & 1.88 & 0.47 & 0.00 & 0.00 & 2.11 & 0.53 \\
\hspace*{4.4em}Female & 0.00 & 0.00 & 1.89 & 0.47 & 0.00 & 0.01 & 1.79 & 0.45 \\
\hspace*{4.4em}Male & 0.00 & 0.00 & 1.86 & 0.46 & 0.00 & 0.00 & 2.52 & 0.63 \\
\rule{0pt}{1.1em} & & & & & & & & \\
\hspace*{0.8em}- \textit{Fatty Acids} & & & & & & & & \\
\hspace*{2.2em}8-11 & 0.72 & 3.25 & 6.15 & 3.34 & 0.90 & 3.06 & 6.67 & 3.42 \\
\hspace*{4.4em}Female & 0.55 & 3.54 & 6.29 & 3.48 & 1.02 & 3.18 & 7.39 & 3.69 \\
\hspace*{4.4em}Male & 0.82 & 3.05 & 5.99 & 3.23 & 0.80 & 3.02 & 6.08 & 3.23 \\
\hspace*{2.2em}12-17 & 1.45 & 3.70 & 7.22 & 4.01 & 1.09 & 3.82 & 7.84 & 4.14 \\
\hspace*{4.4em}Female & 1.47 & 3.75 & 7.27 & 4.06 & 1.51 & 4.14 & 7.88 & 4.42 \\
\hspace*{4.4em}Male & 1.42 & 3.62 & 6.96 & 3.90 & 0.83 & 3.31 & 7.61 & 3.76 \\
\rule{0pt}{1.1em} & & & & & & & & \\
\hspace*{0.8em}\textbf{Moderation} & & & & & & & & \\
\rule{0pt}{0.35em} & & & & & & & & \\
\hspace*{0.8em}- \textit{Refined Grains} & & & & & & & & \\
\hspace*{2.2em}8-11 & 0.51 & 4.18 & 8.27 & 4.28 & 1.25 & 4.42 & 7.46 & 4.38 \\
\hspace*{4.4em}Female & 1.33 & 4.97 & 8.94 & 5.05 & 1.27 & 4.36 & 7.60 & 4.40 \\
\hspace*{4.4em}Male & 0.02 & 3.66 & 7.60 & 3.74 & 1.24 & 4.49 & 7.28 & 4.38 \\
\hspace*{2.2em}12-17 & 0.83 & 4.75 & 8.36 & 4.67 & 0.71 & 4.82 & 8.30 & 4.67 \\
\hspace*{4.4em}Female & 1.09 & 5.05 & 8.07 & 4.81 & 1.25 & 4.88 & 8.50 & 4.88 \\
\hspace*{4.4em}Male & 0.58 & 4.47 & 8.49 & 4.50 & 0.03 & 4.78 & 8.09 & 4.42 \\
\rule{0pt}{1.1em} & & & & & & & & \\
\hspace*{0.8em}- \textit{Sodium} & & & & & & & & \\
\hspace*{2.2em}8-11 & 2.61 & 5.76 & 8.27 & 5.60 & 2.00 & 5.21 & 8.04 & 5.12 \\
\hspace*{4.4em}Female & 2.81 & 6.14 & 8.39 & 5.87 & 2.12 & 5.29 & 7.73 & 5.11 \\
\hspace*{4.4em}Male & 2.07 & 5.27 & 8.21 & 5.21 & 2.00 & 5.17 & 8.16 & 5.13 \\
\hspace*{2.2em}12-17 & 2.28 & 5.69 & 9.00 & 5.67 & 1.82 & 5.20 & 8.30 & 5.13 \\
\hspace*{4.4em}Female & 2.61 & 5.98 & 8.92 & 5.88 & 1.63 & 5.02 & 8.30 & 4.99 \\
\hspace*{4.4em}Male & 2.19 & 5.56 & 9.02 & 5.58 & 1.95 & 5.33 & 8.39 & 5.25 \\
\rule{0pt}{1.1em} & & & & & & & & \\
\hspace*{0.8em}- \textit{Added Sugars} & & & & & & & & \\
\hspace*{2.2em}8-11 & 2.60 & 5.15 & 8.08 & 5.24 & 1.65 & 5.01 & 7.42 & 4.78 \\
\hspace*{4.4em}Female & 2.61 & 5.37 & 7.72 & 5.27 & 1.99 & 5.34 & 7.85 & 5.13 \\
\hspace*{4.4em}Male & 2.60 & 4.95 & 8.24 & 5.19 & 1.33 & 4.79 & 7.23 & 4.54 \\
\hspace*{2.2em}12-17 & 0.26 & 4.05 & 6.91 & 3.81 & 0.30 & 4.19 & 7.16 & 3.96 \\
\hspace*{4.4em}Female & 0.00 & 3.77 & 6.95 & 3.62 & 0.27 & 4.23 & 7.55 & 4.07 \\
\hspace*{4.4em}Male & 0.46 & 4.11 & 6.86 & 3.88 & 0.33 & 3.97 & 6.98 & 3.82 \\
\rule{0pt}{1.1em} & & & & & & & & \\
\hspace*{0.8em}- \textit{Saturated Fats} & & & & & & & & \\
\hspace*{2.2em}8-11 & 2.90 & 5.40 & 8.02 & 5.43 & 2.95 & 5.63 & 8.52 & 5.68 \\
\hspace*{4.4em}Female & 2.69 & 5.93 & 8.33 & 5.72 & 2.63 & 5.53 & 8.16 & 5.47 \\
\hspace*{4.4em}Male & 3.07 & 5.26 & 7.41 & 5.25 & 3.16 & 5.75 & 8.96 & 5.90 \\
\hspace*{2.2em}12-17 & 3.14 & 6.27 & 9.34 & 6.26 & 3.10 & 6.38 & 8.91 & 6.19 \\
\hspace*{4.4em}Female & 3.02 & 6.13 & 9.34 & 6.16 & 3.14 & 6.21 & 9.13 & 6.17 \\
\hspace*{4.4em}Male & 3.41 & 6.43 & 9.34 & 6.40 & 3.10 & 6.45 & 8.74 & 6.18 \\
\rule{0pt}{1.1em} & & & & & & & & \\
\hspace*{0.8em}\textbf{Total Score} & & & & & & & & \\
\rule{0pt}{0.35em} & & & & & & & & \\
\hspace*{0.8em}- \textit{Total} & & & & & & & & \\
\hspace*{2.2em}8-11 & 37.12 & 44.38 & 51.72 & 44.40 & 36.60 & 44.22 & 52.71 & 44.44 \\
\hspace*{4.4em}Female & 37.23 & 45.76 & 54.53 & 45.82 & 36.17 & 43.64 & 53.22 & 44.17 \\
\hspace*{4.4em}Male & 37.05 & 43.78 & 49.63 & 43.56 & 37.12 & 44.62 & 52.26 & 44.66 \\
\hspace*{2.2em}12-17 & 35.67 & 42.21 & 50.86 & 42.74 & 35.56 & 43.45 & 50.65 & 43.28 \\
\hspace*{4.4em}Female & 36.14 & 41.52 & 49.79 & 42.24 & 35.77 & 44.02 & 51.16 & 43.75 \\
\hspace*{4.4em}Male & 35.34 & 42.93 & 51.15 & 43.09 & 34.80 & 42.88 & 49.79 & 42.59 \\
\end{longtable}
\normalsize

\subsection{Choosing a robust test statistic}
For each individual and intersection null hypothesis, we grid search over $15$ candidate test statistics on the pilot sample and choose the statistic that yields the highest sensitivity value. We consider the class of $m$-statistics proposed by \cite{rosenbaum2007sensitivity, rosenbaum2013impact} and set the parameter $\lambda = 1/2$ for each candidate statistic. Observations are scaled by the $\lambda$-quantile of the absolute pair differences; for  $\lambda = 1/2$, this corresponds to the median of all paired absolute differences. We vary the parameter ``outer'' between $1.5$ and $3.5$, while varying the parameter ``inner'' between $0$ and $3$. Inner trimming is done to increase design sensitivity and outer trimming is used to enhance resistance to outliers. Then, absolute pair differences smaller than inner*median receive weight $0$, absolute pair differences larger than median*outer receive weight $1$, and between inner*median and outer*median weights increase linearly from $0$ to $1$. Note that setting (inner, outer) $= (0, 2.5)$ amounts to a choice of $m$-statistic with Huber’s $\psi$-function, whereas (inner, outer) $= (0.5, 2.5)$ corresponds to Rosenbaum's ``InnerTrim'' method.

\subsection{Estimated linear combinations on pilot sample} \label{appendix:lambda-estimates}

Table~\ref{tab:estimated-lambdas} reports the estimated pilot sample linear
combinations used in the data analysis. Each column gives the estimated weights
$\widehat{\lambda}_{\mathrm{plan}}$ for the corresponding stratum, with
entries rounded to two decimal places. Outcomes assigned zero weight in a particular stratum were not
carried forward for testing in that stratum.

\begin{table}[htbp]
\centering
\caption{Estimated pilot sample linear combinations by age and sex stratum. Entries are components of $\widehat{\lambda}_{\mathrm{plan}}$.}
\label{tab:estimated-lambdas}
\footnotesize
\begin{tabular}{lccc}
\multicolumn{4}{c}{{Ages 8--11}}\\
\midrule
Outcome & Pooled & Female & Male \\
\midrule
Waist-to-Height       & 0.00 & 0.30 & 0.00 \\
Cotinine              & 1.00 & 0.45 & 1.00 \\
Systolic BP           & 0.00 & 0.09 & 0.00 \\
Non-HDL Cholesterol   & 0.00 & 0.00 & 0.00 \\
Diet                  & 0.00 & 0.16 & 0.00 \\
\midrule \addlinespace[1.5em]
\multicolumn{4}{c}{{Ages 12--17}}\\
\midrule
Outcome & Pooled & Female & Male \\
\midrule
Waist-to-Height       & 0.00 & 0.00 & 0.05 \\
Moderate LTPA         & 0.00 & 0.06 & 0.00 \\
Vigorous LTPA         & 0.00 & 0.36 & 0.00 \\
Cotinine              & 1.00 & 0.28 & 0.70 \\
Systolic BP           & 0.00 & 0.03 & 0.00 \\
HbA1c                 & 0.00 & 0.11 & 0.06 \\
Non-HDL Cholesterol   & 0.00 & 0.08 & 0.00 \\
eGFR                  & 0.00 & 0.00 & 0.01 \\
Diet                  & 0.00 & 0.08 & 0.18 \\
\bottomrule
\vspace{2em}
\end{tabular}
\end{table}

For example, cotinine was the only variable selected on the pilot sample for our global null tests of young boys,  while almost $70\%$ of the normalized mass of the linear combination vector was designated to cotinine in the analysis of adolescent boys. Among young boys, cotinine was the only variable selected on the pilot sample for our tests of this age/sex-specific global null. Cotinine similarly took around $45\%$ of the normalized mass of the contrast vector in the young female stratum, with the rest predominantly allocated to waist-to-height ratio and diet. Likewise, in adolescent boys, almost $70\%$ of the normalized mass was designated to cotinine and almost $65\%$ was assigned to cotinine and vigorous physical activity in adolescent girls. We remark that eGFR was forced to zero in this latter population since we pre-specified testing the alternative of decreased filtration rate in poorer individuals, while the pilot sample reflected a slight increase; see Table~\ref{tab:edatbl_summary_outcomes}.

\subsection{Results for different levels of $\Gamma$}

We present results at various levels of control for potential unmeasured confounding across the remaining pages.

\begin{figure}[p]
    \centering
    \includegraphics[scale=.22]{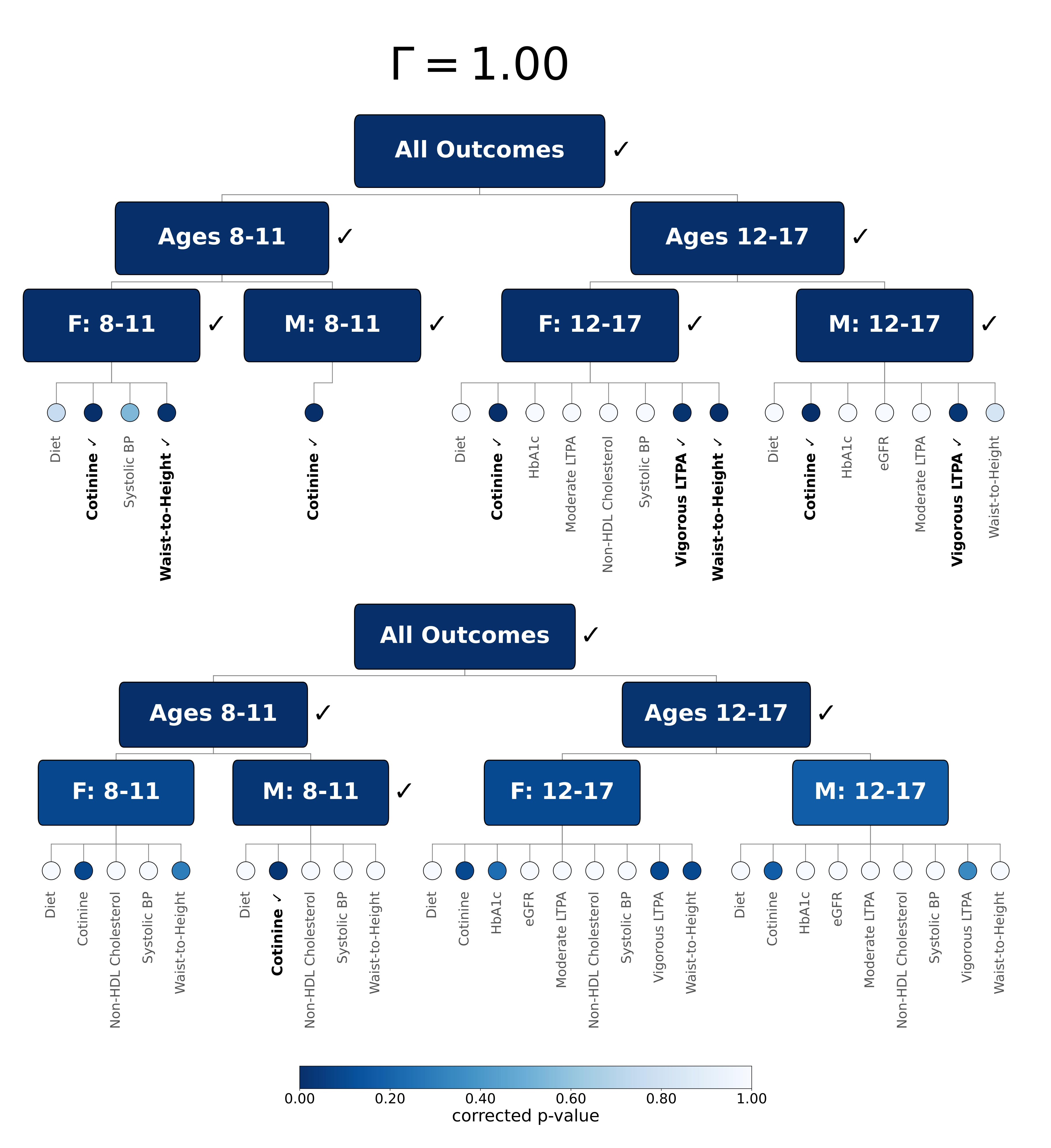}
    \caption{Rejected hypotheses and $p$-values for our method (top row) and the full sample comparator of \cite{cohen2020multivariate} (bottom row) at $\Gamma  = 1.00$.}
    \label{app-results-1.00}
\end{figure}
\newpage
\begin{figure}[p]
    \centering
    \includegraphics[scale=.22]{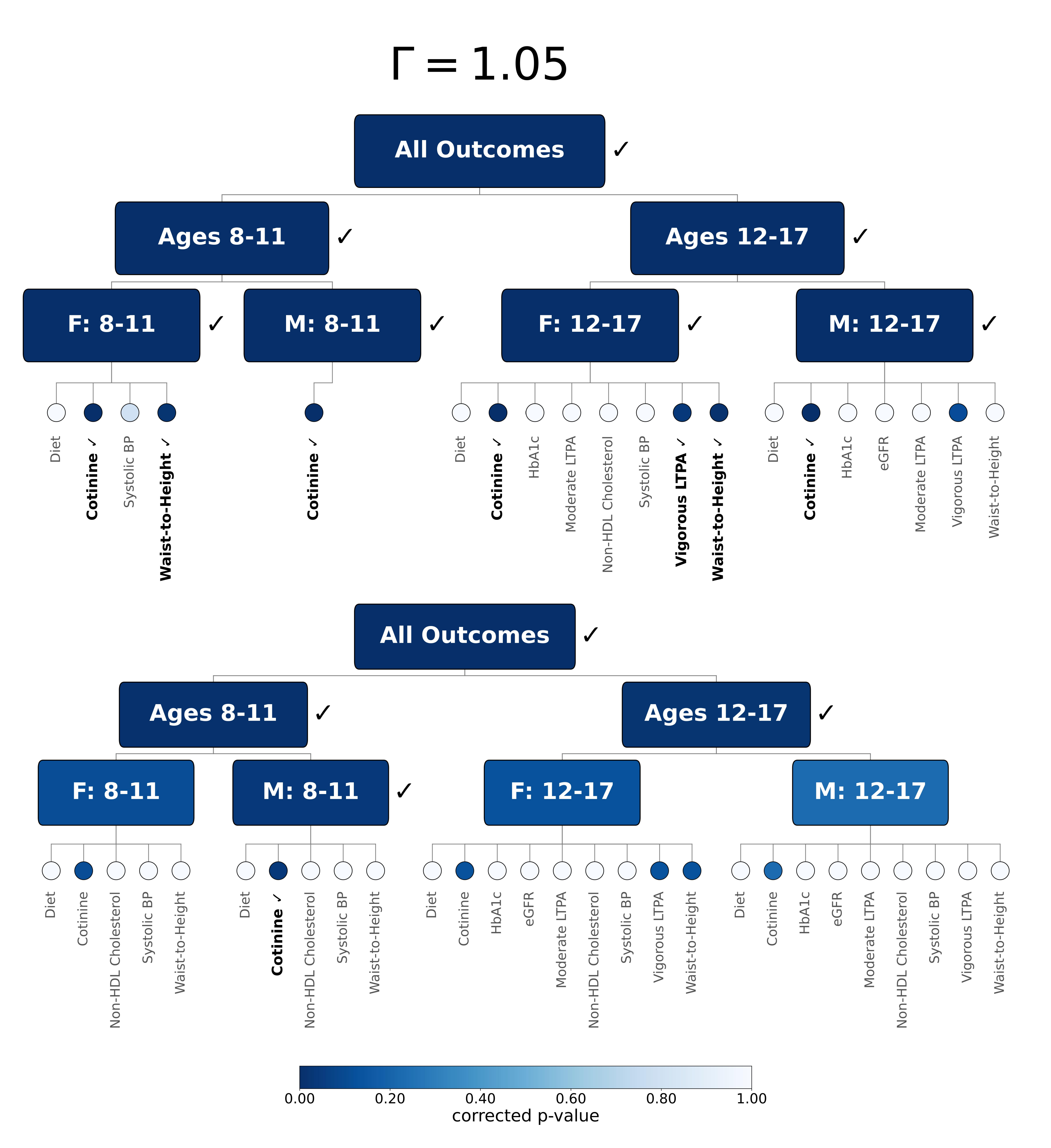}
    \caption{Rejected hypotheses and $p$-values for our method (top row) and the full sample comparator of \cite{cohen2020multivariate} (bottom row) at $\Gamma  = 1.05$.}
    \label{app-results-1.05}
\end{figure}
\newpage
\begin{figure}[p]
    \centering
    \includegraphics[scale=.22]{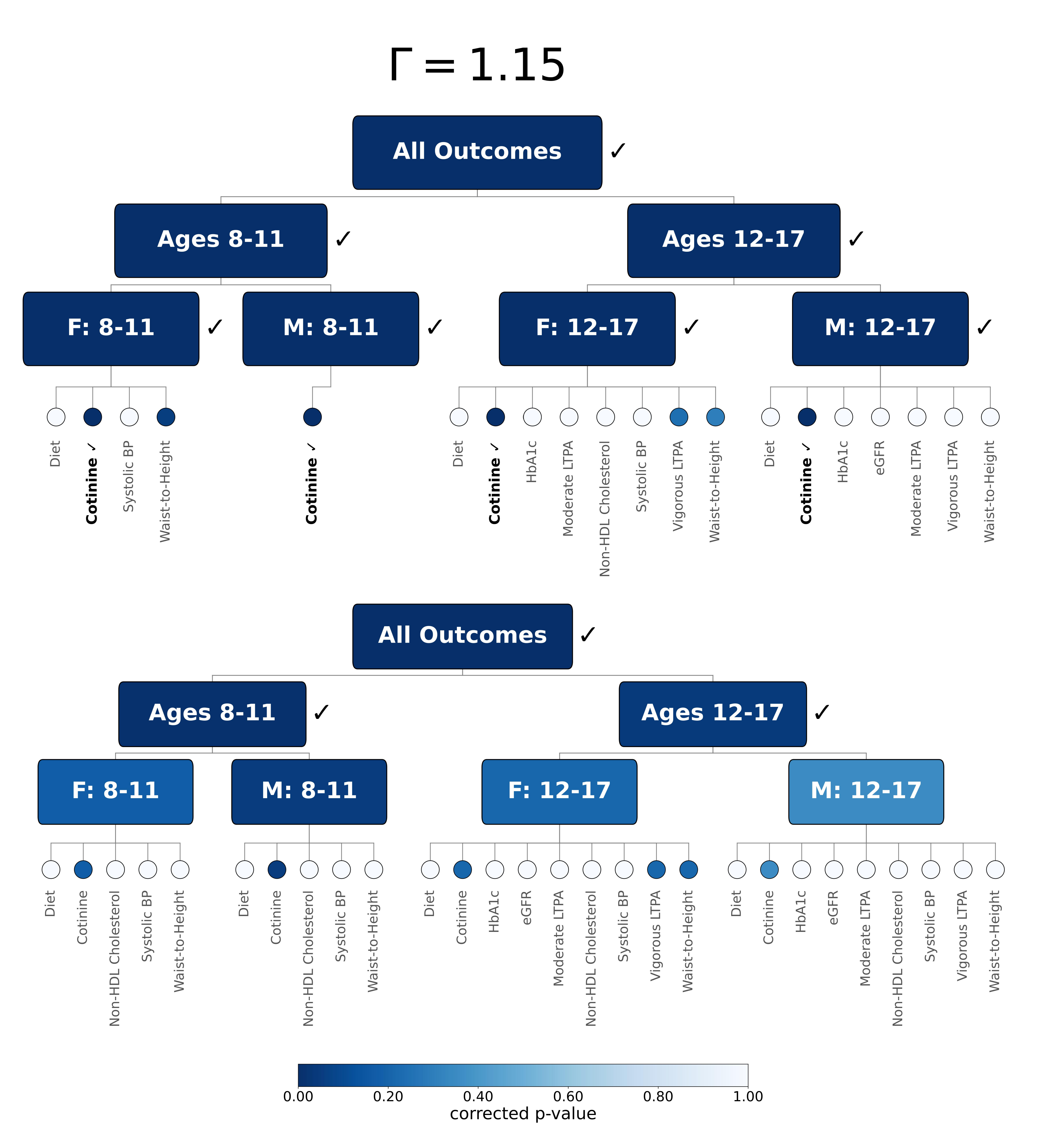}
    \caption{Rejected hypotheses and $p$-values for our method (top row) and the full sample comparator of \cite{cohen2020multivariate} (bottom row) at $\Gamma  = 1.15$.}
    \label{app-results-1.15}
\end{figure}
\newpage
\begin{figure}[p]
    \centering
    \includegraphics[scale=.22]{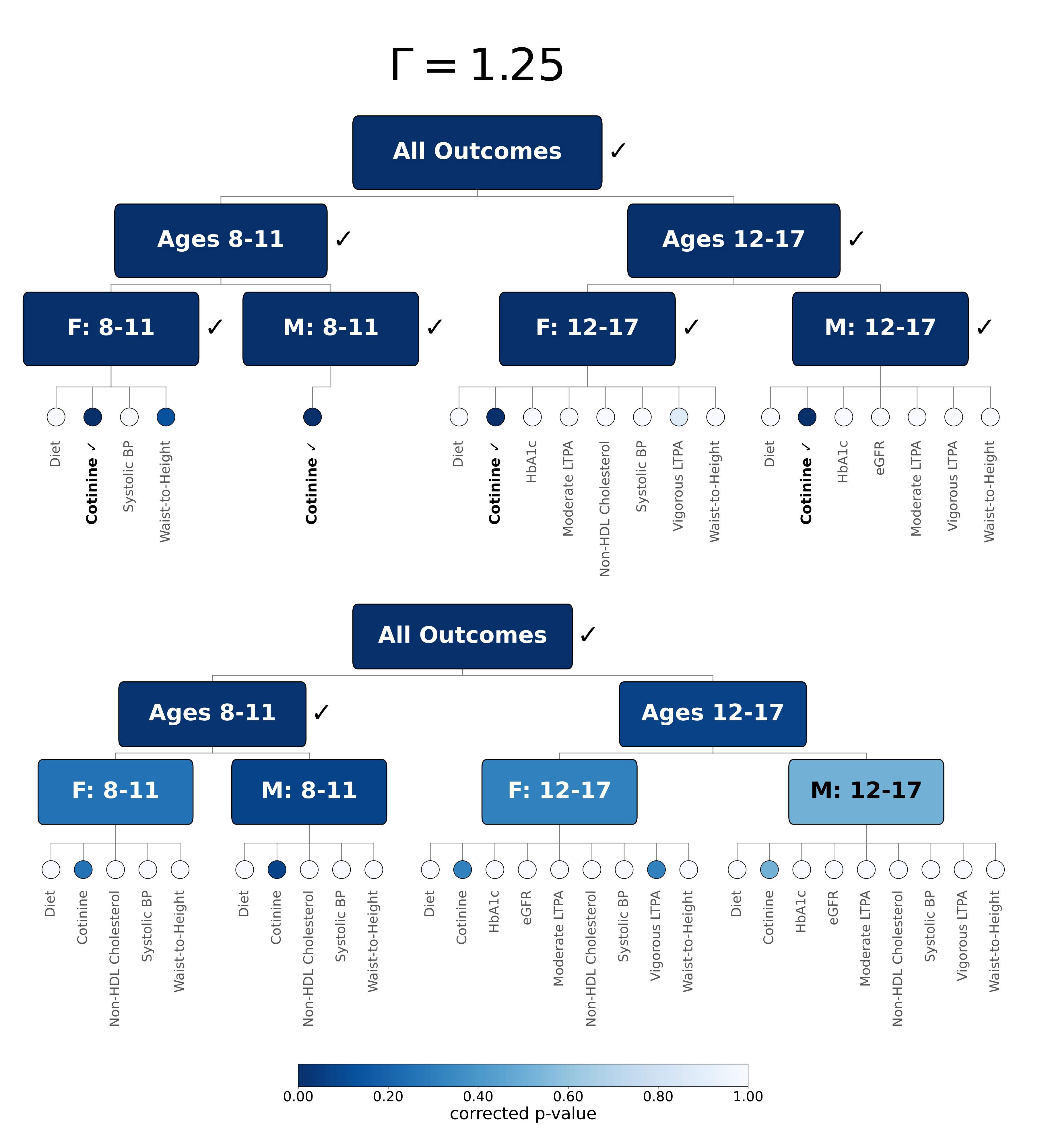}
    \caption{Rejected hypotheses and $p$-values for our method (top row) and the full sample comparator of \cite{cohen2020multivariate} (bottom row) at $\Gamma  = 1.25$.}
    \label{app-results-1.25}
\end{figure}
\newpage
\begin{figure}[p]
    \centering
    \includegraphics[scale=.22]{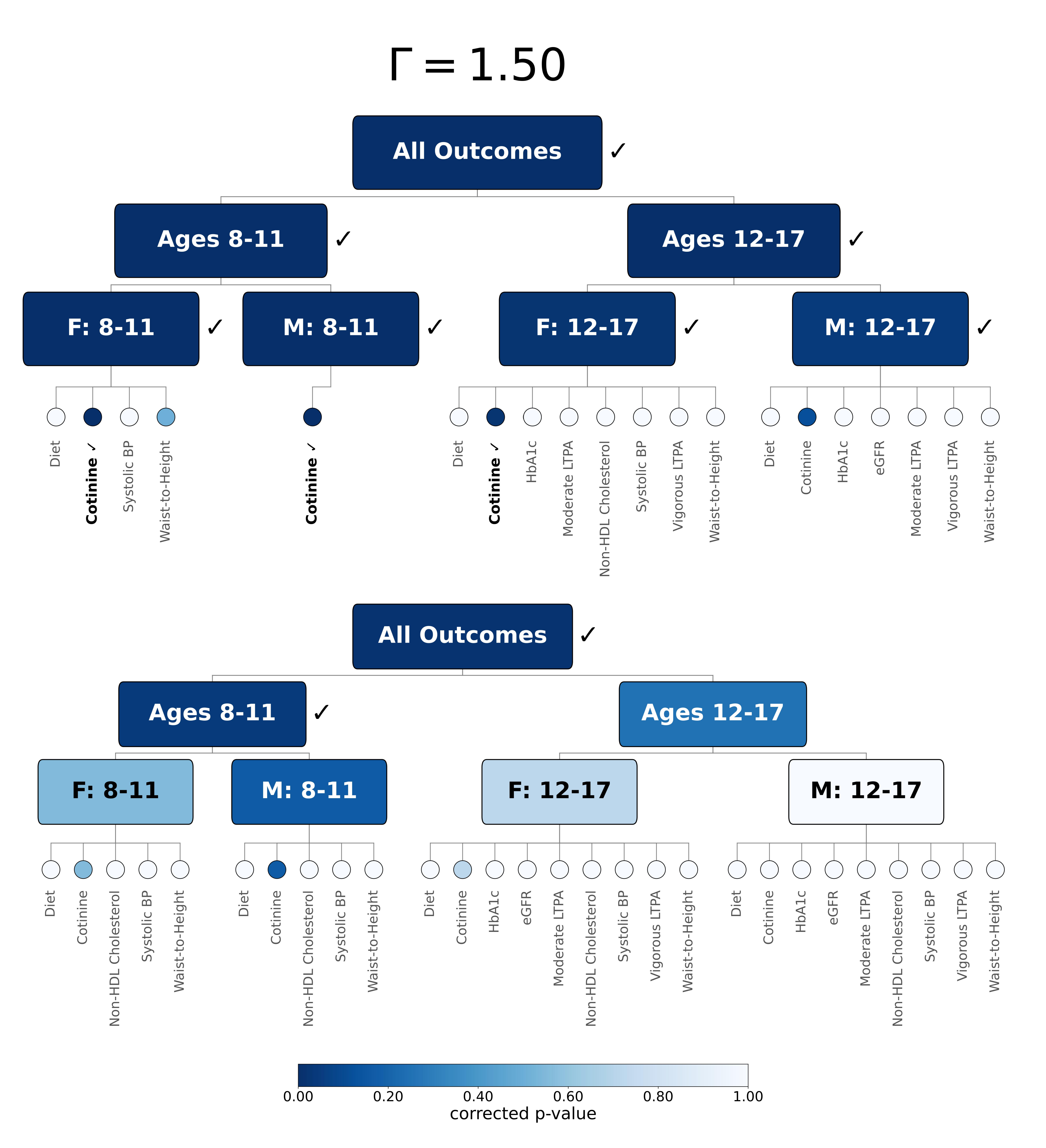}
    \caption{Rejected hypotheses and $p$-values for our method (top row) and the full sample comparator of \cite{cohen2020multivariate} (bottom row) at $\Gamma  = 1.50$.}
    \label{app-results-1.50}
\end{figure}
\newpage
\begin{figure}[p]
    \centering
    \includegraphics[scale=.22]{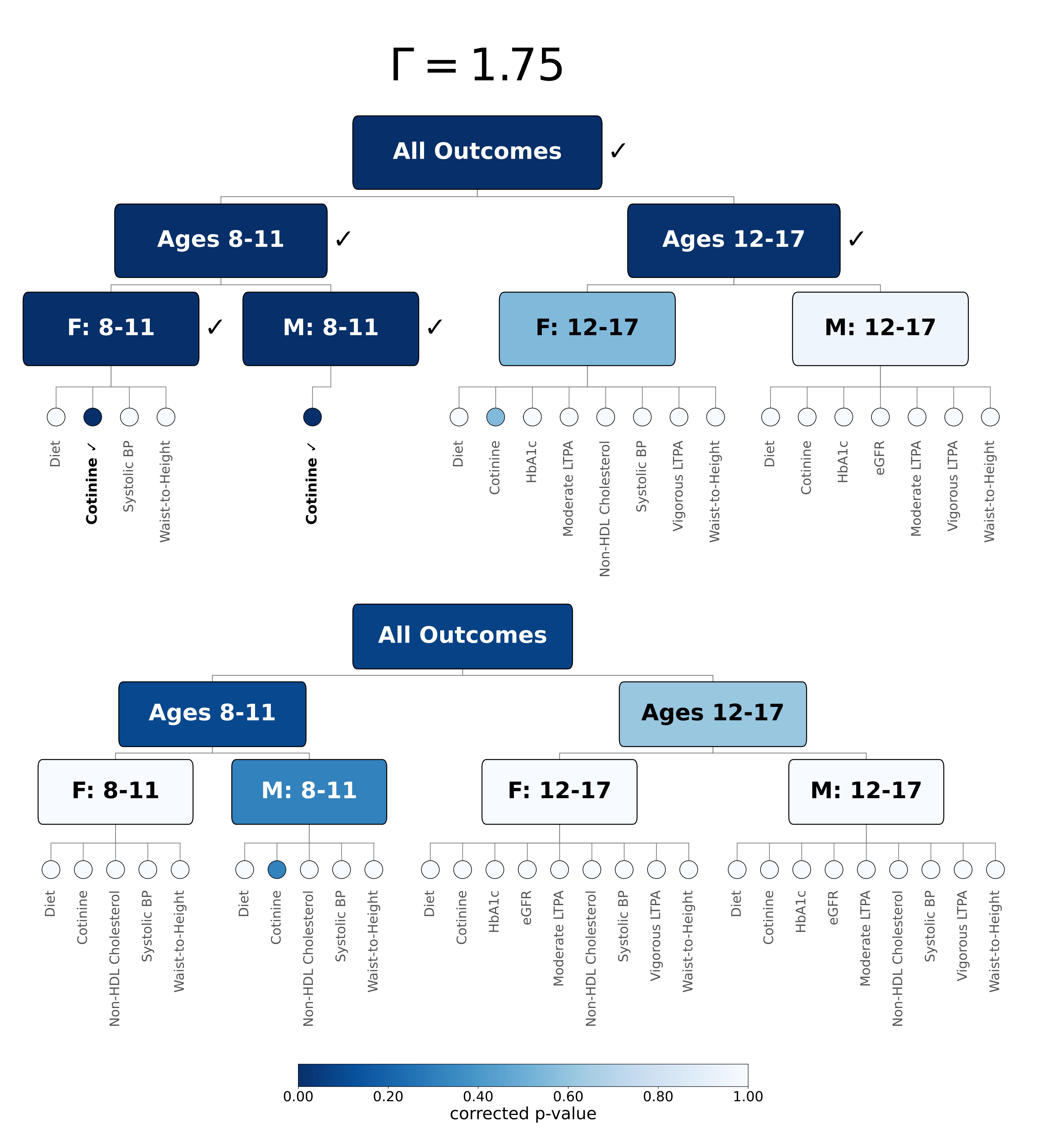}
    \caption{Rejected hypotheses and $p$-values for our method (top row) and the full sample comparator of \cite{cohen2020multivariate} (bottom row) at $\Gamma  = 1.75$.}
    \label{app-results-1.75}
\end{figure}
\newpage
\begin{figure}[p]
    \centering
    \includegraphics[scale=.22]{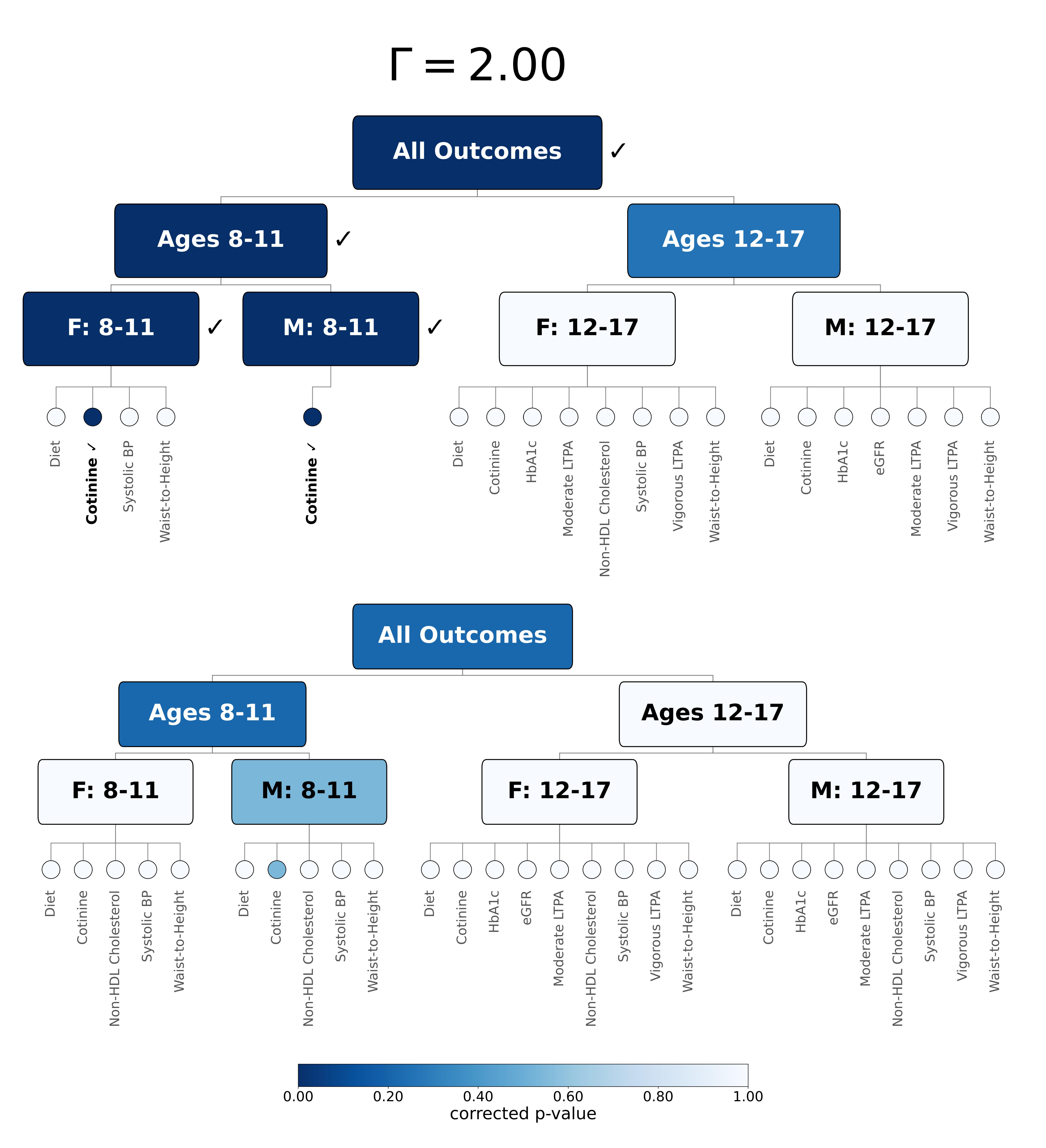}
    \caption{Rejected hypotheses and $p$-values for our method (top row) and the full sample comparator of \cite{cohen2020multivariate} (bottom row) at $\Gamma  = 2.00$.}
    \label{app-results-2.00}
\end{figure}
\newpage
\begin{figure}[p]
    \centering
    \includegraphics[scale=.22]{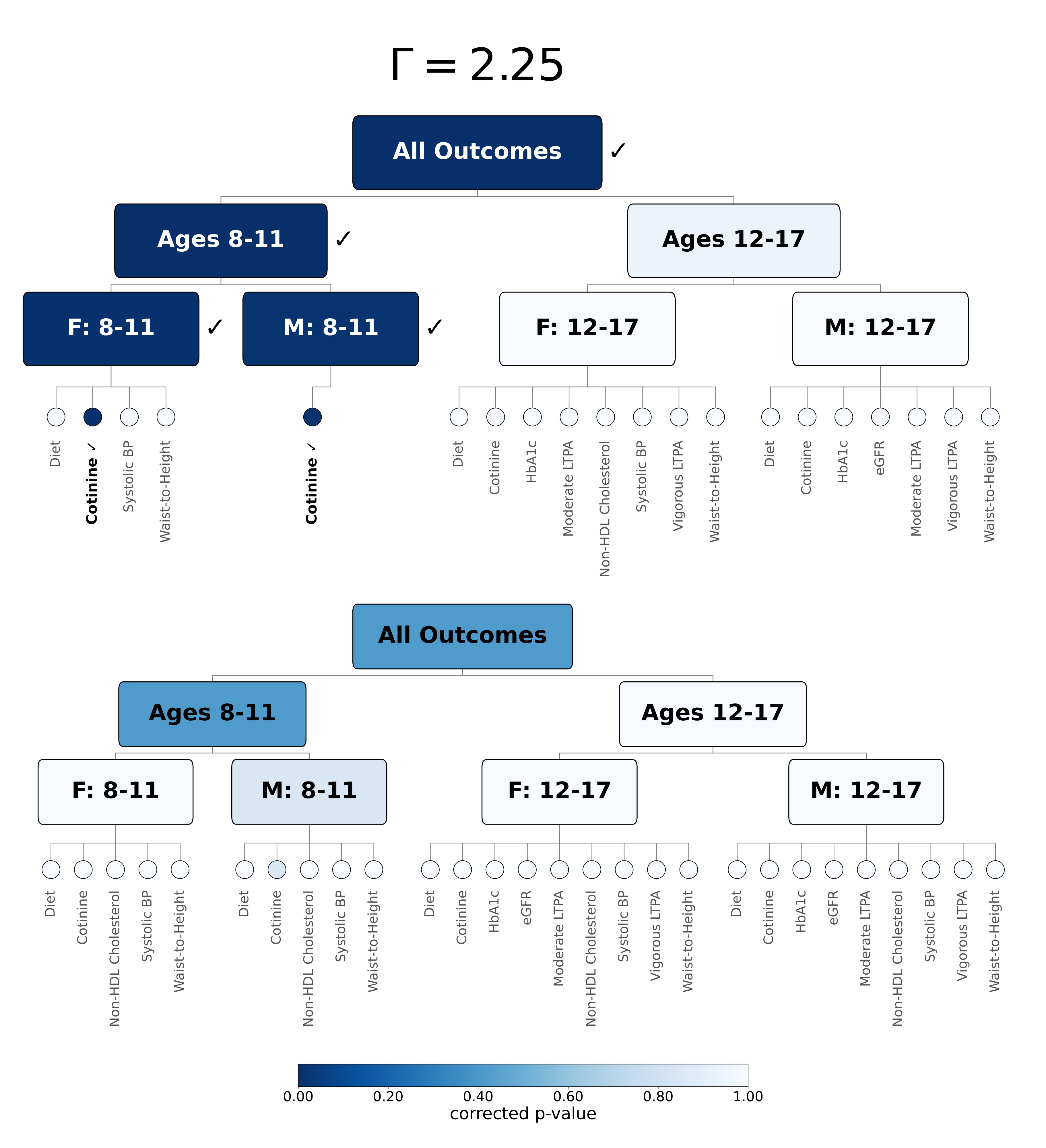}
    \caption{Rejected hypotheses and $p$-values for our method (top row) and the full sample comparator of \cite{cohen2020multivariate} (bottom row) at $\Gamma  = 2.25$.}
    \label{app-results-2.25}
\end{figure}
\newpage
\begin{figure}[p]
    \centering
    \includegraphics[scale=.22]{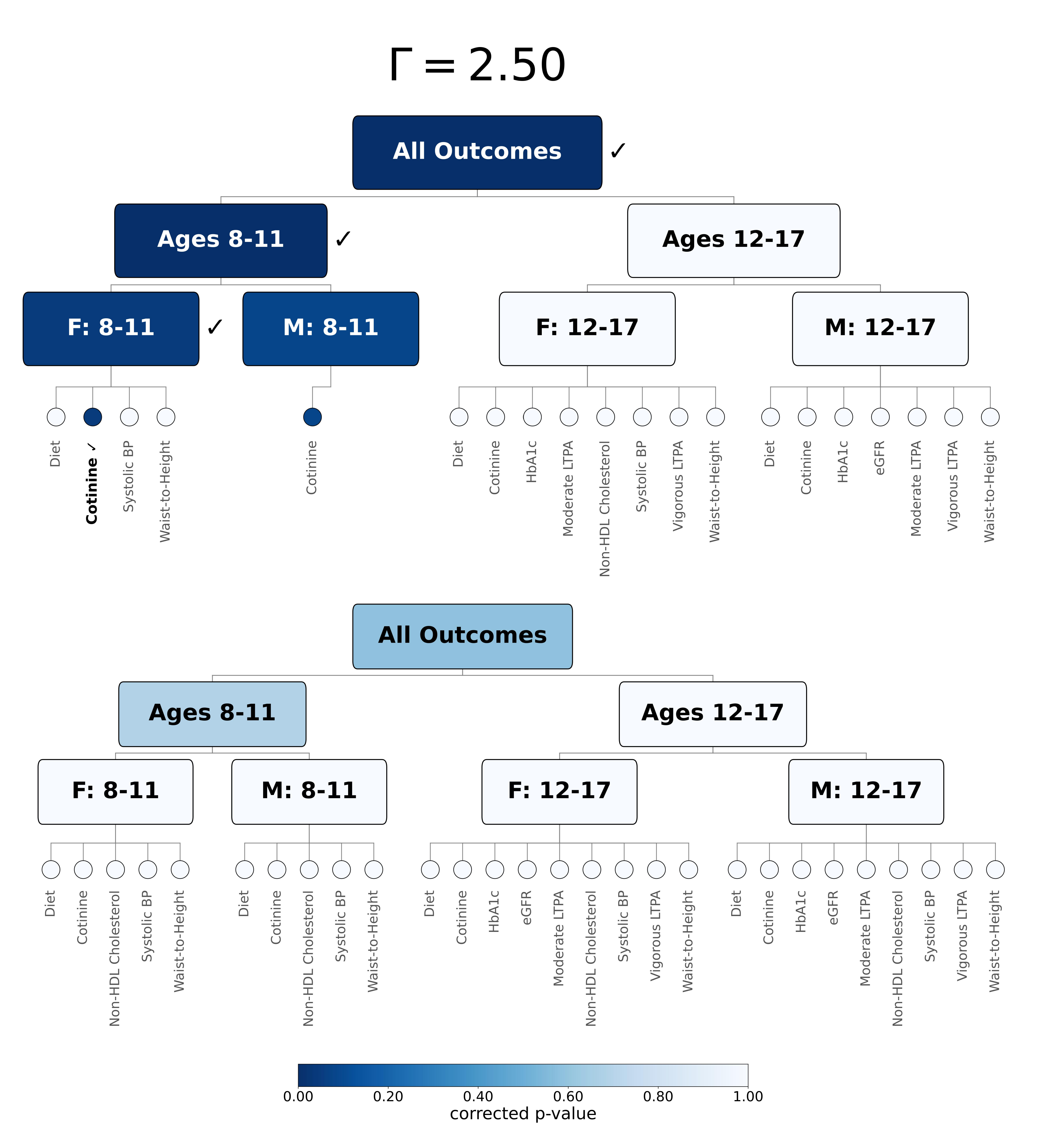}
    \caption{Rejected hypotheses and $p$-values for our method (top row) and the full sample comparator of \cite{cohen2020multivariate} (bottom row) at $\Gamma  = 2.50$.}
    \label{app-results-2.50}
\end{figure}

\clearpage

\renewcommand{\refname}{References}
\putbib

\end{bibunit}
\end{document}